\pdfoutput=1
\documentclass[a4paper,UKenglish,orivec,cleveref]{lipics-v2019}

\usepackage{microtype} \usepackage{booktabs}

\usepackage{amsmath,amssymb}
\def\qedhere{\qed}
\usepackage{enumerate}
\usepackage{mathdots}
\usepackage{mathrsfs}
\usepackage{style/mathpartir}
\usepackage{style/notation}

\usepackage{tikz}
\usepackage{style/tikz-tricks}

\usepackage{hyperref}
\crefname{rule}{Rule}{Rules}
\crefname{problem}{problem}{problems}

\title{Decidable Inductive Invariants for Verification of Cryptographic Protocols with Unbounded Sessions}
\titlerunning{Decidable Inductive Invariants for Cryptographic Protocols Verification}

\author{Emanuele D'Osualdo}
  {Imperial College London, UK}
  {e.dosualdo@ic.ac.uk}
  {https://orcid.org/0000-0002-9179-5827}
  {EU Horizon 2020 Marie Curie Individual Fellowship.}

\author{Felix Stutz}
  {MPI-SWS, Kaiserslautern, Germany \and
   Saarland University, Saarbrücken, Germany}
  {fstutz@mpi-sws.org}
  {https://orcid.org/0000-0003-3638-4096}
  {supported by Imperial College London and International Max Planck Research School for Computer Science.}

\authorrunning{E. D'Osualdo and F. Stutz}
\Copyright{Emanuele D'Osualdo and Felix Stutz}



\supplement{Tool available at \url{https://doi.org/10.5281/zenodo.3950846}}

\acknowledgements{We would like to thank Alwen Tiu, Roland Meyer and V\'{e}ronique Cortier for the useful feedback.}

\EventEditors{Igor Konnov and Laura Kov\'{a}cs}
\EventNoEds{2}
\EventLongTitle{31st International Conference on Concurrency Theory (CONCUR 2020)}
\EventShortTitle{CONCUR 2020}
\EventAcronym{CONCUR}
\EventYear{2020}
\EventDate{September 1--4, 2020}
\EventLocation{Vienna, Austria}
\EventLogo{}
\SeriesVolume{2017}
\ArticleNo{8}

\nolinenumbers \hideLIPIcs

\let\ifdraft\iffalse

\newcommand{\adjustfigure}[1][\small]{\centering#1\belowdisplayskip=0pt\belowdisplayshortskip=0pt\abovedisplayskip=0pt\abovedisplayshortskip=0pt}

\begin{document}
\maketitle

\begin{abstract}
We develop a theory of decidable inductive invariants for an infinite-state
variant of the Applied \picalc,
with applications to automatic verification of stateful cryptographic protocols
with unbounded sessions/nonces.
Since the problem is undecidable in general,
we introduce \emph{depth-bounded protocols},
a strict generalisation of a class from the literature,
for which our decidable analysis is sound and complete.
Our core contribution is a procedure to check that an invariant is inductive,
which implies that every reachable configuration satisfies it.
Our invariants can capture security properties like secrecy,
can be inferred automatically,
and represent an independently checkable certificate of correctness.
We provide a prototype implementation
and we report on its performance on some textbook examples.
\end{abstract}

\section{Introduction}
\label{sec:intro}

Security protocols implement secure communication over insecure channels,
by using cryptography.
Despite underpinning virtually every communication over the internet,
new flaws that compromise security are routinely discovered
in deployed protocols.
Automatic protocol verification is highly desirable,
but also very challenging:
the space of possible attacks is infinite.
Indeed, even under the assumption of perfect cryptography,
security properties are undecidable~\cite{durgin99fmsp}.
The most problematic feature for decidability is the necessity
of considering unboundedly many fresh random numbers, called \emph{nonces},
to distinguish between various sessions of the protocol.
There has been a proliferation of verification tools~\cite{DEEPSEC,AVISPA,SPEC,AKISS,PROVERIF,TAMARIN}
which can be categorised according to the way the undecidability issue is resolved.
A~first approach is to only consider a bounded number of sessions,
possibly missing attacks.
A~second is to over-approximate the protocol's behaviour
by representing nonces with less precision,
possibly reporting spurious attacks.
A~third is to implement semi-algorithms, accepting that the tools might never terminate on some protocols.

In this paper, we devise a sound and complete analysis,
i.e.~one that always terminates
with a correct answer, without need for approximations.
We obtain this by developing decision procedures for
proving invariants of a rich sub-class of protocols with unbounded sessions/nonces.
An invariant is any property that holds for every reachable configuration.
We introduce \emph{depth-bounded protocols},
a strict generalisation of the class of~\cite{DOsualdoOT17},
and prove that a class of invariants, called downward-closed,
can be effectively represented using expressions that we call \emph{limits}.
Our core technical results are a decision procedure for limit inclusion
and an algorithm called $\posthat$ that computes, from a limit~$L$,
a finite union of limits that
represent the (infinite) set of configurations reached in one step from~$L$.
By using these two components, we obtain an algorithm to check
if a limit is inductive, i.e.~$\posthat(L) \subseteq L$.
An inductive limit that contains the initial configuration is guaranteed to
be an invariant for the protocol.
We show how to use this to prove a number of properties including
depth-boundedness itself (a semantic property), secrecy,
and control-state reachability.

We define depth-bounded protocols as a subclass of a variant of
the Applied \picalc~\cite{RyanS11},
with support for user-defined cryptographic primitives,
secure and public channels, stateful principals,
a Dolev-Yao-style intruder~\cite{DolevY83} supporting modelling of dishonest participants and leaks of old keys.
In particular, our results apply to any set of cryptographic primitives
that satisfy some simple axioms; examples include
(a)symmetric encryption, blind signatures, hashes, XOR.
To gain intuition about depth-boundedness,
consider the set of messages
$ \Gamma_n = \set{ \enc{k_1}{k_2}, \enc{k_2}{k_3}, \dots \enc{k_{n-1}}{k_n} }$
which ``chains'' key $k_1$ to $k_2$, $k_2$ to $k_3$ and so on,
obtaining an encryption chain of length~$n$.
A depth-bounded protocol cannot produce, or be tricked to produce,
such chains of unbounded length.
Note that, when computing depth, we only consider chains that are essential:
the set $\Gamma_n \union \set{k_n}$ for example has depth~$1$ (for any~$n$)
because it is equivalent to the set $\set{k_1,\dots,k_n}$.
When there is a bound~$d$ on the depth of reachable configurations,
we say that the protocol is depth-bounded.
We built a proof-of-concept prototype tool to evaluate the approach,
showing that many textbook protocols fall into the depth-bounded class.

More precisely, bounding depth alone is not enough
to obtain decidability~\cite{DOsualdoOT17}:
one needs to bound the size of messages too.
For type-compliant protocols~\cite{ArapinisD14,ChretienCD14}
message size can be bounded without excluding any security violation.
More generally, for typical protocols (including all our benchmarks), our inductive invariants can be computed on the
size-bounded model, and then generalised to invariants for the unrestricted version of the protocol.

Our approach has a number of notable properties.
First, once a suitable inductive invariant has been found,
    it can be provided as a certificate of correctness that
    can be independently checked.
Second, the search of a suitable invariant can be performed
    both automatically (with a trade-off between precision and performance)
    or interactively.
Third, supporting unbounded nonces makes it possible to reason about
    properties like susceptibility to known-plaintext attacks
    (\cref{sec:invariants}).
Finally, even coarse invariants inferred with our method
    can be used to prune the search space of other model checking procedures.

\subparagraph{Related work}
The pure \picalc{} version of depth-boundedness was originally proposed
in~\cite{Meyer08} and developed in~\cite{HuchtingMM14,WiesZH10,ZuffereyWH12}.
Our work builds directly on~\cite{DOsualdoOT17},
which introduced depth-boundedness
for the special case of secrecy of protocols using symmetric encryption.
We generalise to a strictly more expressive class of primitives and properties,
a result that requires much more sophisticated techniques and yields more powerful algorithms.

Our theory of invariants is framed in terms of
ideal completions~\cite{FinkelG09},
which, to the best of our knowledge,
has not been instantiated to cryptographic protocols before.
Our decidability proofs introduce substantial new proof techniques
to deal with an active intruder while being parametric
on the cryptographic primitives.

Types~\cite{SPICA2,ChretienCD15,CortierGLM17}
can be used to capture and generalise
common safe usages of cryptographic primitives,
and reduce verification to constraints which can be solved efficiently.
We speculate that our domain of limits and the associated algorithms could be
used to define an expressive class of solvable constraints
that could be integrated in type systems.

In~\cite{ChretienCD15,Froeschle15} two classes of protocols
with unbounded nonces are shown
to enjoy decidable verification.
They consider a less general calculus
((a)symmetric encryption only, with atomic keys),
different properties
(only secrecy~\cite{Froeschle15}, trace equivalence~\cite{ChretienCD15}),
and restrict protocols using (similar) syntactic conditions,
obtaining classes that are orthogonal to ours.

ProVerif~\cite{PROVERIF} and
Tamarin~\cite{TAMARIN} are two mature tools
with support for a wide range of cryptographic primitives
and expressive properties, and handle unbounded sessions.
Both programs employ semi-algorithms and may diverge on verification tasks.
ProVerif is known to terminate
on so-called \emph{tagged protocols}~\cite{BlanchetP05}
which are incomparable to depth-bounded protocols.
Tamarin offers an interactive mode when a proof cannot be carried out automatically.
To the best of our knowledge, there is no characterisation for a
class of protocols on which Tamarin is guaranteed to terminate.

\subparagraph{Outline}
\cref{sec:prelim} introduces the formal model and depth-bounded protocols.
\Cref{sec:ideal} presents our main theoretical results.
In \cref{sec:discussion} we report on experiments with our tool and discuss limitations.
All omitted proofs can be found in the Appendix.
 \section{Formal Model}
\label{sec:prelim}

\label{sec:intruder-models}

We introduce a variant of the Applied \picalc\ as our formal model of protocols.
Following the Dolev-Yao intruder model,
we treat cryptographic primitives algebraically.
Assume an enumerable set of \emph{names} $a,b,\dots\in\Names$.
A signature $\Sig$ of \emph{constructors},
is a finite set of symbols~$\constr{f}$
with their arity $\arity(\constr{f}) \in \Nat$.
The set of messages over $\Sig$
is the smallest set~$ \Msg^\Sig $ which contains all names,
and is closed under application of constructors.
The domain of finite sets of messages is
$ \Know^\Sig \is \finpow(\Msg^\Sig) $.
We define
  $\size(\constr{f}(M_1,\dots,M_n)) \is
    1 + \max\set{\size(M_i)|1\leq i\leq n}$,
  $\size(a) \is 1$, and
$\names(a) \is \set{a}$,
  $\names(\constr{f}(M_1,\dots,M_n)) \is
      \Union_{i=1}^n\names(M_i) $.
Given $X \subseteq \Names$ and $s\in\Nat$,
we define
$ \Msg^{\Sig,X}_s \is
    \set{M\in \Msg^\Sig | \names(M)\subseteq X, \size(M) \leq s} $.
As is standard,
$\Gamma,\Gamma'$ and $\Gamma,M$ stand for
$\Gamma \union \Gamma'$ and $\Gamma \union \set{M}$ respectively.

A \emph{substitution} is a finite partial function $\theta\from \Names \pto \Msg^\Sig$;
we write $\theta = \subst{x_1 -> M_1,,x_n -> M_n}$,
abbreviated with $ \subst{\vec{x} -> \vec{M}} $,
for the substitution with $\theta(x_i) = M_i$ for all~$1\leq i\leq n$.
We write~$M\theta$ for the application
of substitution~$\theta$ to the message~$M$,
and extend the notation to sets of messages
$\Gamma\theta \is \set{M\theta | M \in \Gamma}$.
A substitution $\theta$ is a \emph{renaming of $X \subseteq \Names$}
if it is defined on~$X$, injective, and with $\theta(X) \subseteq \Names$.

\begin{definition}[Intruder model]
  \label{def:intruder-model}
  A \emph{derivability relation} for a signature $\Sig$,
  is a relation
    ${\deriv} \subseteq \Know^\Sig \times \Msg^\Sig$.
The pair $\Intruder = (\Sig, \deriv)$ is an
    \emph{(effective) intruder model} if
  ${\deriv}$ is a (decidable) derivability relation for $\Sig$, and
  for all $M,N\in\Msg^\Sig$, $\Gamma,\Gamma' \in \Know^\Sig$, $a\in\Names$:
  \begin{align}
    & M \deriv M
    \tag{Id}
    \label{axiom:id}
    \\&
     \Gamma \subseteq \Gamma' \land \Gamma \deriv M
       \implies \Gamma' \deriv M
    \tag{Mon}
    \label{axiom:mon}
    \\&
     \Gamma \deriv M \land \Gamma,M \deriv N
       \implies \Gamma \deriv N
    \tag{Cut}
    \label{axiom:cut}
    \\&
     M_1,\dots,M_n \deriv \constr{f}(M_1,\dots,M_n)
     \quad\text{for every }
       \constr{f}\in \Sigma
     \text{ with }
       \arity(\constr{f})=n
    \tag{Constr}
    \label{axiom:constr}
\\&
     \Gamma\theta \deriv M\theta \iff \Gamma \deriv M
      \quad\text{for any } \theta \text{ renaming of } \names(\Gamma)
    \tag{Alpha}
    \label{axiom:alpha}
    \\&
     \Gamma, a \deriv M \land a\not\in \names(\Gamma,M) \implies \Gamma \deriv M
    \tag{Relevancy}
    \label{axiom:relevancy}
\end{align}
The \emph{knowledge ordering} for $\Intruder$ is the relation
    ${\kleq} \subseteq \Know^\Sig \times \Know^\Sig$
  such that
  $
    \Gamma_1 \kleq \Gamma_2
  $ if and only if $
    \forall M\in\Msg^\Sig \st \Gamma_1 \deriv M \implies \Gamma_2\deriv M.
  $
  We write $\Gamma_1 \keq \Gamma_2$ if
  $\Gamma_1 \kleq \Gamma_2$ and $\Gamma_2 \kleq \Gamma_1$.
\end{definition}

The first three axioms deal exclusively with what it means to be a deduction relation:
  what is known can be derived \eqref{axiom:id};
  the more is known the more can be derived \eqref{axiom:mon};
what can be derived is known \eqref{axiom:cut}.
The \eqref{axiom:constr} axiom ensures the intruder is able to
construct arbitrary messages by composing known messages.
The \eqref{axiom:alpha} axiom
justifies \pre\alpha-renaming in our calculus.
The \eqref{axiom:relevancy} axiom allows us to only consider
boundedly many nonces maliciously injected by the intruder,
at each step of the protocol.

In the rest of the paper, unless otherwise specified,
we fix an arbitrary effective intruder model $\Intruder$
and omit the corresponding superscripts.

\begin{proposition}
\label{lm:locality}
  Given $\Gamma_1,\Gamma_2 \in \Know^\Sig$,
  $ \Gamma_1 \kleq \Gamma_2 $
    if and only if
  $ \forall M \in \Gamma_1\st \Gamma_2 \deriv M$.
  As a consequence, if ${\deriv}$ is decidable, so is ${\kleq}$.
\end{proposition}

Our framework uses the derivability relation as a black box
and does not rely on the way it is specified
(e.g.~with a rewriting system or a deduction system).
It is possible to formalise as an effective intruder model
cryptographic primitives such as XOR,
hashes and blind signatures.
We present here, for illustration, a model of (a)symmetric encryption,
and elaborate on extensions in \cref{app:more-primitives}.
We find it convenient to specify it with
a sequent calculus in the style of~\cite{tiu10lmcs}.
This is an alternative to more intuitive natural-deduction-style rules,
which has the key advantage of being \textsc{cut}-free, while admitting \textsc{cut}.
This simplifies considerably the proofs of properties
(e.g.~\cref{lm:symm-absorbing})
of the intruder model.

\begin{example}[Model of Encryption]
  Symmetric and asymmetric encryption
  can be modelled using the signature
  $
    \SigSymm = \set{
      \pair{\cdot}{\cdot} \,,\,
      \enc{\cdot}{\cdot}  \,,\,
      \aenc{\cdot}{\cdot} \,,\,
      \pub{\cdot}
    }
  $,
  where
    $\pair{M}{N}$ pairs messages $M$ and $N$,
    $\enc{M}{N}$ represents the message $M$
      encrypted with symmetric key $N$,
    $\aenc{M}{N}$ represents the message $M$
      encrypted with asymmetric key $N$, and
    $\pub{K}$ is the public key associated with the private key $K$.
  The intruder model for (a)symmetric encryption is the model
  $\IntrSymm = (\SigSymm, \deriv)$ where $\deriv$ is defined by
  the deduction rules in \cref{fig:symmetric-intruder}.
\end{example}

\begin{proposition}
\label{prop:sigsymm-axioms}
  The model $ \SigSymm $ is an effective intruder model.
\end{proposition}

\begin{figure}[tb]
  \adjustfigure[\footnotesize]\vspace*{-\abovedisplayskip}
  \begin{mathpar}
\inferrule*[right=Id]{\label{rule:id}
  M \in \Gamma
}{
  \Gamma \deriv M
}
\and
\inferrule*[right=Pub]{\label{rule:pub}
  \Gamma \deriv K
}{
  \Gamma \deriv \pub{K}
}
\and
\inferrule*[right=P\textsubscript{L}]{\label{rule:pl}
  \Gamma, \pair{M}{N}, M, N \deriv M'
}{
  \Gamma, \pair{M}{N} \deriv M'
}
\and
\inferrule*[right=P\textsubscript{R}]{\label{rule:pr}
  \Gamma \deriv M \and \Gamma \deriv N
}{
  \Gamma \deriv \pair{M}{N}
}
\and
\inferrule*[right=S\textsubscript{L}]{\label{rule:encl}
  \Gamma, \enc{M}{K} \deriv K
  \and
  \Gamma, \enc{M}{K}, M, K \deriv N
}{
  \Gamma, \enc{M}{K} \deriv N
}
\and
\inferrule*[right=S\textsubscript{R}]{\label{rule:encr}
  \Gamma \deriv M
  \and
  \Gamma \deriv K
}{
  \Gamma \deriv \enc{M}{K}
}
\and
\inferrule*[right=A\textsubscript{L}]{\label{rule:aencl}
  \Gamma, \aenc{M}{\pub{K}} \deriv K
  \and
  \Gamma, \aenc{M}{\pub{K}}, M, K \deriv N
}{
  \Gamma, \aenc{M}{\pub{K}} \deriv N
}
\and
\inferrule*[right=A\textsubscript{R}]{\label{rule:aencr}
  \Gamma \deriv M
  \and
  \Gamma \deriv N
}{
  \Gamma \deriv \aenc{M}{N}
}
\end{mathpar}   \vspace*{-2\belowdisplayskip}
  \caption{Deduction rules for the derivability relation of $\IntrSymm$}
  \label{fig:symmetric-intruder}
\end{figure}

\subsection{A Calculus for Cryptographic Protocols}
\label{sec:calculus}

A common approach to model cryptographic primitives is to
consider both constructors (e.g.~encryption) and destructors (e.g.~decryption).
Here messages only contain constructors,
and ``destruction'' is represented by pattern matching.
Fix a finite signature $\ProcN$ of process names
(ranged over by~$\p Q$)
each of which has a fixed arity $\arity(\p Q) \in \Nat$.
A protocol specification consists of an initial process $P$ and
a finite set $\Defs$ of (possibly recursive) definitions
of the form $ \p Q[x_1,\dots,x_n] \is A $,
with ${\arity(\p Q) = n}$,
where the syntax of $P$ and $A$ follows the grammar:
\begin{grammar}
  P \is
    \zero
    | \new x. P
    |  P \,{\parallel}\, P
    | \out{M}
    | \p Q[\vec{M}]
    & \text{(process)}
    \\
  A \is
\outpc{a}{M}
    | \inpc{a}{\vec{x}:M}.P
    | A + A
    & \text{(action)}
\end{grammar}
We use the vector notation
  $\vec{x} = x_1,\dots,x_n$
for lists of pairwise distinct names.
In an action $ \inpc{a}{\vec{x}:M}.P $, we call
  $\vec{x}:M$ the \emph{pattern},
and $P$ the \emph{continuation};
processes $\p Q[\vec{M}]$ are called \emph{process calls}.
If $\Gamma = \set{M_1, \ldots, M_k}$ is a finite set of messages,
then $\out{\Gamma} \is \out{M_1} \parallel \ldots \parallel \out{M_k}$.
We define $P^0 \is \zero$ and $ P^{n+1} \is P \parallel P^n $.
For brevity, we assume the special name $\mathbf{in}$ is known to the intruder.
The internal action~$\tact$,
is an abbreviation for $\inp{x:x}$, for a fresh~$x$.
Processes of the form $\out{M}$ or $\p Q[\vec{a}]$ are called \emph{sequential}.
The names $\vec{x}$ are bound in both
  $ \new\vec{x}.P $ and $ \inpc{c}{\vec{x}:M}.P $.
We denote the set of free names of a term $P$ with $\freenames(P)$
and the set of bound names with $\boundnames(P)$.
As is standard, we require, wlog, that
  $ \freenames(P) \inters \boundnames(P) = \emptyset $.
When nesting restrictions $\new\vec{x}.\new\vec{y}.P$,
we implicitly assume wlog that $\vec{x}$ and $\vec{y}$ are disjoint.
We assume there is at most one definition for each $\p Q \in \ProcN$,
and that for each definition  $ \p Q[x_1,\dots,x_n] \is A $,
${\freenames(A) \subseteq \set{x_1,\dots,x_n}}$.
The set $\Proc$ consists of all processes over an underlying signature $\ProcN$.

\subparagraph{Structural congruence}
We write $\aeq$ for standard \pre\alpha-equivalence.
Structural congruence, $\congr$,
is the smallest congruence relation that
includes $\aeq$,
and is associative and commutative with respect to
  $\parallel$ and $+$
  with $\zero$ as the neutral element,
and satisfies the standard laws:
  $\new a . \zero \congr \zero$,
  $\new a . \new b . P \congr \new b . \new a . P$, and
  $P \parallel \new a . Q \congr \new a . (P \parallel Q)$
    if $a \not\in \freenames(P)$.
Every process P is congruent to a process in \emph{standard form}:
\begin{equation}
  \new \vec{x}.\bigl(
    \out{M_1} \parallel \dots \parallel \out{M_m} \parallel
    \p{Q_1}[\vec{N}_1] \parallel \dots \parallel \p{Q_k}[\vec{N}_k]
  \bigr)
  \tag{SF}
  \label{def:stdform}
\end{equation}
where every name in $\vec{x}$ occurs free in some subterm.
We write $\stdf(P)$ for the standard form of~$P$,
which is unique up to \pre\alpha-equivalence,
and associativity and commutativity of parallel.
We abbreviate standard forms with
$\new\vec{x}.(\out{\Gamma} \parallel Q)$
where
all the active messages are collected in $\Gamma$, and
$Q$ is a parallel composition of process calls. Let $\stdf(P)$ be the expression~\eqref{def:stdform},
we define $\msg(P) = \set{M_1,\dots,M_m} \union \Union_{i=1}^k \vec{N}_i$.
Thus $\msg(P)$ is the set of messages appearing in a term.
When ${m=0}, k=0, \vec{x}=\emptyset$,
the expression~\eqref{def:stdform} is~$\zero$.

\subparagraph{Reduction semantics}
One can think of standard forms
$ \new\vec{x}.(\out{\Gamma} \parallel Q) $
as runtime configurations of the protocol.
They capture,
at a specific point in time,
the current relevant names (which encode nonces/keys/data),
the knowledge of the intruder~$\Gamma$,
and the local state of each participant.
A sequential term $ \p{Q}[\vec{N}] $ represents a single participant
in control state $\p{Q}$ with local knowledge of messages~$\vec{N}$.

\begin{figure}[tb]
  \adjustfigure \begin{mathpar}
  \infer*[right={Comm}]{
    \p Q_1[\vec{M}_1] \defeq \outpc{c}{N\subst{\vec{x}->\vec{M}'}}.P_1 + A_1 \\
    \p Q_2[\vec{M}_2] \defeq \inpc{c}{\vec{x}:N}.P_2 + A_2 \\
  }{\new \vec{a}.(
      \out{\Gamma}  \parallel
      \p Q_1[\vec{M}_1] \parallel
      \p Q_2[\vec{M}_2] \parallel
      C
    )
    \redto_\Defs
      \new \vec{a}.(
        \out{\Gamma}               \parallel
        P_1                        \parallel
        P_2\subst{\vec{x}->\vec{M}'} \parallel
        C
      )
  }
  \label{rule:comm}
  \\
  \infer*[right={Struct}]{
    P \congr P'
    \redto_\Defs
    Q' \congr Q
  }{P\redto_\Defs Q
  }
  \label{rule:struct}
  \and
  \infer*[right={PubOut}]{
    \p Q[\vec{M}] \defeq \outpc{c}{M}.P + A \\
    \Gamma \deriv c
}{\new \vec{a}.(
      \out{\Gamma}  \parallel
      \p Q[\vec{M}] \parallel
      C
    )
    \redto_\Defs
      \new \vec{a}.(
        \out{\Gamma}  \parallel
        \out{M} \parallel
        P \parallel
        C
      )
  }
  \label{rule:pub-out}
  \\
  \infer*[right={PubIn}]{
    \p Q[\vec{M}] \defeq \inpc{c}{\vec{x}:N}.P + A \\
    \Gamma, \vec{y} \deriv N \subst{\vec{x} -> \vec{M'}}  \\
    \Gamma \deriv c \\
    \vec{y} \text{ fresh}
}{\new \vec{a}.(
      \out{\Gamma}  \parallel
      \p Q[\vec{M}] \parallel
      C
    )
    \redto_\Defs
      \new \vec{a}.\new \vec{y}.(
        \out{\Gamma}  \parallel
        \out{\vec{y}} \parallel
        P\subst{\vec{x}->\vec{M'}} \parallel
        C
      )
  }
  \label{rule:pub-in}
\end{mathpar}   \caption{Operational semantics.}
  \label{fig:op-sem}
\end{figure}

Principals can communicate through channels;
a channel known by the intruder is considered insecure.
An input action over an insecure channel can be fired
if the intruder can produce any message that matches the action's pattern.
An output $\outpc{c}{M}$ to an insecure channel~$c$ leaks message~$M$ to the
intruder, who can decide to forward it angelically to a corresponding
input over~$c$ (modelling an honest step) or hijack the communication.

We write $ \p Q[\vec{M}] \defeq A $
  if $ \p Q[\vec{x}] \is A' \in \Defs $ and $ A \aeq A'[\vec{M}/\vec{x}] $,
  up to commutativity and associativity of~$+$.
The transition relation $ \redto_\Defs $ is defined in \cref{fig:op-sem}.
In \cref{rule:pub-in}, $\vec{y}$~denotes all fresh names introduced by the intruder in this step.
Thanks to \eqref{axiom:relevancy},
one can wlog ignore transitions where
$ \freenames(\vec{y}) \not\subseteq \freenames(\vec{M'}) $,
since unused names would simply not contribute to the intruder knowledge.
The sets
$
  \reach_\Defs(P) \is \set{Q | P \redto_\Defs^* Q}
$
and
$
  \traces_\Defs(P) \is
    \set{ Q_0 \cdots Q_n | P \kcongr Q_0 \redto_\Defs \cdots \redto_\Defs Q_n }
$
collect
  the processes reachable from $P$ and
  all the transition sequences from $P$ respectively,
given the definitions $\Defs$.
We omit $\Defs$ when unambiguous.

\begin{definition}[$\kcongr$]
  \emph{Knowledge congruence}, $P\kcongr Q$,
  is the smallest congruence that
  includes $\congr$ and such that
  $\out{\Gamma_1} \kcongr \out{\Gamma_2}$ if $\Gamma_1 \keq \Gamma_2$.
\end{definition}
Knowledge congruence is also characterised by
\[
  P_1 \kcongr P_2 \iff
    \stdf(P_1) \aeq \new\vec{x}.(\out{\Gamma_1} \parallel Q)
    \land
    \stdf(P_2) \aeq \new\vec{x}.(\out{\Gamma_2} \parallel Q)
    \land
    {\Gamma_1 \keq \Gamma_2}.
\]
Intuitively, modulo derivability, two processes $ P \kcongr Q $
are indistinguishable to the intruder and to the principals.
Formally, if $P \kcongr Q$ then
  the transitions systems ${(P, \redto_\Defs)}$ and ${(Q, \redto_\Defs)}$
  are isomorphic.
We thus close the reduction semantics under knowledge congruence:
we add the rule that if $
  P \kcongr P'\redto_\Defs Q' \kcongr Q
$ then $
  P\redto_\Defs Q
$.

While knowledge congruence captures when two configurations are essentially the same, knowledge embedding formalises the notion of ``sub-configuration''.

\begin{figure}[tb]
  \adjustfigure \begin{align*}
\p{S}[a, b, k_{as}, k_{bs}] &\is
  \inp{n_{a} : (n_{a}, b)}.
    \new k.\bigl(
      \out*{\enc{k}{k_{bs}}} \parallel
      \out*{\enc{k}{(n_{a}, k_{as})}} \parallel
      \p{S}[a, b, k_{as}, k_{bs}]
    \bigr)
\\
\p{A_1}[a, b, k_{as}] &\is
  \tact.\new n_{a}.(
    \out{(n_{a}, b)} \parallel
    \p{A_2}[a, b, k_{as}, n_{a}] \parallel
    \p{A_1}[a, b, k_{as}]
  )
\\
\p{A_2}[a, b, k_{as}, n_{a}] &\is
  \inp*{k : \enc{k}{(n_{a}, k_{as})}}.\p{A_3}[a, b, k_{as}, k]
\\
\p{A_3}[a, b, k_{as}, k] &\is
  \inp{n_{b} : \enc{n_{b}}{k}}.\out{\enc{n_{b}}{(k, k)}}
\\
\p{B_1}[a, b, k_{bs}] &\is
  \inp{k : \enc{k}{k_{bs}}}.
    \new n_{b}.\bigl(
      \out{\enc{n_{b}}{k}} \parallel
      \p{B_2}[a, b, k_{bs}, n_{b}, k] \parallel
      \p{B_1}[a, b, k_{bs}]
    \bigr)
\\
\p{B_2}[a, b, k_{bs}, n_{b}, k] &\is
  \inp{\enc{n_{b}}{(k, k)}}.\p{Secret}[k]
\end{align*} \caption{Formal model of \cref{ex:running}.}
  \label{fig:running-model}
\end{figure}

\begin{definition}[Knowledge embedding]
  The \emph{knowledge embedding} relation
  ${P_1 \kembed P_2}$ holds if
  $
    P_1 \congr \new\vec{x}.(\out{\Gamma_1} \parallel Q)
  $, $
    P_2 \congr \new\vec{x}.\new\vec{y}.(\out{\Gamma_2} \parallel Q \parallel Q')
  $ and $
    {\Gamma_1 \kleq \Gamma_2}
  $.
\end{definition}

\begin{proposition}
\label{prop:kn-embed-vs-kn-congr}
  $ P_1 \kcongr P_2 $
    if and only if
  $ P_1 \kembed P_2 $ and $ P_2 \kembed P_1 $.
\end{proposition}

\begin{theorem}
\label{th:kembed-simulation}
  Knowledge embedding is a simulation, that is,
  for all $P$, $P'$ and $Q$,
  if $ P \redto Q $ and $ P \kembed P' $
  then there is a $Q'$ such that
    $ P' \redto Q' $ and $Q \kembed Q'$.
\end{theorem}

\begin{example}
\label{ex:running}
  Consider the following toy protocol,
  given in Alice\&Bob notation,
  meant to establish a new session key $K$ between $A$ and $B$
  through a trusted server $S$:
\par\noindent \begin{minipage}[c]{.5\textwidth}\begin{alicebob}
  \step A -> S: N_A, B \label{step:AS}
  \step S -> B: \enc{K}{(N_A,K_{AS})}, \enc{K}{K_{BS}} \label{step:SB}
\end{alicebob}\end{minipage}\begin{minipage}[c]{.5\textwidth}\begin{alicebob}[3]
  \step B -> A: \enc{K}{(N_A,K_{AS})}, \enc{N_B}{K} \label{step:BA}
  \step A -> B: \enc{N_B}{(K,K)} \label{step:AB}
\end{alicebob}\end{minipage}   \vspace{\belowdisplayshortskip}
  \par\smallskip\noindent \Cref{fig:running-model} shows the protocol formalised in our calculus.
  Assume the initial state~is
  \[
    P_0 =
      \new a,b,k_{as},k_{bs}.(
        \p S[a,b,k_{as},k_{bs}] \parallel
        \p A_1[a,b,k_{as}] \parallel
        \p B_1[a,b,k_{bs}] \parallel
        \out{a,b}
      ).
  \]
\Cref{step:AS} is initiated by $\p A_1$ which sends
  some new name $n_a$ to the server;
  since communication is over an insecure channel,
  the message is just output without indicating the intended recipient.
  The server receives the message (or any message the intruder may decide to forge instead)
  and outputs the fresh key~$k$ encrypted with $ k_{bs} $
  (the long-term key between~$B$ and~$S$)
  and with the pair $ (n_a,k_{as}) $
  (note the use of non-atomic encryption keys).
  In the protocol, these two messages are sent to $B$ but
  we model \cref{step:SB} by $\p B_1$ which just receives the message
  relevant to $B$.
  The forwarding of $\enc{k}{(n_a,k_{as})}$ from $S$ to $A$
  is performed by the intruder instead of $B$ in the model.

  In the last two steps, modelled by $\p B_2$ and $\p A_3$, $B$ sends a nonce $n_b$ encrypted with $k$,
  to challenge $A$ to prove she knows $k$, which she does by sending back
  $\enc{n_b}{(k,k)}$. At this point, $B$ is convinced that by encrypting messages with $k$
  they will be only accessible to $A$.
  We model this by making $\p B_2$ transition to $\p{Secret}[k]$
  after a successful challenge.
We always assume the definition
  $ \p{Secret}[k] \is \inp{k}.\p{Leak}[k] $.
  A transition to $\p{Leak}[k]$ is only possible when the intruder can derive $k$ so we can check whether the secrecy assertion holds by checking that
  no reachable process contains a call to $\p{Leak}[k]$.

  Notice how $\p A_1$ and $\p B_1$ spawn both the continuation of the session
  and (recursively) a process ready to start a new session.
  This creates the possibility of an unbounded number of sessions,
  each of which will involve fresh $n_a, n_b,$ and $k$.
\end{example}

\subparagraph{Threat model}
Our reduction semantics follows the Dolev-Yao attacker model
in representing the intruder's interference:
the intruder mediates every communication over insecure channels,
is able to create new names and analyse and construct messages from
all the messages that have been communicated insecurely so far.
Threat models that go beyond Dolev-Yao include dishonest participants
and compromised old session keys.
These aspects are not embedded in the semantics,
but can be modelled through the process definitions.
If we wanted to model compromised keys in \cref{ex:running}, for instance,
we could modify the definition of $ \p B_2 $ to
$
  \p{B_2}[a, b, k_{bs}, n_{b}, k] \is
    \inp{\enc{n_{b}}{(k, k)}}.\p{Secret}[k]
    +
    \inp{\enc{n_{b}}{(k, k)}}.\out{k}
$
which makes a non-deterministic choice to declare $k$ a secret,
or to consider it as old and reveal it.

\begin{remark}[Implementable patterns]
  Our calculus represents message deconstruction (e.g.~decryption)
  with pattern matching.
  However, general pattern matching is too powerful:
  a pattern like $\inp{x,k:\enc{x}{k}}$ would obtain \emph{both}
  the key $k$ and the plaintext $x$ from an encrypted message!
  This is only a modelling problem:
one should make sure all patterns can be implemented
  using the cryptographic primitives.
  Consider a pattern $\vec{x}:M$ and
  let $ Z = \names(M) \setminus \vec{x} $;
  the pattern is \emph{implementable},
  if, for all $\theta\from Z \pto \Msg$,
  we have
  $
    M\theta, Z\theta \deriv y
  $
  for all $y \in \vec{x}$.
\end{remark}

\subsection{Depth-Bounded Protocols}
\label{sec:depth}

We can now define the class of depth-bounded protocols,
a strict generalisation of the notion in~\cite{DOsualdoOT17}.
While the definitions of~\cite{DOsualdoOT17} depend on fixing the intruder
to symmetric encryption only,
here we define it fully parametrically to the intruder model.

\begin{definition}[Depth]
\label{def:depth}
The \emph{nesting of restrictions} of a term is given by the function
$
  \nestr(\p{Q}[\vec{a}]) \is \nestr(\out{M}) \is \nestr(\zero) \is 0$,
$  \nestr(\new x.P) \is 1 + \nestr(P)$,
$  \nestr(P \parallel Q) \is \max(\nestr(P), \nestr(Q)).
$
The \emph{depth} of a term is defined as
the minimal nesting of restrictions in its knowledge congruence class,
$
  \depth(P) \is \min \set{ \nestr(Q) | Q \kcongr P }.
$
\end{definition}

\begin{lemma}
  \label{lm:name-reuse}
  Every $Q$ is \pre\alpha-equivalent to a process $Q'$
  such that $\card{\boundnames(Q')} \leq \nestr(Q)$.
\end{lemma}

\noindent Consider for example
$
  P = \new a,b,c.(
    \out{a}
      \parallel
    \out{\enc{b}{a}}
      \parallel
    \out{\enc{c}{b}}
      \parallel
    \out{c}
  )
$
which has $\nestr(P) = 3$.
The process $P$ is knowledge-congruent to
$
  Q = (
    \new a.\out{a}
      \parallel
    \new b.\out{b}
      \parallel
    \new c.\out{c}
  )
$
which has $\nestr(Q) = 1$;
this gives us $\depth(P) = \nestr(Q) = 1$.
Although $\boundnames(Q) = \set{a,b,c}$,
by \pre\alpha-renaming all names to $x$ we obtain
$
  Q' = {(
    \new x.\out{x}
      \parallel
    \new x.\out{x}
      \parallel
    \new x.\out{x}
  )}
$
which has the property $\card{\boundnames(Q')} \leq \nestr(Q) \leq \depth(P)$.
More generally, \cref{lm:name-reuse} says that processes of depth $k$
can always be represented using at most~$k$ unique names,
by reusing names in disjoint scopes.

Let
$
  \Size{s} \is
    \set{ P \in \Proc | \forall M \in \msg(P)\colon \size(M) \leq s }
$
be the set of processes containing messages of size at most~$s$.
The set $\Depth[X]{s,k}$
is the set of processes of depth at most $k \in \Nat$,
with free names in $X$, and messages not exceeding size $s$:
\[
  \Depth[X]{s,k} \is
    \set{ P \in \Size{s} |
            \freenames(P) \subseteq X,
            \exists Q \in \Size{s}\st Q \kcongr P \land \nestr(Q) \leq k }.
\]
When starting from some initial process~$P_0$,
every reachable process~$P$ has $\freenames(P) \subseteq \freenames(P_0)$
so~$X$ can always be fixed to be $\freenames(P_0)$.
We therefore omit~$X$ from the superscripts to unclutter notation.
The set of processes reachable from~$P$ while respecting a size bound~$s$
is the set
$
  \reach_\Defs^s(P) \is
    \set{ Q
      | P\cdots Q \in \traces_\Defs(P) \inters \Size{s}^*
    }.
$

\begin{definition}
\label{def:sk-bounded}
  For some $s,k\in \Nat$,
  we say the process $P$ is \emph{\pre(s,k)-bounded} (w.r.t.~a finite set $\Defs$ of definitions) if
    $ \reach_\Defs^s(P) \subseteq \Depth{s,k} $,
  i.e.~from $P$ only processes of depth at most $k$ can be reached,
  in traces respecting the size bound $s$.
\end{definition}

\begin{example}
\Cref{ex:running} is \pre(3,7)-bounded.
We defer the proof of this fact to \cref{ex:running-is-bounded}.
\end{example}

\begin{example}[Encryption Oracle]
\label{ex:encryption-oracle}
  The definition $ \p E[k] \is \inp{x:x}.(\out{\enc{x}{k}} \parallel E[k]) $
  leads to unboundedness as soon as the initial process contains $\p E[k]$ for some $k$ not known to the intruder, and size bound such that $x$ can match messages of size greater than 1.
  In such case, the intruder can inject messages $(c_i,c_{i+1})$ for unboundedly many $i$, where $c_i$ are intruder-generated nonces.
Since $k$ is secret, the resulting reachable configurations would
  contain ``encryption chains'' of the form
  $
    \new k.\new c_1,\dots,c_n.(
      \out{\enc{(c_1,c_2)}{k}} \parallel
      \out{\enc{(c_2,c_3)}{k}} \parallel
      \dots
      \out{\enc{(c_{n-1},c_n)}{k}}
    ).
  $
  When such chains appear in a set for unboundedly many $n\in\Nat$,
  the set is not depth-bounded.
  This \emph{encryption oracle} pattern could be considered an anti-pattern because
  it can be exploited for a chosen-plaintext attack on the key $k$.
  The pattern can be usually modified or constrained to obtain a bounded protocol. One option is to limit the verification to only consider traces where~$x$ is of size 1.\footnote{In Tamarin one would obtain this by typing $x{:}\mathit{fresh}$.}
\end{example}

The two bounds~$s$ and~$k$ are very different in nature.
For size, we ignore any trace that involves messages exceeding size~$s$.
Then we determine if the depth bound~$k$ is respected by all remaining traces.
Ignoring traces exceeding~$s$ is acceptable for protocols not susceptible to type confusion attacks and is achieved in other tools by using typing.
Our method can however be pushed beyond this limitation:
In \cref{sec:limitations}, we show how the results of our analysis
on the traces of bounded message size, can be generalised to results that hold for the unrestricted set of traces.
 \section{Ideal Completions for Security Protocols}
\label{sec:ideal}

Our main technical contributions are the proofs needed to show
that \pre(s,k)-bounded protocols form a
post-effective ideal completion
in the sense of~\cite{BlondinFM18}.
First we outline the significance and applications of this result,
and then proceed with the proofs.

\subsection{Downward-Closed Invariants and Security Properties}
\label{sec:invariants}

Suppose we want to establish that a protocol $P$ fulfils some security requirement.
In a typical proof, one needs to establish many intermediate facts about
executions of the protocol.
For example, part of the argument may hinge on some key $k$ being always unknown to the intruder.
This kind of property is an \emph{invariant} of the protocol:
it holds at every step of an execution.
Formally, an invariant of {$P$
(under definitions $\Defs$ and size constraint $s$)}
is any set of processes that includes $\reach_\Defs^s(P)$.
For example,
  $k$ is never leaked to the intruder in executions of the protocol $P$
  if the set of processes
    $ \mathcal{S}_k \is \set{Q | \out{k} \not\kembed Q} $
  ---i.e.~all the processes where $k$ is not public---
  is an invariant of $P$.
We will focus here on the class of \pre\kembed-downward-closed invariants.
Formally, given a set of processes $X$,
its \emph{\pre\kembed-downward closure} is the set
$ \dwd{X} \is \set{ Q | \exists P \in X\st Q \kembed P } $.
A set $X$ is \emph{\pre\kembed-downward closed} if $ X = \dwd{X} $.
Many properties of interest are naturally downward closed.
For example,
  the set $ \mathcal{S}_k $ above is downward-closed
  as $\out{k} \not\kembed Q$ and $Q' \kembed Q$
  implies $\out{k} \not\kembed Q'$.

The problem we need to solve is, then, how to show that a given
downward-closed set $X$ is an invariant for a given protocol.
Formally, that corresponds to checking
$ X \supseteq \reach_\Defs^s(P) $ which, by downward-closure of $X$,
is equivalent to checking $ X \supseteq \dwd{\reach_\Defs^s(P)} $.
To prove the latter inclusion, our strategy is to find
an inductive invariant that includes the initial state $P$
and that is included in $X$.
Let $\post_\Defs^s(X) \is \set{Q' | \exists Q \in X, Q \to Q' \in \Size{s} }$ be the set of processes
reachable in one step from processes in $X$.
An invariant $X$ is \emph{inductive} if $X \supseteq \post_\Defs^s(X) $,
which is equivalent to requiring $X \supseteq \dwd{\post_\Defs^s(X)} $
if $X$ is downward-closed.
Any inductive invariant that contains the initial process $P$ will
include $\reach_\Defs^s(P)$.

\noindent
To turn this proof strategy into an algorithm,
we need three components:
\begin{enumerate}
  \item a recursively enumerable
    finite representation of downward-closed sets,
    \label[problem]{prob:fin-repr-dwd}
  \item a way to decide inclusion between two downward-closed sets,
    given their representation,\label[problem]{prob:decide-inclusion}
  \item an algorithm (called $\posthat_\Defs^s$) to compute,
    given a finite representation of a downward-closed set $D$,
    a finite representation of $\dwd{\post_\Defs^s(D)}$.
    \label[problem]{prob:decide-posthat}
\end{enumerate}
Unfortunately, downward-closed sets cannot be finitely represented in general,
especially if one considers unbounded sessions/nonces.
We will show, however, that we can devise solutions to all three items
above for downward-closed subsets of $\Depth{s,k}$,
under a mild restriction on the intruder model.
Solving \cref{prob:fin-repr-dwd,prob:decide-inclusion,prob:decide-posthat}
amounts to proving that
$\Depth{s,k}$ admits a \emph{post-effective ideal completion}
in the sense of~\cite{FinkelG09,BlondinFM18}.
This implies that we can decide if the reachable configurations satisfy any
given downward-closed property, by adapting the enumeration scheme
presented in~\cite{BlondinFM18}.

\begin{theorem}
  \label{th:dwdcl-decidable}
  Given a property $ \mathcal{D} \subseteq \Proc $,
  and a protocol $P$ with definitions $\Defs$,
  we write $ P,\Defs \models^s \mathcal{D} $ if
  $\reach^s_\Defs(P) \subseteq \mathcal{D}$.
  If $\mathcal{D} \subseteq \Depth{s,k}$ is downward-closed,
then $ P,\Defs \models^s \mathcal{D} $ is decidable.
\end{theorem}

\begin{proof}
  The algorithm runs two semi-procedures, Prover and Refuter, in parallel.
  The first procedure, Prover,
  enumerates all the downward-closed subsets $I$ of $\Depth{s,k}$.
For each $I$, Prover checks if
  \begin{enumerate}[(a)]
    \item $P \in I$, \label[equation]{check:PinI}
    \item $I$ is inductive, by checking $ \posthat_\Defs^s(I) \subseteq I $, and
      \label[equation]{check:I-inductive}
    \item $I \subseteq \mathcal{D}$.
      \label[equation]{check:IinD}
  \end{enumerate}
  If we find such a set $I$,
  then we have proven $\reach_\Defs(P) \subseteq \mathcal{D}$,
  and the overall algorithm can terminate returning ``True''.
  The second procedure, Refuter,
  enumerates all $ Q \in \reach_\Defs(P) $
  and checks if $ Q \not\in \mathcal{D} $,
  in which case the overall algorithm can terminate returning ``False''.

  When $ P,\Defs \models^s \mathcal{D} $ holds,
  then, in the worst case, Prover will eventually consider
  the finite representation of $\dwd{\reach_\Defs^s(P)}$,
  which satisfies checks~\labelcref{check:PinI,check:I-inductive,check:IinD}
  above.
  In the case where  $ P,\Defs \models^s \mathcal{D} $ does not hold,
  Refuter would eventually find a reachable process not in $\mathcal{D}$.
  In either case, the algorithm terminates with the correct answer.
\end{proof}

\noindent
The above algorithm can be used to decide
the following properties.

\subparagraph{Deciding \pre(s,k)-boundedness}
The set $\Depth{s,k}$ is itself downward-closed,
so we can decide if $P$ is \pre(s,k)-bounded,
by deciding $ P,\Defs \models^s \Depth{s,k} $.

\subparagraph{Deciding control-state reachability and secrecy}
\label{par:secrecy-encoding}
Control-state reachability asks whether
there is an execution of the protocol which reaches a process
containing a process call $ \p Q[\dots] $ for some given $\p Q$.
Secrecy can be reduced to control-state reachability by introducing
a definition $\p{Secret}[m]\is \inp{m}.\p{Leak}[m]$
(for a special process identifier $\p{Leak}$ with no definition).
In the definition of the protocol one can call $\p{Secret}[m]$
to mark some message $m$ as a secret,
and secrecy corresponds to asking control-state (un)reachability for $\p{Leak}$.

If $P$ is \pre(s,k)-bounded,
control-state reachability for $\p Q$ from $P$,
can decided by $ P,\Defs \models^s \mathcal{D}_{\p Q} $,
where $\mathcal{D}_{\p Q}$ is the (downward-closed) subset of $\Depth{s,k}$
of processes that do not contain calls to $\p Q$.
Notice that, when $P$ is arbitrary, the algorithm checks
\pre(s,k)-boundedness and control-state reachability at the same time.

\subparagraph{Absence of misauthentication}
A misauthentication happens when a principal
$a$ believes she shares a secret $n$ with $b$ but
$b$ believes she shares the secret $n$ with some other entity~$c$.
To check this situation can never arise,
we can produce the process $\p{Auth}[a,b,n]$ when $a$ believes to share the secret $n$ with $b$.
Absence of misauthentication can be decided by
$ P, \Defs \models^s \mathcal{A} $ where
$
  \mathcal{A} =
    \Set{ Q \in \Depth{s,k} |
      \new a,b,c,n.(\p{Auth}[a,b,n] \parallel \p{Auth}[b,c,n]) \not\kembed Q }.
$

\subparagraph{Susceptibility to known-plaintext attacks}
The task of guessing a symmetric key is made much easier if it is possible
for the attacker to have access to an arbitrarily large number
of known nonces encrypted with the same key~$k$.
We can model this situation by asking if:\begin{equation}
  \forall n\in\Nat\st
    \exists Q \in \reach^s_\Defs(P) \st R^n \kembed Q
    \qquad\qquad\text{where }
    R = \new m.(\out{m} \parallel \out{\enc{m}{k}})
  \tag{$\dagger$}
  \label{def:known-plaintext}
\end{equation}
If~\eqref{def:known-plaintext} holds for $P$ then
the intruder does have access to an unbounded supply of known messages $m$
encrypted with the same key $k$.
Interestingly, the property becomes meaningful
only when considering unbounded number of nonces.
Condition~\eqref{def:known-plaintext} is equivalent to
$ \dwd{\reach^s_\Defs(P)} \supseteq \dwd{\set{R^n | n \in \Nat}} $.
If we find a downward-closed inductive invariant $I$ for $P$,
such that $ I \not\supseteq \dwd{\set{R^n | n \in \Nat}} $,
then we can be sure that \eqref{def:known-plaintext} does not hold,
and $P$ is not susceptible to known-plaintext attacks on~$k$.
We can therefore semi-decide~\eqref{def:known-plaintext} by enumerating
all candidate~$I$.
Contrary to the previous algorithms, we are not able
to provide a Refuter procedure (we conjecture the problem is undecidable).
We can get a decision procedure if, instead of unboundedly many plaintext-encrypted pairs we ask whether a sufficiently high, user-provided number~$N$ of such pairs can be produced.
The problem can be extended to cover the case where the known-plaintext
can be any message of size at most~$s$.

\medskip

Notice how, if the protocol is found to satisfy the property,
the algorithms above can output an inductive invariant acting as an independently checkable certificate of correctness.
Although the mentioned security properties alone do not cover the full security requirements, an effectively presented invariant can provide the foundations to prove further properties.

The algorithm of \cref{th:dwdcl-decidable}
relies on expensive enumeration schemes,
which are mainly a theoretical device to prove decidability.
In a more practical setting,
the candidate invariants can be supplied by the user and refined interactively,
avoiding the need for the enumeration of Prover,
or they can be inferred as we describe in \cref{sec:algorithms}.

\subsection{Bounded Processes are Well-Quasi-Ordered}
\label{sec:wqo}

We construct finite representations of downward-closed invariants
by making use of the algebraic structure of
the quasi-order $(\Depth{s,k}, \kembed)$.
A relation ${\qoleq} \subseteq S\times S$ over some set $S$
is a \emph{quasi-order} (qo) if it is reflexive and transitive.
An infinite sequence $s_0,s_1,\ldots$ of elements of $S$ is called \emph{good}
if there are two indexes $i < j$ such that $s_i \qoleq s_j$.
A qo $(S,\qoleq)$ is called a \emph{well quasi order}~(wqo)
if all its sequences are good.

We prove $(\Depth{s,k}, \kembed)$ is a wqo
for any intruder model, by showing a correspondence between
processes in $\Depth{s,k}$ and finitely-labelled
forests of height at most $k$, which we represent as nested multisets.
The details can be found in the \cref{app:wqo-limits}.

\subsection{Limits and Ideal Decompositions}
\label{sec:limits}

By exploiting the wqo structure of $\Depth{s,k}$,
we can provide a finite representation for its downward-closed sets.
Let $(S,\qoleq)$ be a qo.
A set $D \subseteq S$ is an \emph{ideal} if it is downward-closed and directed,
i.e.~for all $x,y \in D$
  there is a $z\in D$ such that $x\qoleq z$ and $y \qoleq z$.
We write $\Idl{S}$ for the set of ideals of $S$.
It is well-known that in a well-quasi-order,
every downward-closed set is equal to
a canonical minimal finite union of ideals,
its \emph{ideal decomposition}.
To represent downward-closed sets of $\Depth{s,k}$ we will only need to provide
finite representations of its ideals.
We represent ideals using \emph{limits},
which have the same syntax as processes
augmented with a construct $\hole^\omega$ to represent
an arbitrary number of parallel components.

\begin{definition}[Limits]
  \label{def:limits}
  We call \emph{limits} the terms $L$ formed according to the grammar:
  \begin{grammar}
    \Limits \ni L \is
      \zero  | (R_1 \parallel \cdots \parallel R_n)
    &\quad R \is
        B
      | B^\omega
    &\quad B \is
        \out{M}
      | \p Q[\vec{M}]
      | \new x.L
  \end{grammar}
\end{definition}

\begin{definition}[Denotation of limits]
  \label{def:limits-sem}
  The denotation of $L$ is the set ${\sem{L}\is\dwd{\instn{L}}}$ where:
\begin{align*}
\instn{\zero}         &\is \set{ \zero }
    &
    \instn{L_1 \parallel L_2} &\is
      \set{ (P_1 \parallel P_2) |
                P_1 \in \instn{L_1}, P_2 \in \instn{L_2}}
    \\
    \instn{\p Q[\vec{M}]} &\is \set{ \p Q[\vec{M}] }
    &
    \instn{B^\omega}      &\is
      \textstyle
      \Union_{n \in \Nat}
      \set{ (P_1 \parallel \dots \parallel P_n) |
                \forall i\leq n: P_i \in \instn{B}} \\
    \instn{\out{M}}       &\is \set{ \out{M} }
    &
    \instn{\new x.L}      &\is \set{ \new x.P | P \in \instn{L}}
  \end{align*}
  We call the processes in $\sem{L}$ \emph{instances} of $L$.
Define $\nestr(L)$ to be as $\nestr$ on processes
  with the addition of the case $\nestr(L^\omega) \is \nestr(L)$.
  It is easy to check that for each $P \in \sem{L}$,
  $\depth(P) \leq \nestr(L)$.
  We write $\Limits_{s,k}^X$ for the set of limit expressions $L$ with free names in $X$ that have $\nestr(L) \leq k$ and do not contain messages of size exceeding $s$.
  We often omit $X$ and understand it is a fixed finite set of names.
\end{definition}

\begin{theorem}
  \label{th:limit-sound-complete}
  Limits faithfully represent ideals:
  ${
    I \in \Idl{\Depth{s,k}}
    \iff
    \exists L \in \Limits_{s,k} \st
      I = \sem{L}
  }$.
\end{theorem}

\subsection{Decidability of Inclusion}
\label{sec:inclusion}

Now we turn to decidability of inclusion between downward-closed sets.
It is well-known that in a wqo,
given the ideal decomposition of two downward-closed sets
$D_1 = I_1 \union \dots \union I_n$ and $D_2 = J_1 \union \dots \union J_m$,
we have $ D_1 \subseteq D_2 $
  if and only if
for all $1 \leq i \leq n$,
  there is a $1 \leq j \leq m$,
    such that $I_i \subseteq J_j$.
Hence, decidability of ideals inclusion implies
decidability of downward-closed sets inclusion.

We extend structural congruence to limits in the obvious way,
with the addition of the law $ \out{M}^\omega \congr \out{M} $
obtaining that $ L \congr L' $ implies $ \sem{L} = \sem{L'} $.
We can define a standard form for limits:
every limit is structurally congruent to a limit of the form
$
  \new \vec{x}.(
    \out{\Gamma} \parallel
    \Parallel_{i\in I}\p Q_i[\vec{M}_i] \parallel
    \Parallel_{j\in J}B_j^\omega
  )
$
where every name in $\vec{x}$ occurs free at least once
in the scope of the restriction,
and for all $j\in J$, $B_j$ is also in standard form.
When we write
  $\stdf(L) \aeq \new \vec{x}.(\out{\Gamma} \parallel Q \parallel R)$
we imply that $Q$ is a parallel composition of process calls
$\Parallel_{i\in I}\p Q_i[\vec{M}_i]$
(in which case we write~$\card{Q}$ for~$\card{I}$)
and $R$ is a parallel composition of iterated limits
$\Parallel_{j\in J} B_j^\omega$.

To better manipulate limits,
we introduce, in \cref{fig:grounding-extension},
the \pre n-th grounding $\lground{L}{n}$
and the \pre n-th extension $\lext{L}{n}$,
of a limit~$L$.
Grounding
replaces
each $\hole^\omega$ with $\hole^n$,
with the obvious property that $\lground{L}{n} \in \sem{L}$.
An extension $\lext{L}{n}$ produces a new limit with
each sub-limit $B^\omega$  unfolded $n$ times.
Note that extension does not alter semantics:
$\sem{L} = \sem{\lext{L}{n}}$.

\begin{figure}[tb]
  \adjustfigure \resizebox{.5\textwidth}{!}{$
    \lground{L}{n} \is
      \begin{cases}
        L                     & \text{if $L$ is sequential or }\zero\\
        \lground{L_1}{n} \parallel \lground{L_2}{n}
                              & \text{if } L = L_1 \parallel L_2\\
        \new x.(\lground{L'}{n}) & \text{if } L = \new x.L'\\
        \left(\lground{B}{n}\right)^n       & \text{if } L = B^\omega
      \end{cases}
    $}\resizebox{.5\textwidth}{!}{$
      \lext{L}{n} \is
        \begin{cases}
          L                     & \text{if $L$ is sequential or }\zero\\
          \lext{L_1}{n} \parallel \lext{L_2}{n}
                                & \text{if } L = L_1 \parallel L_2\\
          \new x.(\lext{L'}{n}) & \text{if } L = \new x.L'\\
          \left(\lext{B}{n}\right)^n \parallel B^\omega
                                & \text{if } L = B^\omega
        \end{cases}
  $}
  \caption{The grounding
      $\lground{\hole}{n} \from \Limits_{s,k} \to \Depth{s,k}$
    and extension
      $\lext{\hole}{n} \from \Limits_{s,k} \to \Limits_{s,k}$
    operations on limits.}
  \label{fig:grounding-extension}
\end{figure}

\subparagraph{The absorption axiom}
The decidability proof hinges on a characterisation of inclusion
that requires an additional hypothesis on the intruder model.

\begin{definition}[Absorbing intruder]
  \label{lm:deriv-variant-idemp}
  \label{def:absorption}
  Fix an intruder model $\Intruder = (\Sig,\deriv)$.
  Let
    $\vec{x}$ and $\vec{y}$ be two lists of pairwise distinct names,
    $\Gamma$ be a finite set of messages, and
    $\Gamma' = \Gamma\subst{\vec{x} -> \vec{y}}$.
  Moreover, assume that $ \names(\Gamma) \inters \vec{y} = \emptyset $.
  We say $\Intruder$ is \emph{absorbing} if,
  for all messages $M$ with $\names(M) \subseteq \names(\Gamma)$,
  we have that
    $\Gamma,\Gamma' \deriv M$
      if and only if
    $\Gamma \deriv M$.
\end{definition}

\begin{lemma}
\label{lm:symm-absorbing}
  $\IntrSymm$ is absorbing.
\end{lemma}

\noindent
For the rest of the paper,
we assume an absorbing intruder model.
The absorption axiom has a technical definition,
which becomes more intuitive if understood in the context
of limits of the form
$
  L = \bigl(\new\vec{x}.(\out{\Gamma}\parallel Q)\bigr)^\omega
$.
Imagine comparing the difference in knowledge between
$\lground{L}{1}$ and
$\lground{L}{2}$:
we have
$
  \stdf(\lground{L}{2}) \aeq
    \bigl(
      \new\vec{x}.\new\vec{x}'.(
        \out{\Gamma} \parallel
        \out{\Gamma'} \parallel
        Q \parallel
        Q'
      )
    \bigr)
$,
where $\Gamma' = \Gamma \subst{\vec{x}->\vec{x}'}$.
The absorption axiom tells us that if we want to check
whether a process $\new\vec{x}.\out{M}$ is embedded in $\stdf(\lground{L}{2})$,
we only need to check if $M$ is derivable from $\Gamma$
and we can ignore $\Gamma'$.
In other words, we only need to check if $\new\vec{x}.\out{M}$ is embedded in $\lground{L}{1}$.

We are now ready to prove our main result:
a small model property that shows decidability of limit inclusion.
Let us present the intuition on the simpler problem of deciding inclusion when one of the limits
is a single process $P$,
i.e.~deciding if $P \in \sem{L}$.
Take $P = \new a,b.( \out{\enc{a}{b}} \parallel \p A[b] )$
and $L = \bigl(\new x.(\out{x} \parallel A[x]) \bigr)^\omega$.
Suppose we replicate the $\omega$ twice,
obtaining the equivalent limit
$
  \lext{L}{2} \congr
    \new x_0,x_1.(
      \out{x_0} \parallel A[x_0] \parallel \out{x_1} \parallel A[x_1]
      \parallel L
    )
$.
The idea is that we can match $P$ against the fixed part of $\lext{L}{2}$:
\begin{align*}
  P \kcongr
    \new x_0,x_1.( \out{\enc{x_0}{x_1}} \parallel \p A[x_1] )
    &\kembed
    \new x_0,x_1.(
      \out{\enc{x_0}{x_1}} \parallel A[x_0] \parallel A[x_1]
    )
    \\
    &\kembed
    \new x_0,x_1.(
      \out{x_0} \parallel \out{x_1} \parallel A[x_0] \parallel A[x_1]
    ) \in \lext{L}{2}
\end{align*}
The first observation is therefore that if we can find some $m$ such that
$P$ is embedded in the fixed part of $\lext{L}{m}$,
we have proven $P \in \sem{L}$.
To turn this into an algorithm, we need to prove that there exists an $n$
so that if we failed to embed $P$ in the fixed part of $\lext{L}{m}$,
for any $m \leq n$ then $P$ is not going to embed in $\lext{L}{m'}$
for \emph{every} $m'$, and therefore $P \not\in \sem{L}$.
In other words, we need to know that after some threshold $n$,
there is no point trying with bigger extensions.
Take for example
$ P = \new x, y.(\p B[x, y] || B[y, x]) $ and
$ L = (\new x, y.\p B[x, y])^\omega $.
We can try and embed $P$ into $\lext{L}{2}$ but we would fail as the fixed part
expands to
$\new x_0, y_0.\p B[x_0, y_0] \parallel \new x_1, y_1.\p B[x_1, y_1]$.
It is easy to see that expanding further would not introduce new patterns
in the fixed part of the limit which would help embed~$P$.

\Cref{th:char-lim-incl} formalises the idea for general inclusion
between two arbitrary limits $L_1$ and~$L_2$:
it proves that the threshold for expansion is the number of fixed restrictions
of $L_1$ plus the number of fixed process calls of $L_1$, plus one;
and it makes use of the absorption axiom to prove the threshold is sound
even in the presence of knowledge.

\begin{theorem}[Characterisation of Limits Inclusion]
\label{th:char-lim-incl}
 Let $L_1$ and $L_2$ be two limits,
 with
 $
   \stdf(L_1) \aeq \new\vec{x}_1.(
     \out{\Gamma_1} \parallel
     Q_1 \parallel
     \Parallel_{i\in I} B_i^\omega
   )
 $,
 and let $n = \card{\vec{x}_1} + \card{Q_1} + 1 $.
 Then:
 \[
   \sem{L_1} \subseteq \sem{L_2}
     \iff
   \begin{cases}
      \stdf(\lext{L_2}{n}) \aeq \new\vec{x}_1,\vec{x}_2.(
         \out{\Gamma_2} \parallel
         Q_1 \parallel
         Q_2 \parallel
         R_2
       )
     \;\text{ and }\;
     \Gamma_1 \kleq \Gamma_2 &
     \cond{A}\\
     \sem{\out{\Gamma_1} \parallel \Parallel_{i\in I} B_i}
       \subseteq
     \sem{\out{\Gamma_2} \parallel R_2}
     & \cond{B}
   \end{cases}
 \]
\end{theorem}

\begin{theorem}
\label{th:lim-incl-decidable}
  Given $L_1, L_2 \in \Limits$ it is decidable whether
  $\sem{L_1} \subseteq \sem{L_2}$.
\end{theorem}
\begin{proof}
  \Cref{th:char-lim-incl} leads to a recursive algorithm.
  Given $L_1$ and $L_2$, one computes
  $\stdf(L_1)$ and $\stdf(\lext{L_2}{n})$.
  For every \pre\alpha-renaming
    that makes condition~\cond{A} hold, one checks
    condition~\cond{B} (recursively).
  If no renaming makes both true then the inclusion does not hold.
  In the recursive case, there are fewer occurrences of $\omega$
  in the limit on the left, eventually leading to the case
  where $L_1$ has no occurrence of $\omega$,
  and only condition~\cond{A} needs to be checked.
\end{proof}

\begin{example}[Limit inclusion]
  Consider the following two limits:
  \begin{align*}
   L_1 & = \new x_1.\bigr(
                    ( \new x_2.
                      ( \out{\enc{x_2}{x_1}} \parallel
                        \p A[x_2] \parallel
                        \p A[x_1] ))^\omega
                  \bigl)
   \\
   L_2 & = \new y_1, y_3. \bigr(
                    \out{y_3} \parallel
                    \p A[y_3]^\omega \parallel
                    ( \new y_2.
                     ( \out{y_2} \parallel
                       \p A[y_2] ) )^\omega
                  \bigl)
  \end{align*}
  We prove that $\sem{L_1} \subseteq \sem{L_2}$ by applying the recursive algorithm from \cref{th:lim-incl-decidable}.
  By \cref{th:char-lim-incl},
  here the threshold for expansion is $n=2$,
  but we will try with $n = 1$.
  In case we succeed, the inclusion holds.
  If not, we might have to increase $n$ up to $2$.
  This results in
  \begin{align*}
    \lext{L_2}{1} & =
      \new y_1, y_3. \bigr(
                    \out{y_3} \parallel
                    \p A[y_3] \parallel
                    \p A[y_3]^\omega \parallel
                    ( \new y_2.
                     ( \out{y_2} \parallel
                       \p A[y_2] ) )^\omega \parallel
( \new y_2.
                     ( \out{y_2} \parallel
                       \p A[y_2] ) )
                  \bigl)
    \\
    & \kcongr
     \new y_1, y_{22}, y_3. \bigr(
                    \out{y_3} \parallel
                    \out{y_{22}} \parallel
                    \p A[y_{22}] \parallel
                    \p A[y_3] \parallel
                    \p A[y_3]^\omega \parallel
( \new y_2.
                     ( \out{y_2} \parallel
                       \p A[y_2] ) )^\omega
                  \bigl)
  \end{align*}
  To try and match the fixed part of $L_1$ with $\lext{L_2}{1}$,
  we could \pre\alpha-rename $x_1$ to $y_1$.
This works for \cond{A} but the remaining goal \cond{B} cannot be shown as it does not hold because we cannot derive $\out{\enc{x_2}{y_1}}$ from
  the knowledge on the right-hand~side:
  \begin{equation*}
   \sem{
      \new x_2.
      ( \out{\enc{x_2}{y_1}} \parallel
        \p A[x_2] \parallel
        \p A[x_1] )
   }
   \not\subseteq \sem{
      \out{y_3} \parallel
      \out{\enc{y_{22}}{y_3}} \parallel
      \p A[x_1]^\omega \parallel
      ( \new y_2.
       ( \out{\enc{y_2}{y_3}} \parallel
         \p A[y_2] ) )^\omega
   }
  \end{equation*}
  Choosing the \pre\alpha-renaming $\subst{y_3 -> x_1}$ leaves us instead with:
  \begin{equation*}
   \sem{
      \new x_2.
      ( \out{\enc{x_2}{y_3}} \parallel
        \p A[x_2] \parallel
        \p A[x_1] )
   }
   \subseteq \sem{
      \out{y_3} \parallel
      \out{\enc{y_{22}}{y_3}} \parallel
      \p A[x_1]^\omega \parallel
      ( \new y_2.
       ( \out{\enc{y_2}{y_3}} \parallel
         \p A[y_2] ) )^\omega
   }
  \end{equation*}
  which holds.
  It suffices to expand the latter limit by $1$
  and to choose the renaming $\subst{y_2 -> x_2}$.
  Then, $\enc{x_1}{x_2} \kleq x_1, x_2$ holds and the process calls match.
  Note that it is crucial to keep $\p A[x_1]^\omega$ on the right side.
\end{example}

\subsection{Computing Post-Hat}
\label{sec:posthat}

The last result we need is the decidability of a function
$ \posthat_\Defs^s(L) $ which, given a limit $L$,
returns a finite set of limits
$\set{L_1,\dots,L_n}$
such that
$
  \dwd{\post_\Defs^s(\sem{L})} =
    \sem{L_1} \union \dots \union \sem{L_n}.
$
The challenge is representing
all the possible successors of processes in $\sem{L}$
without having to enumerate $\sem{L}$.
The key idea hinges again on the absorption axiom:
we observe that to consider all possible process calls that may cause a  transition,
it is enough to unfold each $\omega$ in $L$ by some bounded number $b$.
Any transition taken from further unfoldings will give rise to successors
that are congruent to some of the ones already considered.

The bound $b$ used in extending a limit,
is defined as a function of the process definitions and the intruder model.
The arity of a pattern $\vec{x} : M$ is $\card{\vec{x}}$.
Given a set of definitions $\Defs$,
  $\beta(\Delta)$ is the maximum arity of patterns in $\Delta$.
The function $\gamma(\Intruder)$ returns the maximum arity of the constructors
in the signature of $\Intruder$.

\begin{definition}[$\posthat$]
\label{def:posthat}
  Let
  $ b = \beta(\Delta)  \cdot \gamma(\Intruder)^{s-1} + 1  $ and
  $\stdf(\lext{L}{b}) = {\new \vec{x}.(\out{\Gamma} \parallel Q \parallel R)}$,
\[
    \posthat_\Defs^s(L) \is
        \Set{\;
          \new \vec{y}.\bigl(\out{\Gamma'} \parallel Q' \parallel R\bigr)
          \;|\;
            \new \vec{x}.\bigl(\out{\Gamma} \parallel Q\bigr)
              \redto_\Defs
            \new \vec{y}.\bigl(\out{\Gamma'} \parallel Q'\bigr)
              \in \Size{s}
        \;}.
  \]
\end{definition}

\begin{theorem}\label{th:posthat-correct}
$
  \posthat_\Defs^s(L) = \set{L_1,\dots,L_n}
    \implies
      \dwd{\post_\Defs^s(\sem{L})} =
        \sem{L_1} \union \dots \union \sem{L_n}.
$
\end{theorem}

\subsection{Algorithmic Aspects}
\label{sec:algorithms}
\label{ex:running-is-bounded}

\begin{figure}[tb]
  \adjustfigure \resizebox{\textwidth}{!}{\vbox{\begin{align*}
L &= \new a, b, k_{as}, k_{bs}.(
      \out{a,b} \parallel
      \p{A_{1}}[a, b, k_{as}]^\omega \parallel
      \p{B_{1}}[a, b, k_{bs}]^\omega \parallel
      \p{S}[a, b, k_{as}, k_{bs}]^\omega \parallel
      L_1^\omega
    )
\\
L_1 &= \new n_{a}.\bigl( \out{n_{a}} \parallel
    \p{A_{2}}[a, b, k_{as}, n_{a}] \parallel
    L_2^\omega
  \bigr)
\\
L_2 &= \new k.\bigl(
    \out{\enc{k}{(a, k_{as})}} \parallel
    \out{\enc{k}{(b, k_{as})}} \parallel
    \out{\enc{k}{(n_{a}, k_{as})}} \parallel
    \out{\enc{k}{k_{bs}}} \parallel
    \p{Secret}[k]^\omega \parallel
    \p{A_{3}}[a, b, k_{as}, k]^\omega \parallel
    L_3^\omega
  \bigr)
\\
L_3 &= \new n_{b}.\bigl(
    \out{\enc{n_{b}}{(k, k)}} \parallel
    \out{\enc{n_{b}}{k}} \parallel
    \p{B_{2}}[a, b, k_{bs}, n_{b}, k]
  \bigr)
\end{align*} }}
\caption{An inductive invariant for \cref{ex:running}.}
  \label{fig:running-invariant}
\end{figure}

The limit~$L$ in \cref{fig:running-invariant} represents an inductive invariant
for \cref{ex:running}, under size bound~$3$.
It can be proven invariant by checking that $\posthat^3_\Defs(L)$
is included in~$L$.
Since the initial process of \cref{ex:running},~$P_0$, is in $\sem{L}$,
the invariant certificates that any reachable process
is \pre(3,7)-bounded (note $\nestr(L)=7$) and
satisfies secrecy. In fact, $L^\omega$ is also inductive,
proving boundedness and secrecy for any process in $(P_0)^\omega$.
Since $
  \sem{\new k.(\new x.(\out{x,\enc{x}{k}}))^\omega} \not\subseteq \sem{L^\omega}
$, $L^\omega$~provides proof that the protocol is not susceptible to
known-plaintext attacks where arbitrary known nonces
are available to the intruder encrypted with the same key.\footnote{The property could be extended to cover composite known-plaintext messages (of some maximum size~$s$), and the generation of a sufficiently high number~$N$ of plaintext-encrypted pair with the same key.}
The algorithms for inclusion and $\posthat$
can be used to check inductiveness given a candidate invariant such as $L$.
This leaves open the question of
how to efficiently generate candidates.
The algorithm of \cref{th:dwdcl-decidable} can in principle be used to
enumerate all candidate invariants, with an impractically high complexity.
Luckily, a more directed inference of invariants can be obtained
by a \emph{widening operator}, in the style of~\cite{ZuffereyWH12};
in fact, the invariant~$L$ was automatically inferred from~$P_0$
using the widening of our prototype tool.
The basic observation behind invariant inference,
is that from a sequence of transitions
$ P_1 \redto^* P_2 $ with $P_1 \kembed P_2$,
one can deduce that the same sequence can be simulated from $P_2$
(by~\cref{th:kembed-simulation})
obtaining a~$P_3$ with $P_2 \kembed P_3$ and so on.
The embedding between~$P_1$ and~$P_2$,
is justified by
$
  P_1 \kcongr \new\vec{x}.(\out{\Gamma_1} \parallel Q)
$, $
  P_2 \kcongr \new\vec{x}.(
    \out{\Gamma_1}\parallel
    Q \parallel
    P
  )
$;
we can extrapolate the difference $P$ and
accelerate the sequence of transitions
by constructing the limit $ \new\vec{x}.(
    \out{\Gamma_1}\parallel
    Q \parallel
    P^\omega
  )
$.
This operation can be extended to limits.
The end product is a finite union of limits which is inductive by construction.
This procedure requires exploration of transitions and
many inclusion checks, a costly combination.
To obtain a more practical algorithm, we devised two techniques:
inductiveness checks through ``incorporation'', and a coarser widening.

Consider the inductiveness check
$\sem[\big]{\posthat_\Defs^s(L)} \subseteq \sem{L}$
implemented by checking that,
for each transition considered by $\posthat$
the resulting limit $L'$ is included by~$L$.
We observe that $L'$~and~$L$ will share the context of the
transition: $L$ can be rewritten to $C[\p Q[\vec{M}]]$
and~$L'$ to $C[P]$ for some~$P$.
To prove inclusion of $C[P]$ in $C[\p Q[\vec{M}]]$
it is then sufficient to show that $P$ is embedded in $C[\p Q[\vec{M}]]$.
We call this check an \emph{incorporation} of $P$ in~$C$.
See \cref{app:incorporation} for an example.
Although incomplete in general, incorporation can prove inductiveness
in many practical examples, in a remarkably faster way.

There are cases, however, where the incorporation check fails on inductive invariants.
Consider for example an inductive invariant represented by the union
of two incomparable limits $L_1,L_2$.
Suppose that for some $P\in\sem{L_1}$ there is
$P'$ with $P\redto P'\in\sem{L_2}\setminus\sem{L_1}$.
Then, incorporation of $P'$ in $L_1$ would fail.
To side-step this problem we replace union of limits with parallel composition:
$\sem{L_1}\union\sem{L_2} \subseteq \sem{L_1\parallel L_2}$
by downward closure of $\sem{\hole}$.
Using this over-approximation, we can try to aim for an inductive invariant
which consists of exactly one limit,
and for which the incorporation check suffices.
This approximation is incomplete:
some protocols require inductive invariants
consisting of unions of incomparable limits.
 \section{Evaluation}
\label{sec:discussion}

\begin{table}[tb]
  \caption{Experimental results. Columns:
      \textbf{Infer}ence of invariant
        fully automatic~(\fullyauto) or
        interactive~(\interact);
      \textbf{C}heck of inductiveness;
      \textbf{S}ecrecy
        proved~(\checker),
        not holding~(\notsecret),
        not modelled~(\notmodelled).
  }
  \label{fig:tool-results}
  \centering\footnotesize \begin{tabular}{lr@{\ \ }c@{\quad}r@{\quad}c@{\enspace}|@{\enspace}lr@{\ \ }c@{\quad}r@{\quad}c@{\enspace}|@{\enspace}lr@{\ \ }c@{\quad}r@{\quad}c@{\enspace}}
\toprule\strut
\textbf{Name}  & &\llap{\textbf{Infer}} & \hfill\textbf{C}\hfill\null & \textbf{S} &
\textbf{Name}  & &\llap{\textbf{Infer}} & \hfill\textbf{C}\hfill\null & \textbf{S} &
\textbf{Name}  & &\llap{\textbf{Infer}} & \hfill\textbf{C}\hfill\null & \textbf{S} \\
\midrule\strut
Ex.\ref{ex:running}&  4.4s  & \fullyauto &  1.8s & \checker     &
NHS                &   5.0s & \fullyauto &  1.6s & \notmodelled &
YAH                &   7.8s & \fullyauto &  2.5s & \notmodelled \\
OR                 &   3.4s & \fullyauto &  1.9s & \notmodelled &
NHSs               &   6.8s & \fullyauto &  1.7s & \checker     &
YAHs1              &  12.3s & \fullyauto &  2.6s & \checker     \\
ORl                &  25.0s & \fullyauto &  3.5s & \notsecret   &
NHSr               &  90.9s & \interact  & 20.8s & \notsecret   &
YAHs2              &   8.8s & \fullyauto &  2.0s & \checker    \\
ORa	        	   &  13.8s & \interact  &  5.0s & \notmodelled &
KSL                &  37.0s & \fullyauto &  9.8s & \checker     &
YAHlk              &  11.9s & \interact  &  17.3s & \checker        \\
ORs                &   9.8s & \fullyauto &  2.0s & \checker     &
KSLr               & 200.6s & \fullyauto & 31.4s & \checker     &
ARPC               &  0.5s  & \fullyauto &  0.1s & \checker     \\
\bottomrule
\end{tabular}
 \end{table}

We built a proof-of-concept tool
implementing limit inclusion, inductivity check, incorporation
and a coarse widening.
Currently, the tool only supports symmetric encryption,
but we plan to add support for asymmetric encryption, signatures and hashing.
The source code and all the test protocol models are available at
\cite{lemma9}
where we also provide a tutorial-style explanation of the methodology.
We summarise our experiments\footnote{Tests run with Python~2.7, z3-solver~v4.8, 8GB~RAM, Intel CPU~i5, on Linux.}
in \cref{fig:tool-results}.
The tool is instructed to compute, using the widening, an invariant for the provided model,
under given message size assumptions.
When an invariant can be found, it represents a proof
that the model is depth-bounded.
If the inferred invariant is leak-free, secrecy is also proven.
We model a number of well-known protocols under various threat scenarios (e.g.~with or without leak of old keys):
Needham-Schr\"{o}der~(NHS), Otway-Rees~(OR), Kehne-Sch\"{o}nw\"{a}lder-Landend\"{o}rfer~(KSL), Yahalom~(YAH) and Andrew RPC~(ARPC).
For any example containing a problematic encryption oracle pattern (see \cref{ex:encryption-oracle})
we constrain the input message of the oracle to be a nonce (message of size~1).
With the exception of NHSr, ORa and YAHlk, all the invariants were obtained fully automatically.
For YAHlk we had to combine two widened limits (the timing is the sum of the time spent computing each limit);
for NHSr and ORa we had to tweak a partially widened limit manually to make it inductive.
To simulate using invariants as correctness certificates, for each example we re-checked them for inductiveness.

\subsection{Limitations}
\label{sec:limitations}

\subparagraph{Message size bounds}
Bounding message size in the analysis is not always acceptable,
as it may miss type confusion attacks.
Such patterns arise, for example, in XOR-based protocols,
where considering only typed runs is not appropriate.
This is problematic because in most practical examples
bigger size bounds lead to unboundedness in depth,
as illustrated by \cref{ex:encryption-oracle}.
However, not all is lost:
our limits can be extended to provide invariants with unbounded message size.
The idea, detailed in \cref{app:extended-limits},
is to first analyse a protocol with a size bound that ensures depth-boundedness,
obtaining an inductive invariant as one or more limits;
then we annotate such limits with ``wildcards'' ($\wildcard$)
on occurrences of names.
A wildcarded name $a^\wildcard$ stands for an arbitrary message
(of arbitrary size).
By introducing the appropriate wildcards, one can obtain an inductive invariant
for the protocol with no message size bounds.
For \cref{ex:encryption-oracle},
$
  \p E[k] \parallel \bigl( \new x. \out{\enc{x}{k}}\bigr)^\omega
$
is an inductive invariant if we restrict $x$ to be of size 1.
Since~$x$ is never inspected , injecting larger terms in it would not lead to new behaviour.
We can therefore generalise the limit to
$
  \p E[k] \parallel \bigl( \new x. \out{\enc{x^\wildcard}{k}}\bigr)^\omega
$
which is an inductive invariant for the protocol with no size bounds.

Since verification with no size bounds is undecidable,
this extension is by necessity incomplete:
there are downward-closed sets that cannot be represented by extended limits.
However, since protocols typically achieve resistance to type confusion attacks
by making sure that messages of the wrong type
are discarded by honest principals,
this extension is very effective.
In fact, all the size bounds in our benchmarks can be lifted using this technique.

\subparagraph{Inherent unboundedness}
There are two ways a protocol may fall outside of our class.
The first is when unboundedly many participants in a session form a
ring/list topology, like the recursive protocol of~\mbox{\cite[\textsection 6]{Paulson98}}.
One can provide a partial solution by using an under-approximate model
with a ring/list of fixed size,
or an over-approximate model by using an unbounded star topology.
The second way is when the intruder can produce irreducible encryption chains
and the participants would inspect them generating new behaviour.
In such cases not even extended limits can help.
We deem this situation unlikely to be desirable in a realistic protocol.
 \section{Conclusions and Future Work}
\label{sec:conclusions}

We presented a theory of decidable inductive invariants
for depth-bounded cryptographic protocols.
We showed how one can infer inductive invariants
and evaluated the approach through a prototype implementation.

From a theoretical perspective, it would be interesting to determine precise complexity bounds for inclusion, for general intruder models.
We can show that~$\kembed$ is NP-complete
for any intruder model that has polynomially decidable~$\kleq$~\cite{Stutz19}.

A direction for further improvement is
extending the class of supported properties.
In particular, we plan to study how invariants can be used
to automatically prove diff\hyph-equivalence~\cite{Blanchet08}
without bounding sessions/nonces.

Finally,
we intend to explore ways in which our invariants can be integrated
in existing tools such as ProVerif and Tamarin.
For instance, Tamarin performs a backwards search to find possible attacks.
Our invariants could provide a pruning technique to avoid exploring
paths that are unreachable from the initial state.
Similar combinations of forward and backward search have been shown to improve
performance dramatically for analyses of infinite-state systems
such as Petri~nets~\cite{BlondinFHH16}.

\clearpage
\appendix
{\LARGE\bfseries\noindent Appendix\par}

\section{Proof of Proposition~\ref{prop:sigsymm-axioms}}
We check the model $ \SigSymm $ is an effective intruder model.
\Cref{rule:id} implies \eqref{axiom:id}; it is easy to check that any derivation for $\Gamma\deriv M$ is also a derivation (of same depth) for $\Gamma,\Gamma'\deriv M$, which implies \eqref{axiom:mon};
\eqref{axiom:cut} is a consequence of
the admissibility of the \textsc{cut} rule proven in~\cite{tiu10lmcs}.
The --\textsubscript{R} rules directly entail \eqref{axiom:constr}.
All rules are invariant under renamings so \eqref{axiom:alpha} holds.
\eqref{axiom:relevancy} can be proven by a straightforward induction on the depth of derivations.
Since the premises in every rule involve only sub-terms of the consequence,
derivability is decidable.
\section{Proof of Proposition~\ref{prop:kn-embed-vs-kn-congr}}
We want to prove
  $ P_1 \kcongr P_2 $
    if and only if
  $ P_1 \kembed P_2 $ and $ P_2 \kembed P_1 $.
The ``only if'' direction is trivial.
The ``if'' direction is proved as follows.
Let $\stdf(P_i) \aeq \new\vec{x}_i.(\out{\Gamma_i} \parallel Q_i)$
for $i=1,2$.
It is easy to check that
$ P_1 \kembed P_2 $ if and only if
$ \Gamma_1 \kleq \Gamma_2 $,
$ \vec{x}_2 = \vec{x}_1\vec{y}_2 $ for some $\vec{y}_2$, and
$ Q_2 = Q_1 \parallel Q'_2 $
for some parallel composition of process calls $Q'_2$.
Similarly,
$ P_2 \kembed P_1 $ if and only if
$ \Gamma_2 \kleq \Gamma_1 $,
$ \vec{x}_1 = \vec{x}_2\vec{y}_1 $, and
$ Q_1 = Q_2 \parallel Q'_1 $
for some $\vec{y}_1, Q'_1$.
We get $\Gamma_1 \keq \Gamma_2$, $\vec{x}_1 = \vec{x}_2$ and $Q_1 = Q_2$ as required
to prove $P_1 \kcongr P_2$.

\section{Auxiliary Results}

\subsection{Properties of Knowledge}

\begin{lemma}
\label{lm:kleq-add-same}
  Let
  $\Gamma_1, \Gamma_2, \Gamma \in \Know$ be
  such that $\Gamma_1 \kleq \Gamma_2$.
  Then, $\Gamma, \Gamma_1 \kleq \Gamma, \Gamma_2$.
\end{lemma}

\begin{proof}
  Let $\Gamma_1 = \set{M_1,\dots,M_n}$.
We have to show that for all $N$,
if $\Gamma, \Gamma_1 \deriv N$ then $\Gamma, \Gamma_2 \deriv N$.
We apply~\eqref{axiom:cut} $n$ times obtaining
\begingroup\small
\[
  \infer*[right=\ref{axiom:cut}]{
    \Gamma, \Gamma_2 \deriv M_1\and
    \infer*[right=\ref{axiom:cut}]{
      \Gamma, \Gamma_2, M_1 \deriv M_2\and
      \infer*[right=\ref{axiom:cut}]{
        \infer*[right=\ref{axiom:cut},
                left=\llap{\raisebox{-1em}{$\iddots\quad$}}]{
          \Gamma, \Gamma_2, (\Gamma_1\setminus M_n) \deriv M_n\and
\Gamma,\Gamma_1,\Gamma_2\deriv N
        }{\vdots}
      }{
        \Gamma, \Gamma_2, M_1, M_2 \deriv N
      }
    }{
      \Gamma, \Gamma_2, M_1 \deriv N
    }
  }{
    \Gamma, \Gamma_2\deriv N
  }
\]
\endgroup
From $\Gamma_1 \kleq \Gamma_2$ we have
$\forall i \leq n\st \Gamma_2 \deriv M_i$,
which implies, by~\eqref{axiom:mon},
all the left-most premises.
By~\eqref{axiom:mon}, from $\Gamma,\Gamma_1 \deriv N$
we know $\Gamma,\Gamma_1,\Gamma_2 \deriv N$,
which completes the derivation
showing $\Gamma, \Gamma_2\deriv N$.
 \end{proof}

\begin{corollary}
\label{lm:kleq-add-kleq}
  Let $\Delta_1, \Delta_2, \Gamma_1, \Gamma_2 \in \Know$ be such that
  $\Delta_1 \kleq \Delta_2$ and $\Gamma_1 \kleq \Gamma_2$.
  Then, $\Delta_1, \Gamma_1 \kleq \Delta_2, \Gamma_2$.
\end{corollary}

\begin{proof}
  Easy corollary of \cref{lm:kleq-add-same}:
  $ \Gamma_1,\Delta_1 \kleq \Gamma_2, \Delta_1 \kleq \Gamma_2, \Delta_2 $.
\end{proof}

\begin{lemma}
 \label{lm:persistence-of-knowledge}
$\out{\Gamma} \parallel \out{\Gamma} \kcongr \out{\Gamma}$,
 for all~$\Gamma\in\Know$.
\end{lemma}

\begin{proof}
Let $\Gamma \in \Know$.
By definition, $\out{\Gamma} \parallel \out{\Gamma} \kcongr \out{\Gamma}$ iff $\Gamma, \Gamma \keq \Gamma$.
Notice that these are actually sets, hence $\Gamma, \Gamma = \Gamma$ and the claim follows by reflexivity of $\keq$.
\end{proof}

\subsection{Properties of Knowledge Embedding}

\begin{corollary}
\label{cor:kembed-add-mess}
  Let $\Gamma\in\Know$ and
  $P_1, P_2 \in \Proc$ with
  $\stdf(P_i) = \new \vec{x}_i.(\out{\Gamma_i} \parallel Q_i)$
  for $i \in \set{1,2}$ and
  $P_1 \kembed P_2$.
  Then, $\out{\Gamma} \parallel P_1 \kembed \out{\Gamma} \parallel P_2$.
\end{corollary}
\begin{proof}
We only have to compare knowledge as $Q_1 \kembed Q_2$ holds by assumption.
$\Gamma, \Gamma_1 \kleq \Gamma, \Gamma_2$ follows by \cref{lm:kleq-add-same}.
\end{proof}

\begin{lemma}
\label{lm:kembed-lground-proc}
  Let $P_1, P_2\in\Proc$
  and $n \in \Nat$.
  If $P_1 \kembed P_2$,
  then $P_1^n \kembed P_2^n$.
\end{lemma}
\begin{proof}
For every single instance of both processes use the matching provided by $P_1 \kembed P_2$ for names and process calls. The knowledge embedding part of the goal can be obtained by $n$~applications of \cref{lm:kleq-add-kleq}.
\end{proof}

\subsection{Properties of Grounding}

\begin{lemma}
  \label{lm:lground-in-sem}
  For every $n \in \Nat$,
  $\lground{L}{n} \in \instn{L}$.
\end{lemma}

\begin{lemma}
  \label{lm:less-lground-kembed}
  For every $n \leq m \in \Nat$,
  $\lground{L}{n} \kembed \lground{L}{m}$.
\end{lemma}

\begin{lemma}
  \label{lm:lground-kembed}
  For every $P \in \sem{L}$ there exists an $n\in\Nat$ such that
  $P \kembed \lground{L}{n}$.
\end{lemma}

\begin{proof}[Proof of \cref{lm:lground-in-sem,lm:less-lground-kembed,lm:lground-kembed}]
All claims follow by an easy induction on the structure of~$L$.
\end{proof}

\begin{lemma}
\label{lm:kembed-limit-mult}
  For any $L \in \Limits$, and $m, n \in \Nat$,
  we have
  $(\lground{L^\omega}{n})^{m} \kembed \lground{L^\omega}{m * n}$.
\end{lemma}
\begin{proof}
We assume that $n > 0$ as the claim trivially holds if not.
\[ (\lground{L^\omega}{n})^{m}
                            = ((\lground{L}{n})^{n})^{m}
                            = (\lground{L}{n})^{m * n} \]
\[ \lground{L^\omega}{m * n}
                        = (\lground{L}{m * n})^{m * n}                  \]
By \cref{lm:less-lground-kembed}, we know that $\lground{L}{m} \kembed \lground{L}{m * n}$ as $n > 0$.
The claim follows by  \cref{lm:kembed-lground-proc}.
\end{proof}

\begin{lemma}
 \label{lm:lground-inc-indices}
$
  \sem{L_1} \subseteq \sem{L_2}
    \iff
  \forall n \in \Nat\st
    \exists m \in \Nat\st
     { \lground{L_1}{n} \kembed \lground{L_2}{m}}
$.
\end{lemma}
\begin{proof}
 First, let $n \in \Nat$. We compute the expansion $\lground{L_1}{n} \in \sem{L_1}$ by \cref{lm:lground-in-sem}. By assumption, $\lground{L_1}{n} \in \sem{L_2}$ and by Lemma \ref{lm:lground-kembed} there exists an $m$ such that $\lground{L_1}{n} \kembed \lground{L_2}{m}$ as required.
    Second, let $P \in \sem{L_1}$. We need to show that $P \in \sem{L_2}$. By Lemma \ref{lm:lground-kembed}, there exists an expansion for some $n$ with $P \kembed \lground{L_1}{n}$. By assumption, there is an $m$ such that $\lground{L_1}{n} \kembed \lground{L_2}{m}$. By transitivity of $\kembed$, we know that $P \kembed \lground{L_2}{m}$ which is in $\sem{L_2}$. By downward closure, the claim follows.
\end{proof}

\subsection{Properties of Extension}

\begin{lemma}
\label{lm:ext-same-sem}
For every $L \in \Limits$ and $n \in \Nat$,
$\sem{L} = \sem{\lext{L}{n}}$.
\end{lemma}
\begin{proof}
We prove the fact by two inclusions.
The inclusion $\sem{L} \subseteq \sem{\lext{L}{n}}$ follows trivially by definition.
For the reverse inclusion, we apply \cref{lm:lground-inc-indices}.
Given an arbitrary $n'$,
we need to show that there is $m'$ such that
$\lground{\lext{L}{n}}{n'} \kembed \lground{L}{m'}$.
It is straightforward to verify that $m' = n \cdot {n'}$ suffices.
\end{proof}

\begin{lemma}
\label{lm:ext-n-ground-0-eq-ground-n}
  For every $L \in \Limits$ and $n \in \Nat$,
  $\lground{L}{n} = \lground{\lext{L}{n}}{0}$.
\end{lemma}
\begin{proof}
Straightforward induction on the structure of~$L$.
\end{proof}

\section{Soundness and Completeness of Limits}
\label{app:wqo-limits}

\subsection{Depth-Bounded Processes are Well-Quasi-Ordered}

\subparagraph{Multisets}
Given a set $X$, the multisets over $X$,
are all the functions $\mu\from X \to \Nat$ such that the set
$
  \items(\mu) \is \set{(x, i) \in X \times \Nat | 0 < i \leq \mu(x)}
$
is finite.
We denote by $\multiset(X)$ the set of all multisets over $X$.
We write $\mset{x_1,\dots,x_n}$ for the multiset with elements $x_1,\dots,x_n$
and use $\emptyset$ for the empty multiset.
The union two multisets $\mu_1 \munion \mu_2$ is the pointwise lifting of addition.
Assume $(X, \qoleq)$ is a quasi order;
its \emph{multiset extension} is the quasi order $\membed$ over $\multiset(X)$
defined as follows;
for two multisets $\mu_1,\mu_2 \in \multiset(X)$,
  $\mu_1 \membed \mu_2$
    holds just if
  there is an injective function $f \from \items(\mu_1) \to \items(\mu_2)$
  such that for each $(x,i)\in \items(\mu_1)$,
  if $f(x,i) = (y,j)$ then $x \qoleq y$.
When $X$ is a set ordered by equality,
$\membed$ coincides with multiset inclusion.
If $(X, \qoleq)$ is a wqo
then $(\multiset(X), \membed)$ is a wqo~\cite{FinkelG09}.

\subparagraph{Forests}
We define a domain of forests $\Forests_{s,k}^X$
with sequential processes in $\basic{X}$ as leaves:
\begin{align*}
  \basic{X} &\is
  \mathrlap{
    \set{ \p Q[\vec{M}] | Q \in \ProcN,
                          \arity(\p{Q}) = \card{\vec{M}},
                          \vec{M} \subseteq \Msg_s^X }
    \union \Msg_s^X
  }
  \\
  \Forests_{s,0}^X &\is \multiset(\basic[s]{X})
  &\quad \Forests_{s,k+1}^X &\is
    \multiset\left(
      \basic[s]{X} \dunion
      \Forests[s,k]^{X \union \set{x_{k+1}}}
    \right)
    \qquad
    (\text{assuming } x_{k+1} \not\in X)
\end{align*}
The set $\basic{X}$ contains all sequential processes of
$\Depth[X]{s,k}$.
In the base case, forests of height~0,
$\Forests_{s,0}^X$ is simply a multiset of sequential processes.
Forests of height $k+1$ are represented in $\Forests_{s,k+1}^X$
by a multiset of sequential processes and of subforests of height $k$.

For any given $s \in \Nat$ and finite $X \subseteq \Names$,
$\basic[s]{X}$ is a finite set and thus forms a wqo with equality.
With $\fembed$ (forest embedding)
we denote the qo over $\Forests_{s,k}^X$ defined as follows:
on $\Forests_{s,0}^X$ it coincides with
  the multiset extension of equality on $\basic{X}$;
on $\Forests_{s,k+1}^X$ it coincides with
  the multiset extension of
  the disjoint union of equality on $\basic[s]{X}$, and the
  forest embedding on $\Forests_{s,k}^{X \union \set{x_{k+1}}}$.
By iterating the result on multisets, we obtain that
for any $s,k \in \Nat$ and finite $X\subseteq\Names$,
$(\Forests_{s,k}^X, \fembed)$ is a~wqo.

The function $\Enc[F]{\hole}_k$,
transforms a process $P$ with $ \nestr(P) \leq k $ into its forest encoding,
that is a forest that faithfully represents the structure of $P$.

\begin{definition}[Forest encoding]
\label{def:forest-enc}
  We define the function $\Enc[F]{\hole}_k$ which transforms a process~$P$
  with ${ \nestr(P) \leq k }$ into a forest:
\[
    \Enc[F]{P}_k \is
      \begin{cases}
        \emptyset           \CASE P = \zero               \\
        \mset{P}            \CASE P \text{ is sequential} \\
        \mset{ \Enc[F]{Q\subst{x -> x_k}}_{k-1} }
                            \CASE P = \new x.Q            \\
        \Enc[F]{Q_1}_k \munion \Enc[F]{Q_2}_k
                            \CASE P = Q_1 \parallel Q_2
      \end{cases}
  \]
  assuming
  $ \set{x_1,\dots, x_k} \inters
      (\boundnames(P) \union \freenames(P)) = \emptyset$;
  when this assumption is not met because of bound names,
  we implicitly \pre\alpha-rename the term before applying the definition.
  When $\nestr(P) > k$, $\Enc[F]{P}_k$ is undefined.
\end{definition}

\Cref{fig:forest-encoding} illustrates forest encodings.
The forests in $\Forests_{s,k}^Y$ have sequential processes as leaves
(with messages of size smaller than $s$),
and free names $Y \dunion \set{x_1,\dots,x_k}$.
\Cref{lm:name-reuse} allows us to use only~$k$ distinct bound names:
this makes the forests finitely labelled.

\begin{figure}[tb]
  \adjustfigure\footnotesize \begin{tikzpicture}[
  AST,
  forests distance=3,
  level 1/.style={sibling distance=1cm},
  level 2/.style={sibling distance=1cm},
  baseline=(n.south)
]

\path

  node (f1) {$x_2$}
    child { node {$x_1$}
      child { node (n) {\scalebox{.7}{$\enc{x_1}{k}$}}}
      child { node {\scalebox{.7}{$\p A[x_2,x_1]$}}}
    }
    child { node {\scalebox{.7}{$\p B[x_2,m]$}}}
    child { node {$x_1$}
      child { node {\scalebox{.7}{$\p C[x_2,x_1]$}}}
    }

  ++(2,0)
  node (f2) {$x_2$}
    child { node {\scalebox{.7}{$\enc{x_2}{k}$}} }

  ;
  \node[above=.5 of n,xshift=-5pt] {$\Enc[F]{P}_2$};

\end{tikzpicture}   \begin{minipage}[b]{.5\textwidth}
    \begin{align*}
      P &= P_1 \parallel P_2 &
      P_1 &= \new a.(
          P_3 \parallel
          \p B[a,m]
          \parallel
          P_4
        )\\
      P_2 &= \new b.\out{\enc{b}{k}}&
      P_3 &= \new c.(
            \out{\enc{c}{k}} \parallel
            \p A[a,c]
          )\\
      P_4 &=\new d.\p C[a,d] &
      \mathrlap{\Enc[F]{P}_2 \in \Forests_{2,2}^{\set{m,k}}} &
    \end{align*}
  \end{minipage}
\caption{Forest encoding of $P$.
    Restrictions at nesting level $k$ are renamed canonically to $x_{2-k}$.
  }
  \label{fig:forest-encoding}
\end{figure}
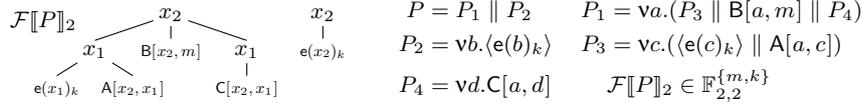

\begin{lemma}
  \label{lm:fenc-embeds}
  If
    $Q \in \Size[Y]{s}$ and
    $\nestr(Q) \leq k$,
  then
    $\Enc[F]{Q}_k \in \Forests_{s,k}^Y$.
\end{lemma}
\begin{proof}
  By easy induction on the structure of $Q$.
\end{proof}

\begin{lemma}
  \label{lm:fenc-pres-order}
  Assume
    $Q_1,Q_2 \in \Size[Y]{s}$ with
    $\nestr(Q_1) \leq k$ and $\nestr(Q_2) \leq k$.
  Then
    $\Enc[F]{Q_1}_k \fembed \Enc[F]{Q_2}_k$
  implies
    $Q_1 \kembed Q_2$.
\end{lemma}
\begin{proof}
  First we strengthen the statement:
we can prove that if
  $\Enc[F]{Q_1}_k \fembed \Enc[F]{Q_2}_k$
then $\stdf(Q_1) = \new \vec{y}. R$ and
     $\stdf(Q_2) = \new \vec{y}.\new \vec{z}. (R \parallel R')$,
which clearly implies $Q_1 \kembed Q_2$.

We proceed by induction on $k$.
The base case is a special case of the induction step, so let us consider the latter first.
Assume
  $\nestr(Q_1), \nestr(Q_2) \leq k+1$ and
  $ \phi_1 = \Enc[F]{Q_1}_{k+1} \fembed \Enc[F]{Q_2}_{k+1} = \phi_2 $,
and let
  $ B = \items(\phi_1) \inters (\basic[Y]{s} \times \Nat) $ and
  $ C = \items(\phi_1) \setminus B $
  (note that $C\subseteq (\Forests_{s,k}^{Y \union \set{x_{k+1}}} \times \Nat)$).
Then, by definition, there is an injective function
  $ f \from \items(\phi_1) \to \items(\phi_2) $
such that for each $(P, i) \in B$,
  $ f(P,i) = (P,j) $ for some $j$;
moreover, for each $(\phi, i) \in C$,
  if $ f(\phi, i) = (\phi',j) $ then $ \phi \fembed \phi' $.
By definition of $\Enc[F]{\hole}$ we know that
for each ${(\phi, i) \in C}$,
there exist a subterm of $Q_1$ \pre\alpha-equivalent to $\new x_{k+1}.P_{\phi}$
for some $P_{\phi}$ with ${\nestr(P_{\phi}) \leq \nestr(Q_1)-1 \leq k}$
and $\phi = \Enc[F]{P_{\phi}}_k $.
Similarly for $Q_2$, we have
that if $f(\phi,i) = (\phi',j)$ then there is
a subterm of $Q_2$ \pre\alpha-equivalent to $\new x_{k+1}.P_{f(\phi,i)}$
for some $P_{f(\phi,i)}$ with
$\Enc[F]{P_{f(\phi,i)}}_k = \phi'$ and
$\nestr(P_{f(\phi,i)}) \leq \nestr(Q_2)-1 \leq k$.

We can therefore apply the induction hypothesis and get
$ \phi \fembed \phi' $ implies that
  $\stdf(P_{\phi}) =
    \new \vec{y}_{\phi}. R_{\phi}$ and
  $\stdf(P_{f(\phi,i)}) =
    \new \vec{y}_{\phi}.\new\vec{z}_{\phi}.(R_{\phi} \parallel R_{f(\phi,i)})$.
As a consequence we have
\begin{align*}Q_1 &\kcongr \bigl(
    \Parallel_{(P,i) \in B} P \parallel
    \Parallel_{(\phi,i) \in C}
      (\new x_{k+1}.\new \vec{y}_{\phi}.R_{\phi})
  \bigr)
  \\
  Q_2 &\kcongr \bigl(
    \Parallel_{(P,i) \in B} P \parallel
    \Parallel_{(\phi,i) \in C}
      (\new x_{k+1}.\new \vec{y}_{\phi}.\new \vec{z}_{\phi}.
        (R_{\phi} \parallel R_{f(\phi,i)}))
    \parallel R'
  \bigr)
\end{align*}
which clearly entails the claim,
  by application of \pre\alpha-renaming and scope extrusion
  to get~the~two standard forms.
In the base case, $C = \emptyset$ from which the claim follows straightforwardly. \end{proof}

\begin{theorem}
  \label{th:bounded-proc-wqo}
  For every
    finite $Y \subseteq \Names$ and
    $s,k\in \Nat$,
  $(\Depth[Y]{s,k}, \kembed)$
  is a wqo.
\end{theorem}

\begin{proof}
  Let $P_1,P_2,\dots$ be an infinite sequence of processes
  in $\Depth[Y]{s,k}$.
  By definition,
  for all $i\in\Nat$ there exists a process $Q_i \in \Size{s}$
  such that $Q_i \kcongr P_i$ and $\nestr(Q_i) \leq k$,
  which implies by \cref{lm:fenc-embeds} that
  ${\Enc[F]{Q_i}_k \in \Forests_{s,k}^Y}$.
  The sequence $\Enc[F]{Q_1}_k, \Enc[F]{Q_2}_k,\dots$
  must then be a good sequence since $(\Forests_{s,k}^Y, \fembed)$ is a wqo.
  Therefore, there are $i,j\in\Nat$ with $i<j$, such that
  $\Enc[F]{Q_i}_k \fembed \Enc[F]{Q_j}_k$,
  which by \cref{lm:fenc-pres-order} implies $Q_i \kembed Q_j$.
  We therefore have $P_i \kembed P_j$ which proves
  that the sequence $P_1,P_2,\dots$ is good.
\end{proof}

\subsection{Soundness and Completeness (Theorem~\ref{th:limit-sound-complete})}

We prove the two directions of \cref{th:limit-sound-complete}
independently.

\begin{theorem}
  \label{th:limit-is-idl}
  For every $L \in \Limits_{s,k}$,
  the set $\sem{L}$ is an ideal of $(\Depth{s,k}, \kembed)$.
\end{theorem}

\begin{proof}
We want to prove that
for every $L \in \Limits_{s,k}$, we have
$\sem{L} \subseteq \Depth{s,k}$
and
$\sem{L}$ is \pre\kembed-downward closed and directed.

Clearly  $\sem{L} \subseteq \Depth{s,k}$
since every $P \in \instn{L}$ has $\nestr(P) \leq \nestr(L) \leq k$ and its active messages are the messages of $L$ up to renaming (which does not alter size).

$\sem{L}$ is downward closed by definition.
Proving that $\sem{L}$ is directed requires showing that
for all $P_1,P_2 \in \sem{L}$, there is $Q \in \sem{L}$ such that
$P_1 \kembed Q$ and $P_2 \kembed Q$.
Such $Q$ can be constructed by applying \cref{lm:lground-kembed} to $P_1$ and $P_2$, obtaining $n_1$ and $n_2$ such that $P_i \kembed \lground{L}{n_i}$, for $i=1,2$.
Then $Q = \lground{L}{\max(n_1,n_2)}$ is the required process by \cref{lm:less-lground-kembed}. \end{proof}

\begin{theorem}
  \label{th:idl-is-limit}
  Every ideal in $\Idl{\Depth{s,k}}$ is the denotation
  of some limit $L \in \Limits_{s,k}$.
\end{theorem}

\begin{proof}
The idea of the proof is to use the forest encoding $\Enc[F]{\hole}$
to apply known facts about the ideal completion for multisets to our setting.
We take an ideal of $\Depth{s,k}$ and show it
corresponds to a downward closed set of $\Forests_{s,k}$, via forest encoding.
For multisets we know exactly how to represent ideals~\cite{FinkelG09}, i.e.~\pre\miter-products. These products are expressions that can be readily seen as limits.
Let us recall the result on multisets first and then proceed with the proof.

Let $X$ be a wqo with a finite representation of ideals $\mathbb{I}(X)$.
Then we can finitely represent ideals of $\multiset(X)$ by
expressions of the form
\[
  p =
    I_1^\miter \odot \cdots \odot I_n^\miter \odot
    J_1^? \odot \cdots \odot J_m^?
\]
where $I_i,J_j \in \mathbb{I}(X)$ for all $1\leq i\leq n$, $1\leq j\leq m$.
When $n,m=0$ we write $p = \epsilon$.
Such expressions are called \pre\miter-products.
Their denotations are the ideals of $\multiset(X)$, via the map
\begin{align*}
  \sem{\epsilon} &\is \set{\emptyset} &
  \sem{p_1 \odot p_2} &\is
    \set{ \mu_1 \munion \mu_2 | \mu_1 \in \sem{p_1}, \mu_2 \in \sem{p_2}} \\
  \sem{I^\miter} & \is \multiset(\sem{I}) &
  \sem{I^?} & \is \Set{ \mset{x} | x \in \sem{I} } \union \set{\emptyset}
\end{align*}

In a wqo, downward closed sets are finite unions of ideals,
so any downward closed set of $\multiset(X)$ can be described
as the union of the denotation of finitely many \pre\miter-products.
For the details\footnote{here we use a variation of \pre\miter-products that can be seen equivalent to the one in~\cite{FinkelG09} by using $ (C_1 \union C_2)^\miter = C_1^\miter \odot C_2^\miter $.}
see~\cite{FinkelG09}.

Another well-known property of ideals we are going to use is the following.
Let $D$ be a downward closed set, then the following are equivalent:
\begin{enumerate}[(1)]
  \item $D$ is an ideal;
  \item for all downward closed sets $D_1, D_2 \subseteq X$,\\
        if $D \subseteq D_1 \union D_2$
        then $D \subseteq D_1$ or $D \subseteq D_2$;
  \item for all downward closed sets $D_1, D_2 \subseteq X$,\\
        if $D = D_1 \union D_2$ then $D = D_1$ or $D = D_2$.
\end{enumerate}

\noindent Now we turn to our theorem.
Let
  $D \subseteq \Depth{s,k}$, we define
\[
  \Enc[F]{D}_k \is \set{\Enc[F]{P} | P \in D, \nestr(P) \leq k}.
\]
By definition, $\Enc[F]{D}_k \subseteq \Forests_{s,k}$.

We also define the opposite mapping,
from forests to processes in the expected way.
Let $\phi \in \Forests_{s,k}^X$,
    $B(\phi) \is \supp(\phi) \inters \basic[s]{X}$ and
    $C(\phi) \is \supp(\phi) \setminus B(\phi)$.
Then
\[
  \Enc[P]{\phi}_k \is
    \left(
      \Parallel_{P \in B(\phi)} P^{\phi(P)}
      \parallel
      \Parallel_{\psi \in C(\phi)} (\new x_k.\Enc[P]{\psi}_{k-1})^{\phi(\psi)}
    \right)
\]
Since when $k=0$, $C(\phi)=\emptyset$, the expression is well-defined.
We have $\Enc[P]{\phi}_k \subseteq \Depth{s,k}$
and\footnote{Note the implicit \pre\alpha-renaming in the application of $\Enc[F]{\hole}$.}
$\Enc*[\big]{F}{\Enc[P]{\phi}_k}_k = \phi$ and
$\Enc*[\big]{P}{\Enc[F]{Q}_k}_k \congr Q$ (when defined).

\begin{lemma}
  \label{lm:fenc-pres-dwcl}
  For any \pre\kembed-downward closed
  $D \subseteq \Depth{s,k}$,
  $\Enc[F]{D}_k$ is \pre\fembed-downward closed.
\end{lemma}
\begin{proof}
  First note that for each $P \in D$ there is a $Q \kcongr P$ in $D$
  such that $\Enc[F]{Q} \in \Enc[F]{D}_k$.

  Towards a contradiction, assume there is a forest $\phi \in \Forests_{s,k}$
  such that there is a $\psi \in \Enc[F]{D}_k$ with $\phi \fembed \psi$
  but $\phi \not\in \Enc[F]{D}_k$.
  We know that $\psi = \Enc[F]{P}_k$ for some $P \in D$.
  From \cref{lm:fenc-pres-order} we can derive that
  $\Enc[P]{\phi}_k \kembed P$.
  But since $D$ is downward closed, $\Enc[P]{\phi}_k \in D$,
  which contradicts the assumption that $\phi \not\in \Enc[F]{D}_k$.
\end{proof}

\subparagraph{Proof of~\cref{th:idl-is-limit}}
We prove, by induction on $k$,
that $\Idl{\Depth[X]{s,k}} \subseteq \sem{\Limits_{s,k}}$.
Again, the base case is a special case of the induction step.
Let $k>0$ and $D \in \Idl{\Depth[X]{s,k}}$;
By induction hypothesis, we can assume that the ideals of
$\Depth[X]{s,k-1}$ are described
by limits in $\Limits^X_{s,k-1}$.
Since $D$ is downward closed, by \cref{lm:fenc-pres-dwcl},
$\Enc[F]{D}_k \subseteq \Forests_{s,k}$ is downward closed as well,
and thus is the finite union of some ideals of $\Forests_{s,k}^X$.
Now we have,
\[
  \Idl*{\Forests_{s,k}^X} =
  \Idl*{\multiset\left(
    \basic[s]{X} \dunion
    \Forests_{s,k}^{X \union \set{x_{k+1}}}
  \right)}
\]
so we can represent $\Idl{\Forests_{s,k}^X}$
with \pre\miter-products over
$
  \basic[s]{X} \dunion
  \Limits_{s,k}^{X \union \set{x_{k+1}}}.
$
Let $p_1,\dots,p_n$ be such that
$\Union_{i=1}^n \sem{p_i} = \Enc[F]{D}_k$.
From each $p_i$ we can obtain a limit $L_{p_i}$ by replacing:
\begin{enumerate}
  \item $\odot$ with $\parallel$;
  \item each $B^?$ with $B$ and
        each $B^\miter$ with $B^\omega$,
        if $B \in \basic[s]{X}$;
  \item each $L^?$ with $\new x_k.L$ and
        each $L^\miter$ with $(\new x_k.L)^\omega$,
        if $L \in \Limits_{s,k}^{X \union \set{x_{k+1}}}$.
\end{enumerate}
We now show that $D = \Union_{i=1}^n \sem{L_{p_i}}$.
It is important to remember that $\sem{p_i}$ are \pre\fembed-ideals
while $\sem{L_{p_i}}$ are \pre\kembed-ideals.
One can show, by induction on the structure of $p_i$, that
  $ Q \in \instn{L_{p_i}} $
if and only if
  $ Q \congr Q' $
  for some $ Q' \in \Enc[P]{p_i} $.
It then follows that
$
  \sem{L_{p_i}} =
    \set{ Q | \exists \phi \in \sem{p_i}: Q \kembed \Enc[P]{\phi} }
$
and therefore
\begin{align*}
  D &= \set{ Q | \exists \phi \in \Enc[F]{D}: Q \kembed \Enc[P]{\phi} } \\
    &= \Union_{i=1}^n
          \set{ Q | \exists \phi \in \sem{p_i}: Q \kembed \Enc[P]{\phi} }
    = \Union_{i=1}^n \sem{L_{p_i}}
\end{align*}

We thus established that $D = \Union_{i=1}^n \sem{L_{p_i}}$.
Since $D$ is an ideal and each $\sem{L_{p_i}}$ is downward closed,
we have that $D = \sem{L_{p_j}}$ for some $1\leq j \leq n$.
 \end{proof}

\section{Proof of Absorption for $\IntrSymm$ (Lemma~\ref{lm:symm-absorbing})}
Let
  $\vec{x}$ and $\vec{y}$ be two lists of pairwise distinct names,
  $\Gamma \subseteq \Msg^{\SigSymm}$ be a finite set of messages, and
  $\Gamma' = \Gamma\subst{\vec{x} -> \vec{y}}$.
Moreover, assume that $ \names(\Gamma) \inters \vec{y} = \emptyset $.
We have to show that,
for all messages $M$ with $\names(M) \subseteq \names(\Gamma)$,
we have that
  $\Gamma,\Gamma' \deriv M$
    if and only if
  $\Gamma \deriv M$.

The $\Leftarrow$ direction is a direct consequence of~\eqref{axiom:mon}.
For the other direction,
let $\sigma = \subst{\vec{x} -> \vec{y}}$
and $\sigma' = \subst{\vec{y} -> \vec{x}}$,
so $N\sigma' \in \Gamma$ if and only if $N \in \Gamma'$.
Note that ${\Gamma \sigma' = \Gamma' \sigma' = \Gamma}$.
Moreover, if $\names(M) \subseteq \names(\Gamma)$,
$M = M\sigma'$ because, by assumption,
$ \names(\Gamma) \inters \vec{y} = \emptyset $.

We proceed by induction on the depth of the derivation
for $\Gamma,\Gamma' \deriv M$.
We assume the statement holds for derivations of depth $d$ and
do a case analysis on the last rule applied in the
\pre (d+1)-deep derivation for $\Gamma,\Gamma' \deriv M$.

\begin{description}
  \item[\cref{rule:id}:]
    $M \in \Gamma\union\Gamma'$.
    If $M \in \Gamma$ we are done.
    Otherwise, if $M\in\Gamma'$, we have
    $ M\sigma' \in \Gamma $ but $M = M\sigma'$ so $M \in \Gamma$.

  \item[\cref{rule:pl}:]
    $(N_1, N_2) \in \Gamma\union\Gamma'$ and
    $
      \inferrule
        {\Gamma,\Gamma', N_1, N_2 \deriv M}
        {\Gamma,\Gamma' \deriv M}.
    $\\
    We then have two cases:
    \begin{description}

      \item[Case] $(N_1,N_2) \in \Gamma$.
        By \eqref{axiom:mon} we can transform
        the \pre d-deep derivation for
        ${\Gamma,\Gamma', N_1, N_2 \deriv M}$ into a derivation for
        ${\Gamma, N_1, N_2, \Gamma', N_1\sigma, N_2\sigma \deriv M}$
        of the same depth.
        Since $M \subseteq \names(\Gamma, N_1, N_2) = \names(\Gamma)$
        and by definition
        ${\Gamma', N_1\sigma, N_2\sigma = (\Gamma, N_1, N_2)\sigma}$,
        we can apply the induction hypothesis,
        obtaining $ \Gamma, N_1, N_2 \deriv M $
        which proves $\Gamma\deriv M$ by application of \cref{rule:pl}.

      \item[Case] $(N_1,N_2) \in \Gamma'$.
          From ${\Gamma,\Gamma', N_1, N_2 \deriv M}$ we get
          $ {(\Gamma,\Gamma', N_1, N_2)\sigma' \deriv M\sigma'} $
          by \eqref{axiom:alpha}.
        Since we have
          $\Gamma\sigma' = \Gamma'\sigma' = \Gamma$ and $M\sigma' = M$,
        we can infer that
          $ \Gamma, N_1\sigma',N_2\sigma' \deriv M $
        which proves, together with $(N_1\sigma', N_2\sigma') \in \Gamma$,
        that $\Gamma\deriv M$ by application of \cref{rule:pl}.

    \end{description}

  \item[\cref{rule:encl}:]
    $ \enc{N}{K} \in \Gamma\union\Gamma' $ and
    $
      \inferrule
        {\Gamma,\Gamma' \deriv K   \and   \Gamma,\Gamma', N, K \deriv M}
        {\Gamma,\Gamma' \deriv M}.
    $\\
    We distinguish two cases:
    \begin{description}

      \item[Case] $ \enc{N}{K} \in \Gamma $.
        Then $ \names(K) \subseteq \names(\Gamma) $
        and by induction hypothesis we have $\Gamma \deriv K$.
        We also have
          $ \names(N) \subseteq \names(\Gamma) $,
          by definition
          $ (\Gamma,N,K)\sigma = \Gamma', N\sigma, K\sigma $ and
          $ \names(M) \subseteq \names(\Gamma,N,K) $.
        We can transform the \pre d-deep derivation for
        ${\Gamma,\Gamma', N, K \deriv M}$ into a derivation for
        $\Gamma, N, K,\Gamma', N\sigma, K\sigma \deriv M$
        with the same depth,
        so by induction hypothesis we obtain
        ${\Gamma, N, K \deriv M}$.
        By applying \cref{rule:encl} to
        $\Gamma \deriv K$ and $\Gamma, N, K \deriv M$
        we obtain $\Gamma \deriv M$.

      \item[Case] $ \enc{N}{K} \in \Gamma' $.
        Then by \eqref{axiom:alpha},
        $\Gamma,\Gamma' \deriv K$ implies
        $\Gamma\sigma',\Gamma'\sigma' \deriv K\sigma'$.
        Since $\Gamma\sigma' = \Gamma$ and
              $\Gamma'\sigma' = \Gamma$ we have
        $\Gamma \deriv K\sigma'$.
        Similarly,
        $ {\Gamma,\Gamma',N,K \deriv M} $ implies
        $ \Gamma,N\sigma',K\sigma' \deriv M\sigma' $.
        Since $M \sigma' = M$, we have
          $ \Gamma, N\sigma', K\sigma' \deriv M $,
        which in conjuction with
          $\Gamma \deriv K\sigma'$ and
          $\enc{N\sigma'}{K\sigma'} \in \Gamma$,
        proves $\Gamma \deriv M$ by application of \cref{rule:encl}.

    \end{description}

  \item[\cref{rule:aencl}:]
    Analogous to the case of \cref{rule:encl}.

  \item[\cref{rule:pr}:]
    $M = (N_1, N_2)$ and
    $
      \inferrule
        {\Gamma,\Gamma' \deriv N_1   \and   \Gamma,\Gamma' \deriv N_2}
        {\Gamma,\Gamma' \deriv M}.
    $\\
    We have $\names(N_1),\names(N_1) \subseteq \names(\Gamma)$
    so we can apply the induction hypothesis to the two premises and
    get $ \Gamma \deriv N_1 $ and $ \Gamma \deriv N_2 $,
    which proves $ \Gamma \deriv M $ by \cref{rule:pr}.

  \item[\cref{rule:encr}:]
    Analogous to the case of \cref{rule:pr}.

  \item[\cref{rule:aencr}:]
    Analogous to the case of \cref{rule:pr}.
    \qedhere

\end{description}

\section{Proof of Characterisation of Limits Inclusion}

For any limit $L$ with standard form
$
 \new \vec{x}.
            (\out{\Gamma} \parallel
            Q \parallel
            \Parallel_{i\in I} B_i^\omega)
$, we call $\Gamma$ and $Q$ non-iterated, also fixed, components while
$B_i$ for $i \in I$ are iterated components.

We recall the statement:
 Let $L_1$ and $L_2$ be two limits,
 with
 $
   \stdf(L_1) \aeq \new\vec{x}_1.(
     \out{\Gamma_1} \parallel
     Q_1 \parallel
     \Parallel_{i\in I} B_i^\omega
   )
 $,
 and let $n = \card{\vec{x}_1} + \card{Q_1} + 1 $.
 Then:
 \[
   \sem{L_1} \subseteq \sem{L_2}
     \iff
   \begin{cases}
      \stdf(\lext{L_2}{n}) \aeq \new\vec{x}_1,\vec{x}_2.(
         \out{\Gamma_2} \parallel
         Q_1 \parallel
         Q_2 \parallel
         R_2
       )
     \;\text{ and }\;
     \Gamma_1 \kleq \Gamma_2 &
     \cond{A}\\
     \sem{\out{\Gamma_1} \parallel \Parallel_{i\in I} B_i}
       \subseteq
     \sem{\out{\Gamma_2} \parallel R_2}
     & \cond{B}
   \end{cases}
 \]

\subsection{Conditions \cond{A} and \cond{B} Are Sound for Inclusion}

\begin{lemma}[Soundness]
 If conditions \cond{A} and \cond{B} hold, then $\sem{L_1} \subseteq \sem{L_2}$.
\end{lemma}

\label{lm:suff-conds-incl}
  \begin{proof}
  Given that \cond{A} and \cond{B} hold,
  it suffices to show that
  \[\forall n_1 \in \Nat, \exists n_2 \in \Nat.
    \lground{L_1}{n_1} \kembed \lground{\lext{L_2}{n}}{n_2}\]
  by \cref{lm:ext-same-sem,lm:lground-inc-indices}.
  Let $n_1$ be an arbitrary number.
  With $R_1 := \Parallel_{i\in I} B_i^\omega$, we have
  $
     \lground{\stdf(L_1)}{n_1} =
     \new \vec{x}_1.
            (\out{\Gamma_1} \parallel
            Q_1 \parallel
            \lground{R_1}{n_1})
  $ and $
     \stdf(\lground{\lext{L_2}{n}}{0}) \aeq
     \new \vec{x}_1,\vec{x}_2.
            (\out{\Gamma_2} \parallel Q_1 \parallel Q_2)
  $
  by \cond{A}.

  We show that
  ${\new \vec{x}_1.
    (\out{\Gamma_1} \parallel Q_1 \parallel \lground{R_1}{n_1})
   \kembed
   \new \vec{x}_1, \vec{x}_2.
    (\out{\Gamma_2} \parallel Q_1 \parallel Q_2 \parallel \lground{R_2}{n_2})}$
  for some~$n_2$.
  Condition \cond{A} induces a knowledge order for both fixed parts.
    Condition \cond{B} gives rise to some relation between the iterated components we will exploit.
  This relation is captured in the following fact.
  For some \pre\alpha-renaming in \cond{A} and some $n_2$:
  \begin{equation}
   \out{\Gamma_1} \parallel \lground{R_1}{n_1}
   \kembed
   \out{\Gamma_2} \parallel \lground{R_2}{n_2}
   \tag{$*$}
   \label{renamed-kembed}
 \end{equation}
  Prior to proving \eqref{renamed-kembed}, we show that this fact and \cond{A} suffice to prove our goal.
  With $\stdf(\lground{R_i}{n_i}) =
        \new\vec{y_i}.(\Gamma_{R_i^{n_i}} \parallel
                       Q'_i)$,
  we call $\Gamma_{R_i^{n_i}}$ the knowledge of $\lground{R_i}{n_i}$ for $i \in \set{1, 2}$.
  Because of \eqref{renamed-kembed}, we can choose $\vec{y}_2$ so that
  $\vec{y}_2 = \vec{y}_1, \vec{z}_2$ and there is no need to rename for the embedding anymore.
  Equipped with this abbreviation, the knowledge of the left hand side
  $(\new \vec{x}_1. (\out{\Gamma_1} \parallel Q_1 \parallel \lground{R_1}{n_1})$
  is $(\Gamma_1 , \Gamma_{R_1^{n_1}})$
  while $\Gamma_2, \Gamma_{R_2^{n_2}}$ is the knowledge of the right hand side
  $\new \vec{x}_1, \vec{x}_2.(\out{\Gamma_2} \parallel Q_1 \parallel Q_2 \parallel \lground{R_2}{n_2})$.
  This is also exactly the knowledge contained in
  $(\out{\Gamma_1} \parallel \lground{R_1}{n_1})$ respectively
  $(\out{\Gamma_2} \parallel \lground{R_2}{n_2})$, hence
  $\Gamma_1, \Gamma_{R_1^{n_1}} \kleq \Gamma_2, \Gamma_{R_2^{n_2}}$ by \eqref{renamed-kembed}.
  It remains to take care of names and process calls for the embedding.
  Condition \cond{A} already takes care of non-iterated names and process calls.
  With our assumption \eqref{renamed-kembed}, we obtain
  \[
   \new \vec{x}_1.
     (\out{\Gamma_1} \parallel Q_1 \parallel \out{\Gamma_1} \parallel \lground{R_1}{n_1})
   \kembed
   \new \vec{x}_1, \vec{x}_2.
     (\out{\Gamma_2} \parallel Q_1 \parallel Q_2 \parallel \out{\Gamma_2} \parallel \lground{R_2}{n_2}).
  \]
  Both parts are knowledge congruent to our goal's sides by
  \cref{lm:persistence-of-knowledge} which proves the goal. \\

  It remains to show that \eqref{renamed-kembed}, i.e.~$
   (\out{\Gamma_1} \parallel \lground{R_1}{n_1})
    \kembed
   (\out{\Gamma_2} \parallel \lground{R_2}{n_2})$ for some~$n_2
  $.
  We start with \cond{B}
  and apply \cref{lm:lground-inc-indices} so that we know that
  for every $n_1 \in \Nat$, there is a $m \in \Nat$ such that
  $
   \out{\Gamma_1} \parallel \lground{\Parallel_{i\in I} B_1}{n_1}
   \kembed
   \out{\Gamma_2} \parallel \lground{R_2}{m}
  $.
  By \cref{lm:kembed-lground-proc}, we multiply on both sides:
  \[
    (\out{\Gamma_1} \parallel
    \lground{\Parallel_{i\in I} B_1}{n_1})^{n_1}
     \kembed
    (\out{\Gamma_2} \parallel \lground{R_2}{m})^{n_1}.
  \]
  By \cref{lm:persistence-of-knowledge}, we pull the messages out of both replications:
  \[
    \out{\Gamma_1} \parallel (\lground{\Parallel_{i\in I} B_1}{n_1})^{n_1}
     \kembed
    \out{\Gamma_2} \parallel (\lground{R_2}{m})^{n_1}.
  \]
  We continue on both sides individually.
  On the left, we omit messages for simplicity as we already know that $\Gamma_1 \kleq \Gamma_2$.
  \[
    (\lground{\Parallel_{i\in I} B_1}{n_1})^{n_1}
=
    (\Parallel_{i\in I}{\lground{B_1}{n_1}})^{n_1}
    =
    \lground{R_1}{n_1}.
  \]
  While we use \cref{lm:kembed-limit-mult} and
  \cref{cor:kembed-add-mess} on the right in order to obtain that:
  \[
    \out{\Gamma_2} \parallel (\lground{R_2}{m})^{n_1}
     \kembed
    \out{\Gamma_2} \parallel \lground{R_2}{n_1 \cdot m}
  \]
  Combining both paths leads to:
  $
   \out{\Gamma_1} \parallel \lground{R_1}{n_1}
    \kembed
   \out{\Gamma_2} \parallel \lground{R_2}{n_1 \cdot m}
  $.
  By choosing $n_2 = n_1 \cdot m$, the claim follows.
  \end{proof}

\subsection{Conditions \cond{A} and \cond{B} Are Complete for Inclusion}

\begin{lemma}[Completeness]
 If $\sem{L_1} \subseteq \sem{L_2}$, then conditions \cond{A} and \cond{B} hold.
\end{lemma}

For this proof we are assuming the intruder model is absorbing.
Before we give the proof,
 we need some auxiliary definitions and results.

  \begin{corollary}
  \label{cor:deriv-variant-idemp}
  Let
      $\vec{x}$ and $\vec{y}$ be two lists of pairwise distinct names,
      $\Gamma, \Gamma_1$ be two finite sets of messages, and
      $\Gamma_2 = \Gamma_1\subst{\vec{x} -> \vec{y}}$.
    Moreover, assume that $ \names(\Gamma_1) \inters \vec{y} = \emptyset$ and
    $\names(\Gamma) \inters \vec{y} = \emptyset = \names(\Gamma) \inters \vec{x}$.
    Then, for all messages $M$ with $\names(M) \subseteq \names(\Gamma, \Gamma_1)$,
    we have that
      $\Gamma,\Gamma_1, \Gamma_2 \deriv M$
        if and only if
      $\Gamma, \Gamma_1 \deriv M$.
  \end{corollary}
  \begin{proof}
   The direction from right to left is obvious.
   For the reverse direction, it is equivalent to show that
   $\Gamma, \Gamma_1, \Gamma, \Gamma_2 \deriv M$.
   Now, $\Gamma, \Gamma_1 = (\Gamma, \Gamma_2)\subst{\vec{x} -> \vec{y}}$.
   The claim follows by the assumption that the intruder is absorbing
   (\cref{def:absorption}).
  \end{proof}

  \begin{lemma}
  \label{lm:congr-exp}
   Let $L, L'$ be two limits
   such that $L \congr L'$.
   Then for every $n \in \Nat: \lground{L}{n} \congr \lground{L'}{n}$.
  \end{lemma}

  \begin{proof}
    The proof is a straightforward structural induction on $L$.
  \end{proof}

  We introduce a refinement of grounding and a function folding the right hand side but preserving the knowledge embedding.
  The \pre(n,k,m)-th grounding takes a limit and unfolds each $\omega$
  $n$ times for the outer $k$ nested levels of $\omega$,
  and $m$ times for the inner ones.

  \begin{definition}[Step-indexed grounding]
  \label{def:step-indexed-exp}
    For a limit $L$ in standard form,
    we define the \pre(n,k,m)-th grounding of $L$ to be the process
    $\lground{L}{n,k,m}$ recursively:
    \[
      \lground{L}{n,k,m} \is
        \begin{cases}
          L
                                & \text{if $L$ is sequential or }\zero\\
          \lground{L_1}{n,k,m} \parallel \lground{L_2}{n,k,m}
                                & \text{if } L = L_1 \parallel L_2\\
          \new x.(\lground{L'}{n,k,m})
                                & \text{if } L = \new x.L'\\
          \left(\lground{B}{n,k-1,m}\right)^{n}
                                & \text{if } L = B^\omega \land k > 0\\
          \left(\lground{B}{m}\right)^{m}
                                & \text{if } L = B^\omega \land k = 0
        \end{cases}
    \]
  \end{definition}

  \begin{definition}[$\omega$-height]
  For a limit L, we define the $\omega$-height as follows:
    \[
      \height(L) \is
        \begin{cases}
          0
                                & \text{if $L$ is sequential or }\zero\\
          \max{(\height(R_1), \cdots, \height(R_n))}
                                & \text{if }
                        L = R_1 \parallel \cdots \parallel R_n\\
          \height(L')
                                & \text{if } L = \new x.L'\\
          \height(B) + 1
                                & \text{if } L = B^\omega
        \end{cases}
    \]
  \end{definition}

  \begin{lemma}
  \label{lm:rho-index-height}
Let $L$ be a limit and $m, n, k \in \Nat$.
   If $k \geq \height(L)$,
   then $\lground{L}{n,k,m} = \lground{L}{n}$.
  \end{lemma}

  \begin{proof}
   The parameter $k$ only decreases when recursing into a limit under $\omega$.
   If $k \geq \height(L)$ the last case of the definition will never apply,
   which makes the definition coincide with the one of \pre n-grounding.
  \end{proof}

  The idea of the following parametrised function is to fold an $m$-grounding to an $n$-grounding up to a certain $\height$ $k$ of the limit.
  Since we want to be very specific about the domains, we define some sets of groundings.

  \begin{definition}[Set of groundings]
   Let $L \in \Limits$ and $m, n, k \in \Nat$. \\
   We define the following two sets of processes:
    \begin{itemize}
    \item $\mathcal{R}_L^m       \is \set{\lground{\stdf(L)}{m}}$
    \item $\mathcal{R}_L^{n,k,m} \is \set{\lground{\stdf(L)}{n,k,m}}$
   \end{itemize}
  \end{definition}

  \begin{definition}
  \label{def:folding-function}
   Let $L$ be a limit in recursive standard form and $m,k,n \in \Nat$.
   The function $\fold_{k,L}^{n,m} : \mathcal{R}_L^m \to \mathcal{R}_L^{n,k,m}$ and is parametrised in all four variables.

   \[
    \fold_{k, P}^{n, m}(P) \is
    \begin{cases}
        P
            \CASE k=0 \\
        \new \vec{x}.\left(\out{\Gamma} \parallel
           Q \parallel
           \Parallel_{j\in J}
           \left(
           \fold_{k,L_j}^{n,m}(\lground{L_j^\omega}{m})
           \right)
           \right)
            \CASE k > 0 \; \land \; P = \new \vec{x}.(\out{\Gamma} \parallel Q \parallel \Parallel_{j \in J}L_j^\omega) \\
        \left( \fold_{k-1,L}^{n,m}(\lground{L}{m}) \right)^n
             \CASE k > 0 \; \land \; P = \lground{L^\omega}{m}

    \end{cases}
   \]

\end{definition}
  We may omit the parameters $n$ and $m$ in the following if obvious from context as they do not change over the process of folding.

  \begin{lemma}[Folding is sound wrt.~knowledge embedding]
  \label{lm:fold-sound-kembed}
   Let $L_1, L_2$ be two limits in standard form
   $L_i = \new \vec{x}_i.
   (\out{\Gamma_i} \parallel Q_i \parallel R_i)$ with
   $n = \card{\vec{x_1}} + \card{Q_1} + 1$.
   If there is a $m \in \Nat$ such that $m > n$ and
   $\lground{L_1}{0} \kembed \lground{L_2}{m}$,
   then $\forall k: \lground{L_1}{0} \kembed \fold_{k,L_2}^{n,m}(\lground{L_2}{m})$.
  \end{lemma}

  \begin{proof}
  We prove the claim by induction on $k$.

  For k = 0, the claim trivially follows as
  $\fold_{0,L_2}^{n,m}(\lground{L_2}{m}) = \lground{L_2}{m}$
  by the assumption that
  $\lground{L_1}{0} \kembed \lground{L_2}{m}$
  holds.

  For the induction step, we assume that
  $\lground{L_1}{0} \kembed \fold_{k,L_2}^{n,m}(\lground{L_2}{m})$
  and prove
  \[
   \lground{L_1}{0} \kembed \fold_{k+1,L_2}^{n,m}(\lground{L_2}{m}).
  \]
  By definition of folding, both $\fold_{k, L}(\lground{L_2}{m})$ and $\fold_{k+1, L}(\lground{L_2}{m})$ are folding in exactly the same way up to the $k$-th recursive calls, i.e.~calls in which $k$ decreases.
  This means that up to the calls to $\fold_{0,L'}$ and $\fold_{1,L'}$,
  both $\fold_{k, L}(\lground{L_2}{m})$ and $\fold_{k+1, L}(\lground{L_2}{m})$
  will have constructed the same context $C[-, \cdots, -]$
  around these final calls.
We thus characterise $A = \fold_{k,L_2}^{n,m}(\lground{L_2}{m})$,
  and $B = \fold_{k+1,L_2}^{n,m}(\lground{L_2}{m})$
  as follows:
  \begin{align*}
    A &= C[\fold_{0,F_1^\omega}^{n,m}(F_1^\omega), \cdots,
            \fold_{0,F_j^\omega}^{n,m}(F_j^\omega)] \\ &=
          C[\lground{F_1^\omega}{m}, \cdots,
            \lground{F_j^\omega}{m}] =
          C[(\lground{F_1}{m})^m, \cdots,
            (\lground{F_j}{m})^m] \\
     B &= C[\fold_{1,F_1^\omega}^{n,m}(F_1^\omega), \cdots,
            \fold_{1,F_j^\omega}^{n,m}(F_j^\omega)] =
          C[(\lground{F_1}{m})^n, \cdots,
            (\lground{F_j}{m})^n]
  \end{align*}
  Recall that
  $L_1 = \new \vec{x}_1. (\out{\Gamma_1} \parallel Q_1 \parallel R_1)$
  and
  hence
  $\lground{L_1}{0} =
         \new \vec{x}_1. (\out{\Gamma_1} \parallel Q_1)$.
  By assumption, we have
  \begin{equation}
  \new \vec{x}_1. (\out{\Gamma_1} \parallel Q_1) \kembed
  \fold_{k,L_2}^{n,m}(\lground{L_2}{m}) = A
  \label{eq:fold-sound-kembed-i}
  \end{equation}
  and want to prove that
  \begin{equation}
  \new \vec{x}_1. (\out{\Gamma_1} \parallel Q_1) \kembed
  \fold_{k+1,L_2}^{n,m}(\lground{L_2}{m}) = B.
  \label{eq:fold-sound-kembed-ii}
  \end{equation}
  Intuitively, we have to find a way to preserve the knowledge embedding from
  \eqref{eq:fold-sound-kembed-i} when removing some branches in the holes of the context to get from $A$ to $B$.
  Let us show what $A$ and $B$ look like explicitly with their context:
  \newcommand{\holesA}{D}
  \newcommand{\holesB}{E}
  \begin{align*}
    A & \congr \new \vec{c}.(\out{\Gamma_c} \parallel
                Q_c \parallel \holesA)
    &&\text{ with }
    \holesA = (\lground{F_1}{m})^m
             \parallel \cdots \parallel
             (\lground{F_j}{m})^m
                \\
    B & \congr \new \vec{c}.(\out{\Gamma_c} \parallel
                Q_c \parallel \holesB)
    \quad
    &&\text{ with }
    \holesB = (\lground{F_1}{m})^n
             \parallel \cdots \parallel
             (\lground{F_j}{m})^n
  \end{align*}
  We call $\Gamma_c$ the knowledge of the context.

  Our goal is to show that reducing the number of iterations in each of $\holesA$'s parallel components does not affect the knowledge embedding described in $(i)$ and thereby obtain $(ii)$.
  Let us first consider names and process calls.
  We can split $\lground{L_1}{0}$ in the following way:
  $\lground{L_1}{0} =
         \new \vec{x}_1. (\out{\Gamma_1} \parallel Q_1)
       =
         \new\vec{y}. \new\vec{z}. (\out{\Gamma_1} \parallel Q_y \parallel Q_z)$
  where $\vec{y} \subseteq \vec{c}$ and $Q_y \subseteq Q_c$
  .
  The intention is to distinguish names and process calls that are already matched in the context.
  This is why we do not require all process calls that are only using names from $\vec{y}$ to be in $Q_y$ but some of them might be in $Q_z$.
  The goal follows with the following claim immediately.

  \textbf{Claim I:} In every hole of context $C[-, \cdots, -]$,
  we can reduce the number of branches to at most $n$, i.e.~for every $1 \leq l \leq j$ we can have a grounding
  $\lground{F_l}{n}$ instead of $\lground{F_j}{m}$.

  Proof of Claim I.
  Towards a contradiction, assume that there is a hole in which we cannot remove $m-n$ branches.
  W.l.o.g.\ let $\lground{L_l}{m}$ for $1 \leq l \leq j$ be the sublimit in this hole.
  There might be three reasons for this: names, process calls and knowledge.

  Considering the names and process calls, we know that
  $\card{\vec{x_1}} + \card{Q_1}~<~n$ by definition and hence
  $\card{\vec{z}} + \card{Q_z}~<~n$.
  Now, we investigate the components that $\vec{z}$ and $Q_z$ are matched to.
  In the worst case, all of them are mapped to this hole but we still delete $m-n$ branches that are not used to match the names $\vec{z}$.
    For the process calls in $Q_z$, we have to distinguish two cases.
  First, if a process call uses any name from $\vec{z}$, it is fine as we will leave them anyway.
  Second, if a process call does not use any name from $\vec{z}$, it is fine to delete this branch as there will be enough copies in the remaining branches to cover this process call.

  It remains to reason about knowledge.
  We make the knowledge of $\holesA$, i.e.~the one having budget $k$, explicit:
  $\stdf(\holesA) = \new\vec{a}.(\Gamma_m \parallel \cdots)$
   and factor out the knowledge from sublimit
  $\lground{F_l}{m}$:
  $\Gamma_m = \Gamma_m', \Gamma_{l,m}$
  so that $\Gamma_{l,m}$ was the knowledge obtained through
  $\lground{F_l}{m}$.
For knowledge, we have to prove that
  $\Gamma_1 \kleq \Gamma_c, \Gamma_m', \Gamma_{l,n}$.
  By \cref{lm:locality},
  we need to prove that for every message $M\in\Gamma_1$,
  $\Gamma_c, \Gamma_m', \Gamma_{l,n} \deriv M$ given that
  $\Gamma_c, \Gamma_m', \Gamma_{l,m} \deriv M$.
  The idea now is to reduce the knowledge from $\Gamma_{l,m}$ to $\Gamma_{l,n}$ by \cref{cor:deriv-variant-idemp}, which is a corollary of the absorbing intruder.
  Let us define $\Gamma_{c,m}' = \Gamma_c, \Gamma_m'$ indicating the context of the hole we are considering.
  As $m$ might be bigger than $2n$, we have to iterate the process of reducing the number of branches.
  Hence, we generalise the notation of $\Gamma_{l,m}$ and $\Gamma_{l,n}$ in the following way:
  $(\lground{F_l}{m})^i
   \congr
   \new\vec{a_i}.(\out{\Gamma_{l,i}} \parallel \cdots)$.

  \textbf{Claim II:} $\forall m > n$,
  $\Gamma_{l,m} \deriv M \implies \Gamma_{l,m-1} \deriv M$.

  Proof of Claim II.
  For convenience, we rename $\Gamma_{l,i}$ to $\Lambda_i$.
  The main observation is that we can split the knowledge $\Lambda_m$ into $\Lambda_{m-1}$ and a remainder $\Lambda'$.
  We can choose a branch which does not use names from $\vec{z}$ to contribute to $\Lambda'$.
  Since we know that $n \geq 1$, we know that $m > 1$ by assumption.
  Therefore, we can split $\Lambda_{m-1}$ again and obtain the knowledge stemming from one branch which we call $\Lambda''$.
  Let us recall the assumption and goal after these rewriting steps:
  Given that
  \begin{equation}
    \Gamma_{c, m}', \Lambda_{m-2}, \Lambda', \Lambda'' \deriv M
    \label{eq:fold-sound-kembed-iii}
  \end{equation}
  holds,
  we want to prove that $\Gamma_{c, m}', \Lambda_{m-2}, \Lambda'' \deriv M$.
  Let $\vec{w}'$ and $\vec{w}''$ be the names only used in $\Lambda'$ and $\Lambda''$ respectively so that:
  \begin{equation}
    \vec{w}' \inters \vec{w}'' = \emptyset
  \text{ and }
   \names(\Gamma_{c, m}', \Lambda_{m-2}) \inters \vec{w}' = \emptyset =
   \names(\Gamma_{c, m}', \Lambda_{m-2}) \inters \vec{w}''
   \label{eq:fold-sound-kembed-iv}
  \end{equation}
  $\Lambda'$ and $\Lambda''$ have been obtained from a branch of the same sublimit, so we can infer that
  \begin{equation}
    \Lambda'' = \Lambda'\subst{\vec{w}' -> \vec{w}''}.
    \label{eq:fold-sound-kembed-v}
  \end{equation}
  Notice that $\vec{w}' \inters \vec{z} = \emptyset$ by the fact how we have chosen the branch for $\Lambda'$.
  Furthermore, $\vec{w}' \inters \vec{y} = \emptyset$ by
  $\vec{y} \subseteq \vec{c}$.
  Combining these observations, we get that
  $ \vec{x_1} \inters \vec{w}' = \emptyset $.
By the fact that $M \in \Gamma_1$, we have
  $ \names(M) \subseteq \names(\Gamma_1)
             \subseteq \vec{x_1} $.
  Therefore, $\names(M) \inters \vec{w'} = \emptyset$ which implies
  that
  \begin{equation}
    \names(M) \subseteq \names(\Gamma_{c, m}', \Lambda_{m-2}, \Lambda'')
    \label{eq:fold-sound-kembed-vi}
  \end{equation}
  Facts \eqref{eq:fold-sound-kembed-iii}~to~\eqref{eq:fold-sound-kembed-vi}
  fulfil the conditions for
  \cref{cor:deriv-variant-idemp}, resulting in
  ${\Gamma_{c, m}', \Lambda_{m-2}, \Lambda'' \deriv M}$ which reads
  $\Gamma_{c, m}', \Lambda_{m-1} \deriv M$ when folding back which is the goal of Claim~II.
  By this, we have shown that we can remove $m - n$ branches which leads to a contradiction which is why Claim I holds.
  As $\lground{F_l}{m}$ was chosen arbitrarily, we have shown that we can remove $m-n$ branches in every hole of the context
  $C[-, \cdots, -]$
  which concludes this proof.
\end{proof}

  \begin{corollary}
   \label{cor:n-is-enough-condition}
   Let $L_1, L_2$ be two limits with $\stdf(L_1) =
   \new \vec{x}_1.(\out{\Gamma_1} \parallel Q_1 \parallel R_1)$
   and
   $\sem{L_1} \subseteq \sem{L_2}$.
   Then, $\stdf(\lext{L_2}{n}) \aeq
         \new\vec{x}_1,\vec{x}_2.(
         \out{\Gamma_2} \parallel
         Q_1 \parallel
         Q_2 \parallel
         R_2)
     \;\text{ and }\;
     \Gamma_1 \kleq \Gamma_2.$
   for $n = \card{\vec{x}_1} + \card{Q_1} + 1$.
  \end{corollary}
  \begin{proof}
   First, we show that
   $\lground{\stdf(L_1)}{0} \kembed \lground{L_2}{n}$.
   By $\sem{L_1} \subseteq \sem{L_2}$, we know that
   there is an $m$ so that
   $\lground{\stdf(L_1)}{0} \kembed \lground{L_2}{m}$
   by \cref{lm:lground-inc-indices}.
   From \cref{lm:fold-sound-kembed}, we obtain that
   that $\lground{\stdf(L_1)}{0} \kembed
   \fold_{k,L_2}^{n,m}(\lground{L_2}{m})$ for every $k$.
   Substituting $\height(L_2)$ for $k$ leads to
   $\fold_{\height,L_2}^{n,m}(\lground{L_2}{m})$.
   This is $\lground{L_2}{n}$ by \cref{lm:rho-index-height}.
   Using this knowledge embedding, we get
   $\lground{L_2}{n} \aeq
   \new \vec{x}_1,\vec{x}_2.(\out{\Gamma_2} \parallel Q_1 \parallel Q_2)$ with $\Gamma_1 \kleq \Gamma_2$.
   With \cref{lm:ext-n-ground-0-eq-ground-n}, we observe that
   $\lground{L_2}{n} = \lground{\lext{L_2}{n}}{0}$.
   By this, the claim follows as $\lground{-}{0}$ just omits the iterated parts, i.e.~$R_2$, from $\stdf(\lext{L_2}{n})$.
  \end{proof}

  \begin{lemma}[Necessary Conditions for Inclusion]
  \label{lem:nec-conds-incl}
   Let $L_1$ and $L_2$ be two limits,
 with
 $
   \stdf(L_1) \aeq \new\vec{x}_1.(
     \out{\Gamma_1} \parallel
     Q_1 \parallel
     R_1
   )
 $ with
 $R_1 =  \Parallel_{i\in I} B_i^\omega$,
 and let $n = \card{\vec{x}_1} + \card{Q_1} + 1 $.
   Given that the inclusion
   $\sem{L_1} \subseteq \sem{L_2}$,
   both conditions \cond{A} and \cond{B} hold.
  \end{lemma}

  \begin{proof}
  For \cond{A}, the claim follows by \cref{cor:n-is-enough-condition} which implicitly gives a renaming for $L_2$. \\
  For \cond{B}, we want to show that
  $\sem{\out{\Gamma_1} \parallel \Parallel_{i\in I} B_i} \subseteq \sem{\out{\Gamma_2} \parallel R_2}$.\\
  Let $N_i = \out{\Gamma_i} \parallel R_i$.
  It is straightforward to see that
  $\sem{\out{\Gamma_1} \parallel \Parallel_{i\in I} B}
  \subseteq
  \sem{N_1}$.
  This is why it is enough to show that
  $\sem{N_1} \subseteq \sem{N_2}$
  by transitivity.\\
  Towards a contradiction, we assume that for all possible renaming so that \cond{A} is satisfied,
  $\sem{N_1} \not \subseteq \sem{N_2}$.
  By \cref{lm:lground-inc-indices},
  \begin{equation}
   \exists m_1', \forall m_1 \geq m_1', \forall m_2.
   \lground{N_1}{m_1}
    \not \kembed
   \lground{N_2}{m_2}.
   \label{eq:nec-cond-notkembed}
  \end{equation}
  First, knowledge could break the embedding.
  Let us define
  \[
    \stdf(\lground{N_i}{m_i}) =
    \new y_i.(\Gamma_i \parallel \Gamma_i' \parallel Q_i').
  \]
  Notice that $\names(\Gamma_i) \inters \vec{y_i} = \emptyset$.
  There might be two reasons why the knowledge embedding does not hold.

  First, the embedding breaks because of knowledge.
  This is impossible as $\Gamma_1, \Gamma_1'$ and
  $\Gamma_2, \Gamma_2'$ represent the knowledge of groundings of the two limits $L_1$ and $L_2$ as their top level knowledge is replicated in \cond{B}.
  Therefore, the inclusion $\sem{L_1} \subseteq \sem{L_2}$ would also break which is a contradiction.

  Second, names or process calls can hence be the only reasons why the knowledge embedding
  $\lground{N_1}{m_1'}
    \kembed
   \lground{N_2}{m_2}$ does not hold for any $m_2$.
  We will derive a contradiction by choosing
  $n_1 = 2 \cdot max(\card{\vec{x}_2} + \card{Q_2} + 1, m_1')$.
  We incorporate $\lground{R_1}{n_1}$ into $\stdf(L_1)$:
  $\lground{\stdf(L_1)}{n_1} =
  \new \vec{x}_1.(\Gamma_1 \parallel Q_1 \parallel
  \lground{R_1}{n_1})$.
  Recall that
  $\lground{\stdf(\lext{L_2}{n})}{m_2}
     \aeq \new\vec{x}_1,\vec{x}_2.(
         \out{\Gamma_2} \parallel
         Q_1 \parallel
         Q_2 \parallel
         \lground{R_2}{m_2}
       )$.
  By the size of $n_1$, at least half of the names and process calls of $\lground{R_1}{n_1}$ have to be covered by $\lground{R_2}{m_2}$ as the non-iterated part
  $\vec{x_2}$ and $Q_2$ cannot do more than half.
  But by definition $\frac{n_1}{2} \geq m_1'$ so
  $\forall m_2, \lground{R_1}{\frac{n_1}{2}} \not \kembed \lground{R_2}{m_2}$ as knowledge cannot be the reason for \eqref{eq:nec-cond-notkembed} to break.
  Altogether, this entails that there is a $n_1$ such that for all $m_2$:
  \begin{align*}
   \lground{\stdf(L_1)}{n_1} & =
  \new \vec{x}_1.(\Gamma_1 \parallel Q_1 \parallel
  \lground{R_1}{n_1})
  \\
  & \not \kembed
  \new\vec{x}_1,\vec{x}_2.(
         \out{\Gamma_2} \parallel
         Q_1 \parallel
         Q_2 \parallel
         \lground{R_2}{m_2}
       )
  \aeq
  \lground{\stdf(\lext{L_2}{n})}{m_2}
  \end{align*}
  Using \cref{lm:ext-same-sem}, this implies that
  $\lground{L_1}{n_1} \not \kembed \lground{L_2}{m_2}$ for every $m_2$.
  In turn, this entails that
  $\sem{L_1} \not \subseteq \sem{L_2}$ by
  \cref{lm:lground-inc-indices} which is a contradiction.
  \end{proof}

\section{Correctness of \texorpdfstring{$\posthat$}{posthat}}

\begin{definition}[$\beta(\Delta)$ and $b$]
Let $\Delta$ be a set of definitions.
We define $\beta(\Delta)$ and $\gamma(\Intruder)$ as follows:
 \begin{align*}
    \beta(\Delta) & \is
\max\Set{\card{\vec{x}} | (\p{Q}[\vec{y}] \is A + \inp{\vec{x}:M}.P + A') \in \Defs}
  \\
    \gamma(\Intruder) & \is
    \max\Set{\arity(\constr{f}) | \constr{f} \in \Sig)}
    \text{ with }
    \Intruder = (\Sig, \deriv)
  \end{align*}
  then $ b \is \beta(\Delta)  \cdot \gamma(\Intruder)^{s-1} + 1 $.
\end{definition}

Recall the definition of $\posthat$ from \cref{def:posthat}.
  Let
$\stdf(\lext{L}{b}) = \new \vec{x}.(\out{\Gamma} \parallel Q \parallel R)$,
  then
  \[
    \posthat_\Defs^s(L) \is
        \Set{\;
          \new \vec{y}.\bigl(\out{\Gamma'} \parallel Q' \parallel R\bigr)
          \;|\;
            \new \vec{x}.\bigl(\out{\Gamma} \parallel Q\bigr)
              \redto_\Defs
            \new \vec{y}.\bigl(\out{\Gamma'} \parallel Q'\bigr)
              \in \Size{s}
        \;}
  \]

  \begin{corollary}[Post of expansion is enough]
  \label{cor:post-exp-enough}
   Let $L$ be a limit and $P \in \sem{L}$. Then,
   for every $P_1 \in \post(P)$,
   there is a $P_2 \in \post(\lground{L}{n})$ for some $n \in \Nat$ such that
   $P_1 \kembed P_2$.
  \end{corollary}

  \begin{proof}
   As $P \in \sem{L}$, we know that there is a $n' \in \Nat$ such that $P \kembed \lground{L}{n'}$ by \cref{lm:lground-kembed}.
   We choose this $n'$ to be $n$.
   We have that $P \redto P_1$ and $P \kembed \lground{L}{n}$.
   As $\kembed$ is a simulation
   (\cref{th:kembed-simulation}),
   we know that there is a
   $P_2$ so that $\lground{L}{n} \redto P_2$ and
   $P_1 \kembed P_2$.
  \end{proof}
  Let us recall the theorem we want to prove:
for all $L\in \Limits_{s,k}$,
\begin{gather*}
  \posthat_\Defs^s(L) = \set{L_1,\dots,L_n}
    \implies
      \dwd{\bigl( \post(\sem{L}) \inters \Size{s} \bigr)} =
        \sem{L_1} \union \dots \union \sem{L_n}.
\end{gather*}

\begin{proof}[Proof of \cref{th:posthat-correct}]
 We assume that
 $\posthat_\Defs^s(L) = \set{L_1,\dots,L_n}$ and prove the
 set equality by inclusion in both directions.

 First, we show that
 $\dwd{\bigl( \post(\sem{L}) \inters \Size{s} \bigr)}
 \supseteq
 \sem{L_1} \union \dots \union \sem{L_n}$.
 It suffices to show that
 $ \sem{L_j} \subseteq
 \dwd{\bigl( \post(\sem{L}) \inters \Size{s} \bigr)} $
 for every $j \in \set{1, \cdots, n}$.
 We choose $j$ arbitrarily and prove the latter statement.
 By definition, we know that $L_j$ stems from at least one transition in the non-iterated part of
 $\stdf(\lext{L}{b}) =
 \new \vec{x}.(\out{\Gamma} \parallel Q \parallel R)$.
 Wlog~let
 $P_1 := \new \vec{x}.\bigl(\out{\Gamma} \parallel Q\bigr)
          \redto_\Defs
          \new \vec{y}.\bigl(\out{\Gamma'} \parallel Q'\bigr)
          =: P_2$
 be this transition in $\Size{s}$.
 Hence,
 $ L_j =
   \new \vec{y}.\bigl(\out{\Gamma'} \parallel Q' \parallel R\bigr).
 $
 By \cref{lm:ext-same-sem},
 $\sem{L} = \sem{\lext{L}{b}}$ and therefore
 $\lground{\new\vec{x}.(\out{\Gamma} \parallel Q \parallel R)}{n}
  =
  \new\vec{x}.(\out{\Gamma} \parallel Q \parallel
  \lground{R}{n})
  \in \sem{L}$ for every $n$.
 Because of $P_1 \redto P_2$, we know that
 $\lground{\new\vec{x}.(\out{\Gamma} \parallel Q \parallel R)}{n}
  \redto
  \lground{\new\vec{y}.(\out{\Gamma'} \parallel Q' \parallel R)}{n}$.
 This is why
 $\lground{L_j}{n} =
  \lground{\new\vec{y}.(\out{\Gamma'} \parallel Q' \parallel R)}{n} \in \post(\sem{L})$ for every $n$.
 With \cref{lm:lground-kembed}, we know that for every
 $P \in \sem{L_j}$, there is a $m$ so that
 $P \kembed \lground{L_j}{m}$.
 By downward-closure, we infer that
 $P \in \dwd{\bigl(\sem{\post(L)} \inters \Size{s}\bigr)}$.

 Second, we show that
 $\dwd{\bigl( \post(\sem{L}) \inters \Size{s} \bigr)}
 \subseteq
 \sem{L_1} \union \dots \union \sem{L_n}$.
 Let $\lground{L}{b}$ has standard form:
 $\stdf(\lground{L}{b}) \aeq
  \new\vec{x}_1.(\Gamma_1 \parallel \p Q[\vec{M}] \parallel C_1)$.
 By \cref{cor:post-exp-enough}, it suffices to consider only successors of groundings of $L$.
 Therefore, let $b < m \in \Nat$ with
 $\stdf(\lground{L}{m})
   \aeq
  \new\vec{x}_1, \vec{x_2}.(\out{\Gamma_1} \parallel
                            \out{\Gamma_2}\parallel
                 \p Q[\vec{M}] \parallel C_2)$.
 We consider the three different reduction rules that can be fired starting from this grounding:
 \begin{itemize}
  \item a principal is waiting for some message on a private channel which is derivable and the pattern can be matched by $\Gamma$
  \item a principal would like to send a message on some private channel which can be derived from the environment
  \item one principal sends a message on a private channel to another principal
 \end{itemize}
 We split the parts into three paragraphs.

 \subparagraph{Public private channel, message input}
 We have
 $\p Q[\vec{M}] := \inpc{a}{\vec{p} : N}.
                 P_1 + A $ for some private channel $a$ and action $A$.
 Wlog~we derive the message which is matched using some fresh intruder names $\vec{c}$:
 $\Gamma_1, \vec{c} \deriv N \subst{\vec{x} -> \vec{M}'}$
 and the channel name $a$ can be derived by assumption: $\Gamma_1, \Gamma_2 \deriv a$.
 Overall, we get the following transition:
 \[
  \lground{L}{m}
  \;
   \underset{\vec{p} \; \rightarrow \; \vec{M}'}{
        \xrightarrow{\p Q[\vec{M}] = \inpc{a}{\vec{p} : N}.P_1 + A}}
  \;
  \underbrace{\new \vec{x}_1,\vec{x}_2, \vec{c}.
    (\out{\Gamma_1} \parallel \out{\Gamma_2}
    \parallel \out{\vec{c}} \parallel
     P_1\subst{\vec{p} -> \vec{M}'} \parallel C_2)
   }_{Q'\mathrlap{{} \in \post^s(L)}}
 \]
 where the annotations explicitly state which process call and action was used with which substitution.
 We want to show that we can have the same reduction in
 $\lground{L}{b}$.
 We do so by using the $\fold$ function as in \cref{lm:fold-sound-kembed}
 to relate the difference between iterating $k$ times and $k+1$ times.

 \textbf{Claim~I:}
For every $k \in \Nat$,
  $\exists \vec{y}_k, \Delta_k, D_k$ such that
 \begin{equation}
  \fold_{k,L}^{b,n}(\lground{L}{n})
   \;
   \underset{\vec{p} \; \rightarrow \; \vec{M}'}{
        \xrightarrow{\p Q[\vec{M}] = \inpc{a}{\vec{p} : N}.P_1 + A}}
  \;
  \new\vec{y}_k,\vec{c}.
    (\out{\Delta_k} \parallel \out{\vec{c}} \parallel
     P_1\subst{\vec{p} -> \vec{M}'} \parallel D_k )
  \label{posthat:trans1}
 \end{equation}
 with
  $\Delta_k, \vec{c} \deriv N \subst{\vec{p} -> \vec{M}'}
  $ and $
  \Delta_k, \vec{c} \deriv a$.

 Proof of Claim~I by induction on~$k$:
 For the base case in which $k = 0$, the claim holds by assumption.
 For the induction step, we assume that~\eqref{posthat:trans1} holds for~$k$ and we prove it for~$k+1$.
 Similar to the proof of \cref{lm:fold-sound-kembed},
 we use a multi-hole context $C[-, \cdots, -]$
 to distinguish between having budget $k$ or $k+1$.
 Let $F_1, \cdots, F_j$ be $j$ limits and
  $C[-, \cdots, -]$ a multi-hole context so that:
 \begin{align*}
   \fold_{k,L}^{b,n}(\lground{L}{n})
    & \congr
     C[\fold_{0,F_1^\omega}^{b,n}(F_1^\omega), \cdots,
       \fold_{0,F_j^\omega}^{b,n}(F_j^\omega)] =
     C[\lground{F_1^\omega}{n}, \cdots,
       \lground{F_j^\omega}{n}] \\ & =
     C[(\lground{F_1}{n})^n, \cdots,
       (\lground{F_j}{n})^n]
     = A
   \\
   \fold_{k+1,L}^{b,n}(\lground{L}{n})
    & \congr
     C[\fold_{1,F_1^\omega}^{b,n}(F_1^\omega), \cdots,
       \fold_{1,F_j^\omega}^{b,n}(F_j^\omega)] \\ & =
     C[(\lground{F_1}{n})^b, \cdots,
       (\lground{F_j}{n})^b]
     = B
 \end{align*}
 We want to show that it suffices to have $b$ copies of $F_l$ for any $1 \leq l \leq j$.
 Since $\p Q[\vec{M}]$ also occurs in $\stdf(\lground{L}{b})$, it is trivial to keep the process call which is reduced.
 It remains to argue that the same redex is enabled in B.
  Towards a contradiction:
Assume that there is a hole in which $b$ copies are not sufficient.
Wlog~let $L_l$ be the limit in this hole.

  We did not explicitly single out the message
 $N \subst{\vec{p} -> \vec{M}'}$
 but only the substitution
 $\subst{\vec{p} -> \vec{M}'}$
 for the continuation since
 the substitution is solely determining the successor.

We know that $A$'s knowledge is $\Delta_k$ and
factor out the knowledge from sublimit
$\lground{F_l}{m}$:
$\Delta_k = \Delta_k', \Delta_{l,m}$
  so that $\Delta_{l,m}$ was the knowledge obtained through
  $\lground{F_l}{m}$.
We want to prove that
${
  \Delta_k', \Delta_{l,m}, \vec{c} \deriv
  N \subst{\vec{p} -> \vec{M}'}
 \implies
  \Delta_k', \Delta_{l,b}, \vec{c} \deriv
  N \subst{\vec{p} -> \vec{M}'}
}$
where $\Delta_{l,b}$ denotes the knowledge obtained through
 $\lground{F_l}{b}$ respectively.

Now, we consider the different names used in
$N \subst{\vec{p} -> \vec{M}'}$.
It is straightforward to see that $\names(N) \subseteq \vec{x}_1$.
Recall the definition of $b$:
  \[
   b \is \beta(\Delta)  \cdot \gamma(\Intruder)^{s-1} + 1
\]
By this, we know that
$\beta(\Delta) \geq \card{\vec{p}}$ and hence
$b \geq \card{\vec{p}} \cdot \gamma(\Intruder)^{s-1} + 1$.
Intuitively, this ensures that there are enough distinct names for every single parameter in the parameter list as well as for channel $a$ as the size determines the maximum depth of the syntax tree of a message.
As $\names(\vec{p}) < b$, we can remove at least $m-b$ branches without loosing names used in $\vec{p}$ or channel name $a$.
Therefore, we can assume that
$\names(N\subst{\vec{p} -> \vec{M}}) \cup \set{a} \subseteq \vec{x}_1$.
The idea is to reduce the knowledge from $\Delta_{l,m}$ to
$\Delta_{l,b}$ by \cref{cor:deriv-variant-idemp}, which is a corollary of the absorbing intruder.
As $m$ might be bigger than $2b$, we have to iterate the process of reducing the number of branches.
Hence, we generalise the notation of $\Delta_{l,m}$ and $\Delta_{l,b}$ in the obvious way:
  $(\lground{F_l}{m})^i
   \congr
   \new\vec{a_i}.(\Delta_{l,i} \parallel \cdots)$.\\
\textbf{Claim~II:}
  For all~$m > b$,
  \[
    \bigl(
     \Delta_{l,m} \deriv N\subst{\vec{p} -> \vec{M}} \; \land \;
     \Delta_{l,m} \deriv a
     \bigr)
    \implies
     \bigl(
     \Delta_{l,m-1} \deriv N\subst{\vec{p} -> \vec{M}}
     \; \land \;
     \Delta_{l,m-1} \deriv a
     \bigr).
  \]

  Proof of Claim~II.
  For convenience, we rename $\Delta_{l,i}$ to $\Lambda_i$.
  The main observation is that we can split the knowledge $\Lambda_m$ into $\Lambda_{m-1}$ and a remainder $\Lambda'$.
  We can choose a branch which does not use names from $\vec{x}_1$ to contribute to $\Lambda'$.
  Since we know that $n \geq 1$, we know that $m > 1$ by assumption.
  Therefore, we can split $\Lambda_{m-1}$ again and obtain the knowledge stemming from one branch which we call $\Lambda''$.
  Let us recall the assumption and goal after these rewriting steps:
  Given that
  \begin{equation}
    \Delta_k', \Lambda_{m-2}, \Lambda', \Lambda''
    \deriv
    N\subst{\vec{p} -> \vec{M}}
    \quad \text{ and } \quad
    \Delta_k', \Lambda_{m-2}, \Lambda', \Lambda''
    \deriv
    a
    \label{eq:posthat-II-iii}
  \end{equation}
  holds, we want to prove that
  $\Delta_k', \Lambda_{m-2}, \Lambda'' \deriv
  N\subst{\vec{p} -> \vec{M}}$ and
  $\Delta_k', \Lambda_{m-2}, \Lambda'' \deriv a$.
  Let $\vec{w}'$ and $\vec{w}''$ be the names only used in $\Lambda'$ and $\Lambda''$ respectively so that:
  \begin{equation}
    \vec{w}' \inters \vec{w}'' = \emptyset
  \text{ and }
   \names(\Gamma_{c, m}', \Lambda_{m-2}) \inters \vec{w}' = \emptyset =
   \names(\Gamma_{c, m}', \Lambda_{m-2}) \inters \vec{w}''
   \label{eq:posthat-II-iv}
  \end{equation}
  $\Lambda'$ and $\Lambda''$ have been obtained from a branch of the same sublimit, so we can infer that
  \begin{equation}
    \Lambda'' = \Lambda'\subst{\vec{w}' -> \vec{w}''}
    \label{eq:posthat-II-v}.
  \end{equation}
  Notice that $\vec{w}' \inters \vec{x}_1 = \emptyset$ by the fact how we have chosen the the branch for $\Lambda'$.
  Because of
  $\names(N\subst{\vec{p} -> \vec{M}}) \cup \set{a} \subseteq \vec{x_1} $,
  we can infer that
  $(\names(N\subst{\vec{p} -> \vec{M}}) \cup \set{a}) \inters \vec{w'} = \emptyset$ which implies
  that
  \begin{equation}
    \names( N\subst{\vec{p} -> \vec{M}}) \cup \set{a} \subseteq \names(\Delta_k', \Lambda_{m-2}, \Lambda'')
    \label{eq:posthat-II-vi}
  \end{equation}
  Facts \eqref{eq:posthat-II-iii}~to~\eqref{eq:posthat-II-vi}
  fulfil the conditions for
  \cref{cor:deriv-variant-idemp} resulting in
  ${\Delta_k', \Lambda_{m-2}, \Lambda''
    \deriv N\subst{\vec{p} -> \vec{M}}}$
  and
  ${\Delta_k', \Lambda_{m-2}, \Lambda''\deriv a}$
  which reads
  $\Delta_k', \Lambda_{m-1}
  \deriv  N\subst{\vec{p} -> \vec{M}}$
  and
  ${\Delta_k', \Lambda_{m-1} \deriv a}$
  respectively
  when folding back which is the goal of Claim~II.

  In turn, this concludes the proof of Claim~I.

  Instantiating the statement of Claim~I with
  $k > \height(L)$ shows that this transition is still enabled in
  $\lground{L}{b}$.
    We prove that the same holds for transitions enabled by the remaining reduction rules.

  \subparagraph{Public private channel, message output}
  Recall that
  $\stdf(\lground{L}{m})
   \aeq
  \new\vec{x}_1, \vec{x_2}.(\out{\Gamma_1} \parallel
                            \out{\Gamma_2}\parallel
                 \p Q[\vec{M}] \parallel C_2)$.
   We have
 $\p Q[\vec{M}] := \outpc{a}{M}.
                 P_1 + A $ for some private channel $a$ and action $A$.
   The channel name $a$ can be derived by assumption: $\Gamma_1 \deriv a$.
 Overall, we get the following transition:
 \[
  \lground{L}{m}
  \;
\xrightarrow{\p Q[\vec{M}] =
                     \outpc{a}{M}.P_1 + A}
\;
  \underbrace{\new \vec{x}_1,\vec{x}_2, \vec{c}.
    (\out{\Gamma_1} \parallel \out{\Gamma_2}
    \parallel \out{M} \parallel
     P_1 \parallel C_1)
  }_{Q'\mathrlap{{} \in \post^s(L)}}
 \]

  and the channel name~$a$ can be derived by assumption:
  $\Gamma_1, \Gamma_2 \deriv a$.
 We want to show that we can have the same reduction in
 $\lground{L}{b}$.
 We sketch the proof that this transition is still enabled as it is analogous to the first proof.
 We start with an analogous claim.

 Claim~I: In every hole of context $C[-, \cdots, -]$,
 we can reduce the number of branches to at most~$n$, i.e.~for every~$j$ we can have a grounding
 $\lground{F_j}{n}$ instead of $\lground{F_j}{m}$.
  for every $k \in \Nat$:
 \begin{align*}
  \exists \vec{y}_k, \Delta_k, D_k \st
  \fold_{k,L}^{b,n}(\lground{L}{n})
   \quad
\xrightarrow{\p Q[\vec{M}] = \outpc{a}{M} + A}
\quad
  & \new\vec{y}_k.
    (\out{\Delta_k} \parallel
     P_1 \parallel \out{M} \parallel D_k )
  \\
  & \text{ with }
  \Delta_k \deriv a
\end{align*}
 It is obvious to see that the side condition is subsumed by \eqref{posthat:trans1} which we have used in the first case and hence the proof could be proceeded in the same way.

 \subparagraph{Private channel, two participants}
 For the last case, we consider:
 \begin{align*}
 \stdf(\lground{L}{m})
   & \aeq
  \new\vec{x}_1, \vec{x_2}.(\out{\Gamma_1} \parallel
                            \out{\Gamma_2}\parallel
                            \p Q_1[\vec{M}_1] \parallel
                            \p Q_2[\vec{M}_2] \parallel
                            C_2)
                            \text{ with }\\
 \p Q_1[\vec{M}_1] & \defeq \outpc{c}{N\subst{\vec{x}->\vec{M}'}}.P_1 + A_1 \\
    \p Q_2[\vec{M}_2] & \defeq \inpc{c}{\vec{x}:N}.P_2 + A_2
 \end{align*}
 The resulting transition looks as follows:
  \begin{multline*}
    \new \vec{x_1}.\new \vec{x_2}(
      \out{\Gamma_1}  \parallel
      \out{\Gamma_2}  \parallel
      \p Q_1[\vec{M}_1] \parallel
      \p Q_2[\vec{M}_2] \parallel
      C_2
    )
  \\
  \redto_\Defs
  \new \vec{x_1}.\vec{x_2}(
      \out{\Gamma_1}  \parallel
      \out{\Gamma_2}  \parallel
      P_1                        \parallel
      P_2\subst{\vec{x}->\vec{M}'} \parallel
      C_2
    )
  \end{multline*}
 There are no side conditions on derivability in this case so we only have to ensure that all possible combiniations of
 $\outpc{c}{N\subst{\vec{x}->\vec{M}'}}.P_1$
 and
 $\inpc{c}{\vec{x}:N}.P_2$ can occur.
 As both actions could stem from the same process call, it suffices to have an expansion factor greater or equal to $2$.
 This only poses restrictions on $\beta$ and $\gamma$.
 By definition of $b$, we need to ensure that
 $\beta(\Delta) \cdot \gamma(\Intruder)^{s-1} \geq 1$.
 As argued before, considering $s < 2$ is unreasonable since it would be impossible to have encryptions in this setting.
 Hence, even if $\gamma(\Intruder)$ was $0$,
 $\gamma(\Intruder)^{s-1} \geq 1$ for reasonable $s$.
  If $\beta(\Delta) < 1$, it is $0$ and hence the transition we consider would be impossible.
 We therefore can assume that $\beta(\Delta) > 0$.
 Hence, we know that $b \geq 2$ and this suffices to have a congruent transition in $\lground{L}{b}$.
 \bigskip

 We now turn back to our general goal and know that all kinds of transitions are enabled in $\lground{L}{b}$.
 Consider the extension $\lext{L}{b}$ of limit $L$ whose standard form we choose to resemble the correlation with $\lground{L}{b}$:
  $\stdf(\lext{L}{b}) \aeq
   \new\vec{x}_1.(\Gamma_1 \parallel \p Q[\vec{M}] \parallel C_1
                  \parallel R_1)$.
  By the definition of $\posthat(L)$, we know that
  \[
    L' := \new\vec{x}_1.(
          \out{\Gamma_1} \parallel
          P_1 \subst{\vec{p} -> \vec{M}}
          \parallel C_1
          \parallel R_1)
      \in \posthat(L).
  \]
  Recall that
  $Q' = \new \vec{x}_1,\vec{x}_2, \vec{c}.
    (\out{\Gamma_1} \parallel \out{\Gamma_2}
    \parallel \out{\vec{c}} \parallel
     P_1\subst{\vec{p} -> \vec{M}'} \parallel C_1)$
  was the successor of $\lground{L}{m}$ from the beginning.
  It remains to show that
  $Q' \in \sem{L'}$.
  As before for $\lext{L}{b}$, we relate the standard form of
  $\lext{L}{m}$ to $\lground{L}{m}$ and get the following:
  \[
   \stdf(\lext{L}{m}) \aeq
   \new\vec{x}_1, \vec{x_2}.(\out{\Gamma_1} \parallel
                            \out{\Gamma_2}\parallel
                 \p Q[\vec{M}] \parallel C_2 \parallel R_2)
  \]
  By \cref{lm:ext-same-sem}, we have that
  $\sem{\lext{L}{m}} = \sem{\lext{L}{b}}$ and taking one step hence leads to
  \[
  \begin{array}{rccl}
   K'
   & := &
   \sem{\new \vec{x}_1,\vec{x}_2, \vec{c}.
    (\out{\Gamma_1} \parallel \out{\Gamma_2}
    \parallel \out{\vec{c}} \parallel
     P_1\subst{\vec{p} -> \vec{M}'} \parallel C_2 \parallel R_2} \\
   & = &
   \sem{\new\vec{x}_1.(\out{\Gamma_1} \parallel \out{\vec{c}}
                        \parallel P_1\subst{\vec{p} -> \vec{M}'} \parallel C_1 \parallel R_1)}
   & =: L'
  \end{array}
  \]
  As $Q' \in K'$, it also holds that $Q' \in L'$ which concludes this proof.

\end{proof}

\section{Support for Other Cryptographic Primitives}
\label{app:more-primitives}

We claim the assumptions we make on the intruder model are mild,
and are satisfied by the symbolic models of many cryptographic primitives.
We illustrated the treatment of (a)symmetric encryption;
treatment of hashes and blind signatures is entirely analogous.
By using the sequent calculus formalisation of \cite{tiu10lmcs}
one can trivially extend our proofs to prove all these primitives
form an effective absorbing intruder.

Supporting XOR requires a bit more analysis on the algebraic properties of the primitives.
We take as reference the model of XOR analysed in~\cite{AbadiCortierXOR,ChevalierKRT03}.
The two constructors $\xor$ (of arity~2) and $\xorzero$ (of arity~0)
are added to the set of constructors.
Their algebraic properties are formalised through a congruence relation:
\begin{align}
  M_1 \xor (M_2 \xor M_3) &\xoreq (M_1 \xor M_2) \xor M_3 &
  M_1 \xor M_2 &\xoreq M_2 \xor M_1
  \label{law:xor-AC}
  \\
  M \xor \xorzero &\xoreq M &
  M \xor M &\xoreq \xorzero
  \label{law:xor-directed}
\end{align}

The results of~\cite{AbadiCortierXOR,ChevalierKRT03} establish that
the laws~\eqref{law:xor-directed} can be always orientated from left to right.
Formally, one can define the rewriting system $ \xorrew $ with two rules
$ M \xor \xorzero \xorrew M$ and $M \xor M \xorrew \xorzero $;
and the congruence $\aceq$ defined by laws \eqref{law:xor-AC}.
Then the relation ${\xorrewac} \is ({\aceq\circ\xorrew})$
is terminating, and confluent modulo~$\aceq$.
The set of normal forms $ \xornf{M} \is \set{N | M \xorrewac^* N \not\xorrewac} $ is then guaranteed to be finite, computable, and such that
$M \xoreq N \iff (\xornf{M} \inters \xornf{N}) \neq \emptyset$.

We can harmonise the equational theory of XOR with the deduction system of
\cref{fig:symmetric-intruder} by adding the rules
\begin{mathpar}
  \infer*[right=Xor$_{\xorzero}$]{ }{
  \label{rule:xorzero}
    \Gamma \deriv \xorzero
  }
  \and
  \infer*[right=Xor\textsubscript{R}]{
  \label{rule:xorR}
    \Gamma \deriv M_1 \\
    \Gamma \deriv M_2
  }{
    \Gamma \deriv M_1 \xor M_2
  }
  \and
  \infer*[right=Xor\textsubscript{L}]{
  \label{rule:xorL}
    \Gamma, M_1, M_2, M_1 \xor M_2 \deriv N
  }{
    \Gamma, M_1, M_2 \deriv N
  }
  \and
  \infer*[right=$\xornf{}$\textsubscript{L}]{
  \label{rule:xoreqL}
    \Gamma, M, \xornf{M} \deriv N
  }{
    \Gamma, M \deriv N
  }
  \and
  \infer*[right=$\xornf{}$\textsubscript{R}]{
  \label{rule:xoreqR}
    \Gamma \deriv N \\
    N \in \xornf{M}
  }{
    \Gamma \deriv M
  }
\end{mathpar}
The deduction system accurately models XOR,
even if it uses $\xornf{}$ instead of $\xoreq$:
as proven in~\cite[Prop.~1]{ChevalierKRT03},
one can always restrict the intruder
to manipulate messages in normal form without loosing expressive power.

We are left to prove that the derivability satisfies the effective absorbing intruder axioms.
Decidability has been proven in~\cite{AbadiCortierXOR,ChevalierKRT03}.
The axioms of \cref{def:intruder-model,def:absorption} are easily satisfied by the same arguments we used for $\IntrSymm$.
The proofs of \eqref{axiom:relevancy} and absorption make use of the fact
that if $N \in \xornf{M}$ then $\freenames(N) \subseteq \freenames(M)$.

\medskip

We conjecture Diffie-Hellman exponentiation
(following the model of e.g.~\cite{SchmidtMCB12})
can be shown to satisfy our axioms
in the same way we treated XOR.
The main issue with Diffie-Hellman, with respect to \eqref{axiom:relevancy} and absorption, is the inverses law $ M * M^{-1} \cong 1 $:
by using the law from right to left,
one can involve arbitrary names in a derivation.
This could be handled in the same way we handle cancellation of XOR,
by normalising derivations so that the law is always applied left to right.
We leave the formal development of this remark as future work.

\medskip

A delicate point is the bounded message size assumption.
With advanced primitives like XOR it is easier (but not inevitable)
to encounter protocols
for which it is impossible to extend the results on the bounded model
to the unbounded case.

\section{Towards Unbounded Message Size}
\label{app:extended-limits}

The analysis presented in \cref{sec:ideal} considers only traces (and attacks)
involving messages of size smaller than some given bound~$s$.
We show here how the results of the analysis can be generalised
to inductive invariants for the full set of traces, with no restriction on the size.
Because of undecidability of the general problem,
this generalisation will not be precise for every protocol:
there are protocols for which the most precise generalised limit is trivial (i.e.~does not ensure any non-trivial property of the protocol).
For the protocols in our benchmarks however, one can get precise invariants
by generalising the ones inferred by the tool.

We first introduce the syntax and semantics of generalised limits,
and describe how we can use them to generalise our benchmarks.
Finally,
although studying how to automate this generalisation is beyond the scope
of this paper, we briefly sketch how $\posthat$ can be adapted to
work on generalised limits.

\subparagraph{Aside: finer size bounds via typing}
In our development, we assumed a global size bound~$s$.
To have finer control on the message sizes, as we do in our tool,
one can introduce a primitive form of typing.
Assuming wlog that all pattern variables are unique,
a \emph{typing} is a partial function $ \ty \from \Names \pto \Nat $,
assigning to each pattern variable a maximum size for the messages it can match.
A typing induces a typed transition relation $\redto_{\Defs,\ty}$
which only matches patterns with subsitutions respecting $\ty$;
$\reach^\ty_\Defs(P)$ collects all the terms reachable from $P$ through $\redto_{\Defs,\ty}$.
A typing $\ty$ of $P,\Defs$ is \pre s-bounding if
$\reach^\ty_\Defs(P) \subseteq \Size{s}$.
One can check if $\ty$ is \pre s-bounding on-the-fly while computing $\posthat$.

\subsection{Generalised Limits}
Recall the ``encryption oracle'' of \cref{ex:encryption-oracle}:
$
 E[k] = \inp{x : x}.(\out{ \enc{x}{k} } \parallel E[k] )
$.
Without any restrictions on the pattern variable~$x$,
there is no way to prevent the encryption chains
described in \cref{ex:encryption-oracle}.
However, we can use the \pre 2-bounding typing $ \ty = \map{x->1} $,
to obtain the limit
$
  \p E[k] \parallel \bigl( \new m. \out{\enc{m}{k}}\bigr)^\omega
$
which is inductive (wrt~$\redto_{\Defs,\ty}$)
and contains the initial state $\p E[k]$.
Since~$x$ is never inspected, injecting larger messages in it would not lead to new behaviour.
We represent this arbitrary injection of messages in a limit by introducing a ``wildcard'' annotation on occurrences of names $a^\wildcard$.

Technically, we duplicate the set of names $\Names$ to a disjoint set
of \pre \wildcard-annotated names
  $ \Names^\wildcard \is \set{ a^\wildcard | a \in \Names } $,
and we allow messages to contain names from $\Names \dunion \Names^\wildcard$.
All the definitions of this paper can be adapted straightforwardly to support
wildcards by simply not distinguishing between $a$ and $a^\wildcard$.
To stress the fact that a set of processes/limits contains annotations,
we annotate the set with a wildcard, e.g.~$\Limits_{s,k}^{\wildcard}$.

\begin{definition}[Wildcard semantics]
Given $P \in \Proc^\wildcard$, its \emph{wildcard semantics} is defined as
\[
    \wsem{P} \is
      \Set{ P' |
        \begin{array}{l}
          \stdf(P) = \new \vec{x}.( \out{\Gamma} \parallel Q ) ,\;
          P' \kcongr
            \new \vec{x}, \vec{c}. \bigl(
              ( \out{\Gamma} \parallel Q )
              \subst{\vec{y}\,^\wildcard -> \vec{M}}
            \bigr)
        \\
          \vec{y}\,^\wildcard = \names(P) \inters \Names^\wildcard,
          \names(\vec{M}) \subseteq \vec{x} \union \vec{c}
        \end{array}
      }
\]
For $L \in \Limits_{s,k}^{\wildcard}$, define
$
    \wsem{L} \is \Set{ \wsem{P} | P \in \sem{L}}
$.
\end{definition}

With the help of wildcards, we can then take a limit $L$,
annotate it with wildcards obtaining a limit $L^\wildcard$.
Then we can check that the wildcards generalise the limit enough
to make it inductive wrt $\redto_{\Defs}$
(i.e.~without restricting the message size):
\begin{equation}
  \forall P \in \wsem{L^\wildcard},
    \forall Q\st
      P \redto_{\Defs} Q  \implies  Q \in \wsem{L^\wildcard}
  \label{eq:wildcard-inductive}
\end{equation}
For our encryption oracle example, the annotated limit
$
  \p E[k] \parallel \bigl( \new m. \out{\enc{m^\wildcard}{k}}\bigr)^\omega
$
represents an inductive invariant for the unrestricted semantics.

For a more realistic application of generalised limits,
consider our running \cref{ex:running}.
The limit of~\cref{fig:running-invariant} is inductive with respect to
the typing $ \ty = \map{n_a->1} $.
To remove the size bound we only need to annotate the
occurrence of $n_a$ in $L_2$
obtaining the limit $L^\wildcard$:
\begin{align*}
L^\wildcard &= \new a, b, k_{as}, k_{bs}.(
      \out{a,b} \parallel
      \p{A_{1}}[a, b, k_{as}]^\omega \parallel
      \p{B_{1}}[a, b, k_{bs}]^\omega \parallel
      \p{S}[a, b, k_{as}, k_{bs}]^\omega \parallel
      L_1^\omega
    )
\\
L_1 &= \new n_{a}.\bigl( \out{n_{a}} \parallel
    \p{A_{2}}[a, b, k_{as}, n_{a}] \parallel
    L_2^\omega
  \bigr)
\\
L_2 &= \new k.\bigl(
\out{\enc{k}{(n_{a}^{{\color{red}\wildcard}}, k_{as})}} \parallel
    \out{\enc{k}{k_{bs}}} \parallel
    \p{Secret}[k]^\omega \parallel
    \p{A_{3}}[a, b, k_{as}, k]^\omega \parallel
    L_3^\omega
  \bigr)
\\
L_3 &= \new n_{b}.\bigl(
    \out{\enc{n_{b}}{(k, k)}} \parallel
    \out{\enc{n_{b}}{k}} \parallel
    \p{B_{2}}[a, b, k_{bs}, n_{b}, k]
  \bigr)
\end{align*}

Indeed this generalised limit is inductive:
the intruder can send any message $(M, b)$ to the server $\p{S}[a, b, k_{as}, k_{bs}] \is \inp{x : (x, b)}.\bigl(\dots
      \out*{\enc{k}{(x, k_{as})}} \dots
\bigr)$, which will use it as part of the encryption key $(M,k_{as})$.
None of the input patterns of the other processes would however be able to match the message $\enc{k}{(M, k_{as})}$ unless $M=n_{a}$;
thus this attack attempt will not generate new behaviour,
and our generalised limit captures all the reachable configurations
without assuming bounds on the size of messages.

Similarly, it is not difficult to annotate the limits of our benchmarks so that the generalised limits satisfy condition~\eqref{eq:wildcard-inductive}:

\begin{itemize}
 \item ARPC:
    There is only one annotation reasonable:
    $ \pair{n_x^\wildcard}{\pair{k}{b}} $.
    This is still inductive even though it can be fed back
    as we can have a different substitutions for $n_x^\wildcard$.
\item KSL:
    We annotate
    $\p B_2[a, b, k_{bb}, k_{bs}, n_x^\wildcard, n_y]$ in $L_2$
        and
        $\enc{\pair{nz}{\pair{b, n_x^\wildcard}}}{k_{xy}}$ in $L_4$.
    $A$ knows the actual $n_x$ as it produced it
    and hence it is still inductive.
 \item KSLr:
    Additionally to KSL, we annotate
    $ \enc{\pair{m_x^\wildcard}{m_y}}{k_{xy}} $
    and
    $ \p B_5[a, b, k_{bb}, k_{bs}, k_{ab}, t_y, m_x^\wildcard, m_y] $ in $L_7$.
    The new message cannot flow into $\p B_3[-]$ or $\p D_2[-]$ as both pattern-match names on the second position which are different from $m_y$.
    However, due to the wildcard as parameter which could become $m_y$, $\p B_4[-]$ can produce $\p B_5[-]$'s input message.
    This is still covered and inductive but does not comply with an normal reauthentification run.
 \item NHS:
    We annotate
    $ \enc{
        \pair{n^\wildcard}{\pair{b}{\enc{\pair{k}{a}}{k_{bs}}}}
        } {k_{as}} $
    in $ L_2 $.
    $A$ can only match on the original $n$ so it is still inductive.
 \item NHSr:
    Same annotations as in NHS.
    Additionally,
    we annotate
    $ \enc{\pair{one}{s^\wildcard}}{k} $.
    in $L_3$
    and hence have to annotate the $s$ in
    $\p {Secret}[s^\wildcard] $ as well as $\p {Leak}[s^\wildcard] $.
    This covers all the enabled transitions and is hence inductive.
 \item NHSs:
    The same annotation as in NHS works.
 \item OR/ORl:
    We annotate
    $ \enc{\pair{n_y}{\pair{m^\wildcard}{\pair{a}{b}}}}{k_{bs}} $
    in $ L_2 $.
    This does not enable any new transition so it is inductive.
 \item ORs:
    Same annotation as OR/ORl.
    But we also have to annotate the key $k_{xy}$ which we declare as secret which is the difference to ORl and hence does not work here.
 \item ORa:
    Similar annotation to OR/ORl.
    Analogously, inductivity is preserved.
 \item YAH:
    We annotate
    $ \enc{\pair{a}{\pair{n_a^\wildcard}{n_b}}}{k_{bs}} $
    in $ L_3 $.
    This can be fed back to $\p B_2[-]$ and hence we also annotate
    $ \enc{n_b}{\pair{n_a^\wildcard}{n_b}} $
    in $ L_5 $.
    This is still inductive.
 \item YAHs1:
    We annotate as in YAH.
    The declaration of a secret happens only if $A$ plays back and hence it is inductive.
 \item YAHs2:
    Without the size constraint for the key,
    the invariant obtained for YAHs1 is also one for YAHs2.
 \item YAHlk:
    We annotate
    $ \enc{\pair{one}{s^\wildcard}}{k} $
    in $L_4$
    which represents the case where the key $k$ is leaked.
    This annotation covers all newly enabled transitions.
    In case $k$ is not leaked,
    there are no new transitions.
    Hence, it is inductive.
\end{itemize}

\subsection{Towards Automating Generalisation}

Automating the check of condition~\eqref{eq:wildcard-inductive} is beyond
the scope of this paper, but we can sketch how one would approach the problem
by adapting our $\posthat$ definition to work on generalised limits.
The idea is that one can design a computable function $\posthat^{\ty}_{\Defs,\wildcard}$, a ``symbolic'' version of $\posthat$, and $\Subset$ satisfying:
\begin{gather}
  L_1^\wildcard \Subset L_2^\wildcard
    \implies
      \wsem{L_1^\wildcard} \subseteq \wsem{L_2^\wildcard}
  \label{eq:symb-inclusion}
  \\
  \post(\wsem{L^\wildcard}) \subseteq
    \wsem{\posthat^{\ty}_{\Defs,\wildcard}(L^\wildcard)}
  \label{eq:symb-posthat}
\end{gather}
With these components, one can check~\eqref{eq:wildcard-inductive} by checking
$
  \posthat^{\ty}_{\Defs,\wildcard}(L^\wildcard)
    \Subset L^\wildcard
$.
Note that
  \eqref{eq:symb-inclusion} is an implication,
and \eqref{eq:symb-posthat} a subset relation:
an equivalence would be impossible to achieve due to undecidability of the general problem; we therefore only require a sound over-approximation.

Designing $\Subset$ requires defining an approximate version of $\kleq$
which can work on \pre\wildcard-annotated sets of messages.
A precise version of this can only be defined on a per-intruder-model basis.
A generic definition of $\Subset$ could simply extend the non-annotated inclusion check with the axioms
$x^\wildcard \kleq x^\wildcard$ and $M \kleq x^\wildcard$.

Designing a suitable $\posthat^{\ty}_{\Defs,\wildcard}(L^\wildcard)$
similarly depends on the choice of intruder-model.
One could, for example, define a symbolic matching function
$\operatorname{match}(\Gamma^\wildcard, \vec{x} : N^\wildcard)$
returning a finite set of substitutions such that
\[
  \Gamma \deriv N^\wildcard \theta
  \land
  \Gamma \in \wsem{\Gamma^\wildcard}
    \implies
      \exists \sigma\in
        \operatorname{match}(\Gamma^\wildcard, \vec{x} : N^\wildcard) \st
          N^\wildcard \theta \in \wsem{N^\wildcard \sigma}
\]
and use it in $\posthat^{\ty}_{\Defs,\wildcard}$ to find all symbolic redexes.
It is relatively straightforward to define a match function that is precise enough to check the generalised limits of our benchmarks.
Exploring the design of these symbolic analyses in general is left for future work.

\section{Benchmarks}

The Needham-Schr\"{o}der protocol
\cite{NeedhamS78}
is modelled with and without secrecy (NHS/NHSs).
The NHSr version models leaks of old session keys, which leads to a replay attack (and hence the invariant is leaky).
We provide four models of the Otway-Rees protocol
\cite{OtwayR87}.
OR does not model secrecy and is used to prove the protocol depth-bounded.
ORl models secrecy but the inferred invariant contains a genuine leak, which is the result of a known type-confusion attack.
The attack substitutes a composite message for some input $x$ that is (wrongly) assumed to be a nonce by a principal.
ORa models authentication and the invariant shows the genuine misauthentication based on this attack.
ORs models the same situation with the assumption that $x$ is of size one; with this assumption the inferred invariant is not leaky.

ARPC models Lowe’s modified BAN concrete ARPC protocol
\cite{Lowe96}
(we model succ(-) with pairs,
i.e.~$\operatorname{succ}(\pair{\operatorname{zero}}{-})$ is
$\pair{\operatorname{one}}{-}$).
We modelled the Kehne-Sch\"{o}nw\"{a}lder-Landend\"{o}rfer protocol~\cite{Lowe96}, as modified by Lowe, with (KSLr) and without (KSL) re-authentication.

We produced four models of the Yahalom protocol
\cite{BurrowsAN89}.
Our first model (YAH) does not model secrecy and is used to establish depth-boundedness.
The protocol has a type-confusion attack (similar to the case of Otway-Rees) which does not lead to a leak of a secret.
This is modelled in YAHs1.
We rule out this type-confusion by adding a size assumption in YAHs2.
YAHlk is a variant where an additional fresh nonce is exchanged at the beginning of each session, which makes the protocol secret even when old session keys can be leaked, a property entailed by the inferred invariant.

\section{Example of Incorporation}
\label{app:incorporation}

Consider the limit
$
  L =
    \new y.\Bigl(
      \bigl(\new x.(\p A[x] \parallel \p B[y,x]^\omega)\bigr)^\omega
    \Bigr).
$
To compute $\posthat$ we consider the limit
$
  \lext{L}{1} =
    \new y.\Bigl(
      \bigl(\new x.(
          \p A[x] \parallel \p B[y,x] \parallel \p B[y,x]^\omega)
      \bigr)
    \parallel
      \bigl(\new x.(\p A[x] \parallel \p B[y,x]^\omega)\bigr)^\omega
    \Bigr)
$
and assume
\[
  \p A[x] \redto \new z.(A[z] \parallel \p B[y,z]) \parallel \p B[y, x] = P.
\]
Then, $\posthat$ will contain the limit
$ L' = C[P] $
for
$
  C[\bullet] =
    \new y.\Bigl(
      \bigl(\new x.(
          \;\bullet \parallel \p B[y,x] \parallel \p B[y,x]^\omega)
      \bigr)
    \parallel
      \bigl(\new x.(\p A[x] \parallel \p B[y,x]^\omega)\bigr)^\omega
    \Bigr)
$.
Since we are trying to prove inclusion between $L$ and $ L' = C[P] $,
and $L$ is equivalent to $\lext{L}{1} = C[\p A[x]]$,
it seems likely that the inclusion could be proven by
matching $C$ with $C$ in the two limits, and by matching $P$ in $\lext{L}{1}$.
Intuitively:
\begin{center}
  \begin{tikzpicture}[
  group/.style={decoration=brace,decorate,thick},
  match/.style={->,semithick, shorten=5pt, looseness=#1},
  match/.default={.8},
]
  \begin{scope}[
      start chain=going right,
      node distance=-2pt,
      every node/.style={on chain,inner sep=2pt,rounded corners,strut sized}
    ]

    \node(C1){$
    \new y.\Bigl(
           \bigl(\new x.(
    $};
    \node[red](A){$
               \p A[x]
    $};
    \node(C5){$
      {} \parallel \p B[y,x] \parallel {}
    $};

    \node(B1){$\p B[y,x]^\omega$};

    \node(C2){$
      )\bigr) \parallel \bigl(
    $};
    \node(xAB){$
      \new x.(
        \p A[x]
        \parallel
        \p B[y,x]^\omega
      )
    $};
    \node(C4){$
      \bigr)^\omega\Bigr)
    $};
  \end{scope}
  \begin{scope}[
    yshift=-1.5cm,
    start chain=going right,
    node distance=-2pt,
    every node/.style={on chain,inner sep=2pt,rounded corners,strut sized},
    ]
    \node(C1'){$
    \new y.\Bigl(
           \bigl(\new x.(
    $};
    \node[blue](P-zAB){$ \new z.( \p A[z] \parallel \p B[y,z] ) $};
    \node[blue](P-par){$ {} \parallel {} $};
    \node[blue](P-B){$ \p B[y,x] $};
    \node(C5'){$
      {} \parallel \p B[y,x] \parallel {}
    $};

    \node(B1'){$\p B[y,x]^\omega$};

    \node(C2'){$
      )\bigr) \parallel \bigl(
      \new x.(
      \p A[x]
      \parallel
      \p B[y,x]^\omega
      )\bigr)^\omega\Bigr)
    $};
  \end{scope}
  \draw[group] (P-zAB.north west) -- (P-zAB.north east);
  \draw[group] (xAB.south east) -- (xAB.south west);
  \draw[group] (B1.south east) -- (B1.south west);
  \draw[group] (P-B.north west) -- (P-B.north east);
  \draw[match=.5] (P-zAB) edge[in=-90,out=90] (xAB);
  \draw[match=1] (P-B) edge[in=-80,out=90,preaction={draw=white,line width=3pt,-}] (B1);
\end{tikzpicture} \end{center}
where the part in blue is $P$ and the context $C$ is the part in black.
The arrows indicate the sublimits that
can be further unfolded to show the inclusion of $C[P]$.

Formally, we want to check that
$
  \forall L' \in \posthat^s_\Defs(L)\st
    \sem{L'} \subseteq \sem{L}
$.
By construction,
  $L' = C[P]$ for some context $C$ such that
  $ \lext{L}{b} = C[\p Q[\vec{M}]] $.
Then,
if we can find $P$ in the fixed part of $\stdf(\lext{L}{n})$,
conditions~\cond{A} and~\cond{B} of \cref{th:char-lim-incl} are
automatically satisfied because of the shared context $C$.
\end{document}